\documentclass[12pt]{article}

\pdfoutput=1

\usepackage{latexsym,amsfonts,bm,epsfig,amsmath,natbib,authblk,amsthm,thmtools,amssymb}

\usepackage{float}
\usepackage{booktabs}
\usepackage{graphicx}
\usepackage{subcaption}
\usepackage[width=.95\textwidth]{caption}
\usepackage{color}
\usepackage{xr-hyper}
\usepackage[colorlinks,linkcolor=black,citecolor=black,filecolor=black,bookmarks=false,pagebackref]{hyperref}
\floatstyle{plain}
\newfloat{Algorithm}{thp}{lop}
\floatname{Algorithm}{Algorithm}

\usepackage[plain,noend]{algorithm2e}

\makeatletter
\renewcommand{\algocf@captiontext}[2]{#1\algocf@typo. \AlCapFnt{}#2} 
\def\@algocf@capt@plain{top}
\renewcommand{\algocf@makecaption}[2]{%
	\addtolength{\hsize}{\algomargin}%
	\sbox\@tempboxa{\algocf@captiontext{#1}{#2}}%
	\ifdim\wd\@tempboxa >\hsize
	\hskip .5\algomargin%
	\parbox[t]{\hsize}{\algocf@captiontext{#1}{#2}}
	\else%
	\global\@minipagefalse%
	\hbox to\hsize{\box\@tempboxa}
	\fi%
	\addtolength{\hsize}{-\algomargin}%
}
\makeatother

\def\tr{\text{\rm tr}}
\def\T{{\mathrm{\scriptscriptstyle T} }}

\usepackage{pgfplots}
\usepackage{tikz}
\newlength\figureheight
\newlength\figurewidth
\usepackage{apptools}

\floatstyle{plain}

\newcommand{\argmin}{\operatornamewithlimits{argmin\,}}

\newcommand{\E}{\mathbb{E}}

\newcommand{\1}{\mathbb{I}}
\newcommand{\cut}{\mathrm{cut}}
\newcommand{\full}{\mathrm{full}}

\mathchardef\mhyphen="2D
\def \z{z_{1:n}}

\def \E{\mathbb{E}}

\def \dt {\mathrm{d}}

\newcommand{\I}{\mathcal{I}}
\newcommand{\M}{\mathcal{M}}

\newtheorem{assumption}{Assumption}
\newtheorem{remark}{Remark}
\newtheorem{theorem}{Theorem}
\newtheorem{corollary}{Corollary}
\newtheorem{lemma}{Lemma}

\begin{document}
	\def\spacingset#1{\renewcommand{\baselinestretch}%
		{#1}\small\normalsize} \spacingset{1}
	
	\title{Posterior risk of modular and semi-modular Bayesian inference}
	\date{\empty}
	\author[1]{David T. Frazier\thanks{Corresponding author:  david.frazier@monash.edu}}
	\author[2,3]{David J. Nott}
	\affil[1]{Department of Econometrics and Business Statistics, Monash University, Clayton VIC 3800, Australia}
	\affil[2]{Department of Statistics and Applied Probability, National University of Singapore, Singapore 117546}
	\affil[3]{Operations Research and Analytics Cluster, National University of Singapore, Singapore 119077}
	
	\maketitle
	\vspace{-2cm}
	
	\begin{abstract}
		Modular Bayesian methods perform inference in models
		that are specified through a collection of coupled sub-models, known as modules.  
		These modules often arise from modeling
		different data sources or from combining domain knowledge from
		different disciplines.   {\color{black}``Cutting feedback'' is a Bayesian inference method that ensures 
			misspecification of one module does not affect inferences for parameters in other modules, and produces what is known as the cut posterior.}
		However, choosing between the cut posterior and the standard Bayesian
		posterior is challenging. When misspecification
		is not severe, cutting feedback 
		can greatly increase posterior uncertainty without a large reduction of estimation bias, 
		leading to a bias-variance trade-off. This trade-off
		motivates semi-modular posteriors, which interpolate between standard and cut posteriors based on a tuning parameter.  In this work,   we provide the first precise formulation of the bias-variance
		trade-off that is present in {\color{black}cutting feedback, and  we propose a new semi-modular posterior that takes advantage of it. Under general regularity conditions, we prove that this semi-modular posterior is more accurate than the cut posterior according to a notion of posterior risk. An important implication of this result is that point inferences made under the cut posterior are \textit{inadmissable}. The new method is demonstrated in a number of examples.
			\vspace{.25cm}
		}
		
	\end{abstract}
	\spacingset{1.9} 
	
	
	\section{Introduction}
	
	In many applications, statistical models arise which can be viewed as a combination
	of coupled submodels (referred to as modules in the literature).  
	Such models are often complex, frequently containing 
	both shared and module-specific parameters, and module-specific
	data sources.  Examples include pharmacokinetic/phamacodynamic
	(PK/PD) models \citep{bennett+w01,lunn+bsgn09} which couple a PK module describing movement
	of a drug through the body with a PD module describing
	its biological effect, or models for health effects of air pollution \citep{blangiardo+hr11} 
	with separate modules for predicting pollutant concentrations and predicting health outcomes based on exposure.  See \cite{bayarri2009modularization}
	and \cite{jacob2017better} for further examples.
	
	In principle, Bayesian inference is attractive in modular
	settings  due to its ability to combine the different sources of information
	and update uncertainties about unknowns coherently conditional on all the available data.    
	However, it is well-known that Bayesian inference can be 
	unreliable when the model is misspecified 
	\citep{kleijn2012}.  For conventional Bayesian inference
	in multi-modular models, misspecification in one module can 
	adversely impact inferences about parameters in correctly specified modules.  
	``Cutting feedback" approaches modify Bayesian inference to address
	this issue.  They consider a sequential or conditional decomposition
	of the posterior distribution following the modular structure, and then modify
	certain terms so that unreliable information is isolated and cannot
	influence inferences of interest which may be sensitive to the misspecification.

	Cutting feedback 
	is only one technique belonging to a wider
	class of modular Bayesian inference methods \citep{bayarri2009modularization}.  
	Good introductions to the basic idea and applications of cutting
	feedback are given by \cite{lunn+bsgn09}, \cite{plummer2015cuts} and \cite{jacob2017better}.  Computational aspects of the approach are discussed in \cite{plummer2015cuts}, \cite{jacob2020unbiased}, 
	\cite{liu+g20}, \cite{yu2021variational} and \cite{carmona+n22}.  
	Most of the above references deal only with cutting feedback
	in a certain ``two module" system considered in \cite{plummer2015cuts}, which
	although simple is general enough to encompass many practical
	applications of cutting feedback methods.  We also consider this
	two module system throughout the rest of the paper.   Some recent progress
	in defining modules and cut posteriors in greater generality is
	reported in \cite{liu+g22}.

	A useful extension of cutting feedback is the semi-modular posterior (SMP) approach of \cite{carmona2020semi}, which avoids
	the binary decision of using either the cut or full posterior distribution, and can be viewed as a continuous interpolation between two  distributions. Further developments and applications are discussed 
	in \cite{liu+g21}, \cite{carmona+n22}, \cite{nicholls+lwc22} and \cite{frazier2022cutting}.  
	The  motivation for semi-modular inference is explained clearly in \cite{carmona+n22}: ``In Cut-model inference, feedback from 
	the suspect module is completely cut. However, [...] if the parameters 
	of a well-specified module
	are poorly informed by “local” information then limited information from misspecified modules may allow us to bring the uncertainty down 
	without introducing significant
	bias''.  {\color{black}The above quote nicely describes the intuition behind SMI.   However,  there is no formal treatment of the bias-variance trade-off that exists in SMI, nor is there any rigorous discussion as to how SMI could ``leverage'' such a trade-off in practice.}  
	
	{\color{black}Herein, we make three fundamental  contributions to the literature on cutting feedback and misspecified Bayesian models.  First, we formally demonstrate that when model misspecification is not too severe, cut posteriors can deliver inferences with smaller bias but more variability than standard Bayesian posteriors, which provides formal evidence for the bias-variance trade-off that motivates SMI and SMPs. Second, we use this result to develop a novel SMP that leverages this trade-off.  Lastly, using a notion that captures the risk of a posterior, we demonstrate that the proposed SMP is preferable to the cut, as well as the full posterior under certain conditions. More specifically, under this notion of risk we show that the cut posteriors produce point estimators that are \textit{inadmissible}, while, under additional conditions, the standard posterior is also shown to deliver \textit{inadmissible} point estimators. 
	}
	
	The remainder of the paper is organized as follows. In Section 2 we give the general framework and make rigorous conditions that are necessary for the existence of a bias-variance trade-off in modular inference problems.  In Section 3 we discuss semi-modular inference, and describe our semi-modular posterior approach. In this section, a simple
	example is presented in which the new method produces superior results to those based on the cut and full posterior. In Section 4 we prove, under `classical' regularity conditions, that our semi-modular posterior outperforms the standard and cut posteriors according to a notion of asymptotic risk for a posterior. Section 5 gives two empirical examples, and Section 6 concludes.  Proofs of all stated results are given in the supplementary material.

	\section{Setup and Discussion of Cut posterior}\label{sec2}
	
	Our first contribution is to formalize the potential benefits of using semi-modular inference methods  (\citealp{carmona2020semi}) in misspecified models.  The semi-modular inference approach was originally introduced by \cite{carmona2020semi} using a two-module system discussed by \cite{plummer2015cuts}. In our current work, we focus on the same two-module system.  It is important to describe our motivation for this choice, in the context of previous work on Bayesian modular inference.
	
	One method for defining cutting feedback methods in misspecified models with more than two modules uses an implicit approach by modifying
	sampling steps in Markov chain Monte Carlo (MCMC) algorithms.  
	The {\texttt cut} function in the WinBUGS and OpenBUGS software packages 
	implements this in a modified Gibbs sampling approach.  
	See \cite{lunn+bsgn09} for a detailed description.  However, 
	defining a cut posterior in terms of an algorithm, while quite general, does not allow us to easily understand the implications of cutting feedback, or the general structure of the posterior, due to the implicit
	nature of such a definition.  To give a better understanding, \cite{plummer2015cuts} 
	considered a two-module system where the cut posterior can be defined
	explicitly.  In many models where cut methods are used, there might be one
	model component of particular concern, and a definition of the modules in 
	a two-module system can often be made based on this.  
	Many applications of cutting feedback use such a two-module system.  
	Two module systems also play an important
	role in the recent attempt by \cite{liu+g22} to explicitly 
	define multi-modular systems
	and cutting feedback generally, where existing modules
	can be split into two recursively based on partitioning of the data and using the graphical structure of the model.  
	We now describe the two-module system precisely and describe cutting
	feedback methods for it, before giving a motivating example.  
	
	\subsection{Setup}
	We observe a sequence of data $z_{1:n}=(z_1,\ldots,z_n)$, $z_i \in\mathcal
	{Z}$ for each $i$.  The observations $z_{1:n}$ are considered an observation of
	a random vector $Z$, 
	and wish to conduct Bayesian inference on the unknown parameters $\theta$ 
	in the assumed joint likelihood $f(z_{1:n}\mid\theta)$, where $\theta= (\theta_1^\top,\theta_2^\top)^\top$, $\theta\in\Theta_1\times\Theta_2
	\subseteq\mathbb{R}^{d}$, and $\Theta_1\subseteq \mathbb{R}^{d_1}$, 
	$\Theta_2\subseteq \mathbb{R}^{d_2}$. Our prior beliefs over $\Theta$ are expressed via a prior density $\pi(\theta)=\pi(\theta_1)\pi(\theta_2\mid\theta_1)$. 
	
	In modular Bayesian inference, the joint likelihood $f(z_{1:n}\mid\theta)$ can
	often be expressed as a product, with terms 
	for data sources from different modules. The simplest case is the two module system described in \cite{plummer2015cuts}, where the random variables are $Z=(X,Y)$.  The first module consists of a likelihood term depending on $X$ and $\theta_1$, given by $f_1(\cdot\mid\theta_1)$, 
	and the prior $\pi(\theta_1)$.  The second module consists of a likelihood term depending on $Y$ and $\theta_1,\theta_2$, given by $f_2(\cdot\mid\theta_1,\theta_2)$, and the conditional prior $\pi(\theta_2\mid\theta_1)$.  
	An example of such a model is given below.
	A model with this structure leads to the following \textit{full posterior}
	\begin{flalign}\label{eq:exactpost}
		\pi_{\full}(\theta_1, \theta_2 \mid z_{1:n})&=\pi(\theta_1 \mid x_{1:n}, y_{1:n}) \pi(\theta_2 \mid y_{1:n}, \theta_1),
	\end{flalign}
	in which the conditional posterior for $\theta_2$ given $\theta_1$ does not
	depend on the observations for the random variable $X$, denoted as $x_{1:n}=(x_1,\dots,x_n)$.  
	The parameter $\theta_1$ is shared between the two modules, 
	while $\theta_2$ is specific to the second module.  
	
	While the full posterior in \eqref{eq:exactpost} delivers reliable inference when the model is correctly specified, in the presence of  model misspecification  Bayesian inference can sometimes be unreliable and not ``fit for purpose''; see, e.g., \citet{grunwald2017inconsistency} for examples in the case of linear models, and \cite{kleijn2012} for a general discussion. Following the literature on cutting feedback methods, we restrict our attention  to settings where misspecification is confined to the second module, while specification of the first module is not impacted.  Because the parameter $\theta_1$ is shared
	between modules, inference about this parameter can be corrupted by
	misspecification of the second module.  This can also impact the interpretation
	of inference about $\theta_2$, which can be of interest even if the second
	module is misspecified provided this is done conditionally on values of $\theta_1$
	consistent with the interpretation of this parameter in the first module.  
	
	Rather than use the posterior \eqref{eq:exactpost}, in possibly misspecified models it has been argued that we can cut the link between the modules to produce more reliable inferences for $\theta_1$. This is the idea of cutting feedback; see Figure \ref{two-module} for a graphical depiction of this ``cutting'' mechanism where,
	for simplicity, we assume $\pi(\theta_2\mid\theta_1)$ does not depend on $\theta_1$. In the case of the two module system, the cut (indicated by the vertical dotted line in Figure \ref{two-module}) severs the feedback between the modules and allows us to carry out inferences for $\theta_1$ based on module one, using the likelihood $f_1(\cdot\mid\theta_1)$, and then inference for $\theta_2$ can be carried out conditional on $\theta_1$. In the case of Bayesian inference, this philosophy has led researchers to conduct inference using the cut posterior distribution (see, \citealp{plummer2015cuts}, \citealp{jacob2017better}). 
	
	\begin{figure}[h]
		\centering{\includegraphics[width=35mm, height=30mm]{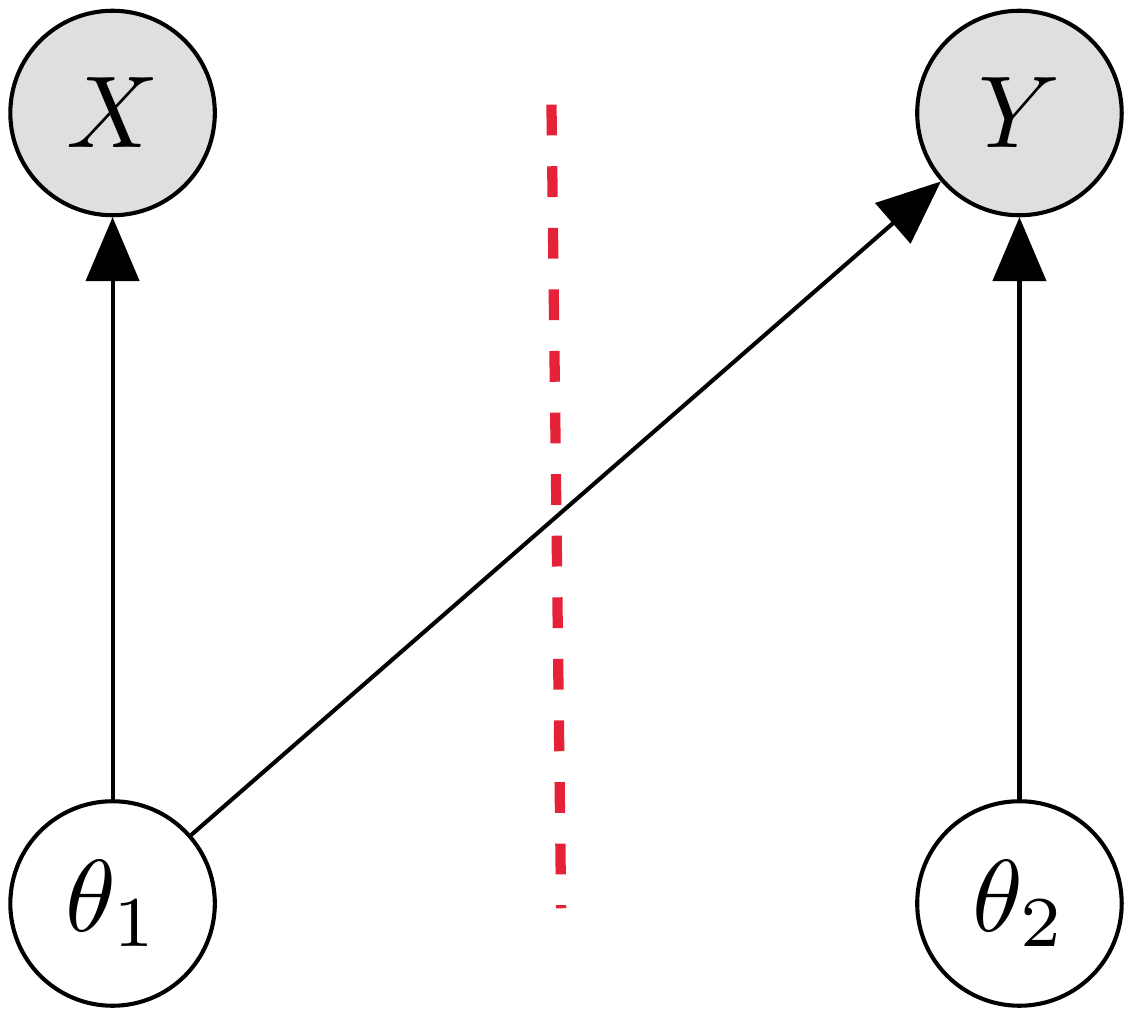}}
		\captionsetup{width=.99\linewidth}
		\caption{Graphical structure of the two-module system.  The dashed line indicates the cut.  Shaded nodes are observed quantities.}\label{two-module}
	\end{figure}
	
	As shown in \cite{carmona2020semi} and \cite{nicholls+lwc22}, the cut posterior is a ``generalized'' posterior distribution (see, e.g., \citealp{bissiri2016general}) that restricts the information flow to guard against model misspecification (\citealp{frazier2022cutting}). In the canonical two module system, the \textit{cut posterior} takes the form 
	\begin{equation*}
		\pi_{\cut}(\theta \mid z_{1:n}):=\pi_{\cut}(\theta_1 \mid x_{1:n}) \pi(\theta_2 \mid y_{1:n}, \theta_1), \quad \pi_\cut(\theta_1\mid x_{1:n})\propto f_1(x_{1:n}\mid\theta_1)\pi(\theta_1).
	\end{equation*}The common argument given for the use of $\pi_{\cut}(\theta\mid z_{1:n})$ instead of $\pi(\theta\mid z_{1:n})$ is the assumption that misspecification adversely impacts inferences for $\theta_1$; e.g., \cite{bayarri2009modularization}, and \cite{jacob2017better}. The cut posterior uses only information from the data $x_{1:n}$ in making inference about $\theta_1$, ensuring inference is insensitive to misspecification of the model for $y_{1:n}$.  However, uncertainty about $\theta_1$ can still be propagated through for inference on $\theta_2$ via the conditional posterior $\pi(\theta_2\mid\theta_1,y_{1:n})$. 
	
	\subsubsection*{Motivating example:  HPV prevalence and cervical cancer incidence}
	
	We now discuss a simple example described in \cite{plummer2015cuts} that
	illustrates some of the benefits of cut model inference.  
	The example,  which is discussed further in the supplementary material, 
	is based on data from a real epidemiological study 
	\citep{maucort2008international}.  Of interest is the international 
	correlation between
	high-risk human papillomavirus (HPV) prevalence and cervical cancer incidence, for women in a certain age group.  There are two data sources.
	The first is survey data on HPV prevalence for 13 countries.  There are 
	$X_i$ women with high-risk HPV in a sample of size $n_i$ for country $i$, 
	$i=1,\dots, 13$.  There is also data on cervical cancer incidence, with $Y_i$
	cases in country $i$ in $T_i$ woman years of follow-up.  The data are modelled as
	$$X_i\sim \text{Binomial}(n_i,\theta_{1,i}),\;\;\;i=1,\dots, 13,$$
	$$Y_i\sim \text{Poisson}(\lambda_i),\;\;\;\log\lambda_i=\log T_i+\theta_{2,1}+\theta_{2,2}\theta_{1,i}.$$
	The prior for $\theta_1=(\theta_{1,1},\dots, \theta_{1,13})^\top$
	assumes
	independent components with uniform marginals on $[0,1]$.  
	The prior for $\theta_2=(\theta_{2,1},\theta_{2,2})^\top$,  assumes independent
	normal components, $N(0,1000)$.    
	
	Module 1 consists of $\pi(\theta_1)$ and $f(X\mid\theta_1)$ (survey data module) and module 2 consists of $\pi(\theta_2)$ and $f(Y\mid\theta_2,\theta_1)$ (cancer incidence module).  The Poisson regression model in the second module is misspecified.  Because of the coupling of the survey and cancer incidence modules, with the HPV prevalence values $\theta_{1,i}$ appearing as covariates
	in the Poisson regression for cancer incidence, the cancer incidence module
	contributes misleading information about the HPV prevalence parameters.  
	The cut posterior estimates these parameters based on the survey data only,
	preventing contamination of the estimates by the misspecified module.  
	This in turn results in more interpretable estimates of the parameter $\theta_2$
	in the misspecified module, since $\theta_2$ summarizes the relationship
	between HPV prevalence and cancer incidence, but the 
	summary produced can only be useful
	when the inputs to the regression (i.e. the HPV prevalence covariate values) are properly estimated.  
	
	Figure \ref{hpv-figure1} shows the marginal posterior distribution for $\theta_2$
	for both the full and cut posterior distribution, which are very different, illustrating how the misspecification of the cancer incidence module distorts inference about HPV prevalence in the full posterior, resulting in uninterpretable
	estimation for $\theta_2$.   Further discussion of this example in the supplement shows 
	the benefits of a semi-modular inference approach in which 
	the tuning parameter $\omega$ interpolating between the cut and full posterior
	can be chosen based on a user-defined loss function reflecting the purpose
	of the analysis.   A more complex real example is considered in Section 5.
	\begin{figure}[h]
		\centerline{\includegraphics[width=80mm]{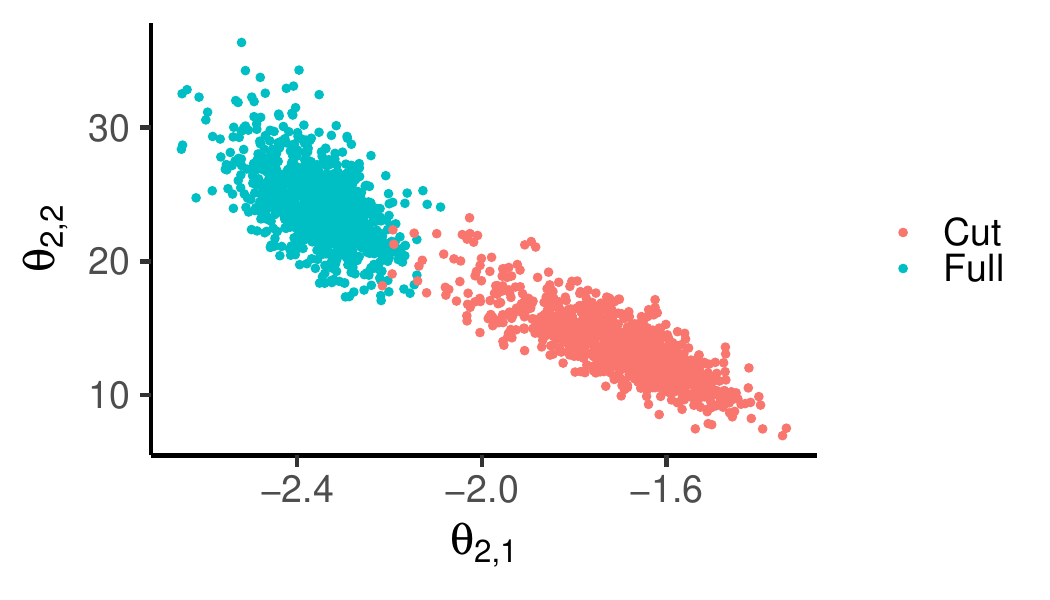}}
		\captionsetup{width=.99\linewidth}
		\caption{\label{hpv-figure1}Marginal full and cut posterior distributions
			for $\theta_2$ in the HPV example.}
	\end{figure}
	
	\subsection{The Impact of Misspecification in Modular Inference}\label{sec:formal}
	
	In the remainder we consider a generalization of the canonical two-module system by assuming that the joint likelihood takes the form 
	$$f(z_{1:n}\mid\theta)= f_1(z_{1:n}\mid\theta_1)f_2(z_{1:n}\mid\theta_1,\theta_2),$$where the individual models may or may not depend on the entire dataset $z_{1:n}$.  In 
	the expressions above, $f_1(z_{1:n}\mid\theta_1)$ and $f_2(z_{1:n}\mid \theta_1,\theta_2)$ need not be densities as functions in $z_{1:n}$ so long as $f(\z\mid\theta)$ itself is a density; the terms $f_1$ and $f_2$ 
	describe a decomposition of the likelihood having the dependence
	indicated by the arguments. {\color{black}The goal of modular Bayesian inference is to conduct inference on $\theta=(\theta_1^\top,\theta_2^\top)^\top$ in the specific setting where inference for $\theta_1$ based only on $f_1(z_{1:n}\mid\theta_1)$ is reasonable, but where the module $f_2(z_{1:n}\mid\theta_1,\theta_2)$ contains additional useful information on $\theta_1$; see the regularity conditions in Appendix \ref{app:A} for specific details. Modular inference on $\theta$ can be carried out using the cut posterior}
	\begin{equation}\label{eq:cutpost}
		\pi_{\cut}(\theta \mid z_{1:n}):=\pi_\cut(\theta_1 \mid z_{1:n}) \pi(\theta_2 \mid z_{1:n}, \theta_1),\; \pi_\cut(\theta_1\mid z_{1:n})\propto { f_1(z_{1:n}\mid\theta_1)\pi(\theta_1)}.
	\end{equation}
	The cut posterior is beneficial when misspecification of
	$f_2(\cdot\mid\theta_1,\theta_2)$ makes full Bayesian inference on $\theta_1$ less accurate compared to Bayesian inference using only $f_1(\cdot\mid\theta_1)$.  
	The benefits of cutting feedback can be formalized by assuming that misspecification
	is limited to $f_2(\cdot\mid\theta_1,\theta_2)$ {\color{black} and that the true data generating 
		process (DGP) has a density of the form}\footnote{Strictly speaking, this analysis extends beyond density functions but we maintain the use of densities to simplify our discussion.}
	$$h(z\mid\theta_{1,0},\delta_0)=f_1(z\mid\theta_{1,0})\delta_0(z),$$ for some unknown component $\delta_0(z)$ that controls the amount of model misspecification. The form of $h(z\mid\theta_{1,0},\delta_0)$ allows us to capture both gross model misspecification, and situations where misspecification is ambiguous. We first restrict our attention to gross model misspecification by imposing the following assumption.

	\begin{assumption}[Gross Misspecification]\label{ass:missgross}
		The observed data $\{z_{i}:i\ge1,n\ge1\}$  is independent and identically distributed according to $h(z\mid\theta_{1,0},\delta_0)=f_1(z\mid\theta_{1,0})\delta_0(z)$. For some $\widetilde{\mathcal{Z}}\subseteq\mathcal{Z}$, with $\int_{\widetilde{\mathcal{Z}}} h(z\mid\theta_{1,0},\delta_0)\dt z>0$, and all $z\in\mathcal{Z}'$,  $\delta_0(z)\ne f_2(z\mid\theta_{1,0},\theta_2)$ for any $\theta_2\in\Theta_2$. There exist $\theta_\star\in\mathrm{Int}(\Theta)$ that minimizes $\theta\mapsto \mathrm{KL}\{h(z\mid\theta_{1,0},\delta_0)\|f(z\mid\theta)\}$, the Kullback–Liebler divergence of $f(z\mid\theta)$ over $\Theta$ with respect to $h(z\mid\theta_{1,0},\delta_0)$. 
	\end{assumption}
	
	Under Assumption \ref{ass:missgross} and regularity conditions, the cut posterior $\pi_{\cut}(\theta_1 \mid z_{1:n})$ can be shown to deliver accurate inferences for $\theta_{1,0}$, while the full posterior $\pi_{\full}(\theta_1 \mid z_{1:n})$ does not. To state this result, denote the cut posterior distribution for $\theta_1$ as  $\Pi_{\cut}(\theta_1\in\cdot\mid z_{1:n})$, and let $\Pi_{\mathrm{full}}(\theta_1\in\cdot\mid z_{1:n})$ denote the full posterior distribution for $\theta_1$. For any $\varepsilon>0$, define $\Theta_{1}(\varepsilon):=\{\theta_1\in\Theta_1: \|\theta_1-\theta_{1,0}\|\le\varepsilon\}$. For $P_0^{(n)}$ the true data generating process, we say that $X_n\xrightarrow{P} a$ (or $X_n=a+o_p(1)$) if for all $\varepsilon>0$, $ P_0^{(n)}(\|X_n-a\|\ge \varepsilon)=o(1)$ as $n\rightarrow\infty$. 
	
	\begin{lemma}\label{lemma:gross_miss} If Assumption \ref{ass:missgross} and the regularity conditions in Appendix \ref{app:G} are satisfied, then for any $\varepsilon>0$, 
		$
		\Pi_{\cut}\left\{\theta_1\in \Theta_{1}(\varepsilon)\mid z_{1:n}\right\}\xrightarrow{P}1$. For $\varepsilon>0$ such that $\theta_{1,\star}\notin \Theta_{1}(\varepsilon)$, $ \Pi_{\mathrm{full}}\left\{\theta_1\in \Theta_{1}(\varepsilon)\mid z_{1:n}\right\}\xrightarrow{P}0$.
	\end{lemma}
	{{\color{black}Assumption \ref{ass:missgross} embodies cases where a practitioner can confidently determine that the second module is misspecified.  
			Lemma \ref{lemma:gross_miss} then shows that, in these cases, the cut posterior concentrates on the true value $\theta_{1,0}$, while the full posterior does not. By restricting the flow of information across modules, cutting feedback produces inferences for $\theta_{1,0}$  that are more accurate than full Bayesian inference.  Although Lemma \ref{ass:missgross} implies that cut posterior inference is more accurate than full posterior inference when the second module is grossly misspecified, it is of limited use when $\pi_{\cut}(\theta_1 \mid z_{1:n})$ and $\pi_{\full}(\theta_1 \mid z_{1:n})$ are similarly located, which can occur when misspecification of the second module is not severe. In such cases, the full posterior will have less uncertainty than the cut posterior, while also being similarly located, making it unclear whether to prefer the cut or full posterior.}}
	
	To capture the empirically relevant setting where, for any fixed $n$, it is unclear which posterior to use, we investigate the behavior of these posteriors under a certain class of locally misspecified DGPs. Following the literature of robust asymptotic statistical analysis (see, e.g., Ch 3. \citealp{rieder2012robust}), we approximate the empirical situation where neither method is clearly preferable using a local perturbation about the assumed model:
	\begin{flalign}\label{eq:localpert}
		h(z\mid\theta_0,\delta_n):=f_1(z\mid\theta_{1,0})\{1+\psi^\top\zeta(z)/\sqrt{n}\}f_2(z\mid\theta_0),\quad \delta_n=\psi^\top\zeta/\sqrt{n},
	\end{flalign}which depends on  a (random) direction of misspecification $\zeta\in\mathbb{R}^{d}$, and magnitude $\psi$ that takes values in $\Delta\subset\mathbb{R}^d$. Under the local perturbation framework in \eqref{eq:localpert} the cut and full posteriors have similar locations in $\Theta$ for small-to-moderate sample sizes, and conventional specification testing methods cannot consistently detect which method delivers more accurate inferences for $\theta_{0}$. As such, this class of perturbations allows us to compare cut and full posterior inference when correct specification of the second module is ambiguous.

	\begin{assumption}[Local Misspecification]\label{ass:miss}
		The triangular array $\{z_{i,n}:1\le i\le n, n\ge1\}$ is independent and identically distributed according to $h(z\mid \theta_{0},\delta_n)$ in \eqref{eq:localpert} for fixed $n$. For $\psi\in\Delta\subset\mathbb{R}^{d}$, $\Delta$ compact, $\zeta(z)$ satisfies:
		(i) $
		\int_{\mathcal{Z}} \frac{\partial \log f_1(z\mid\theta_{1,0})}{\partial\theta_1}\zeta(z)^\top\psi f_1(z\mid\theta_{1,0})f_2(z\mid\theta_{0})\dt\mu(z)=0$; (ii) For 
		$\eta=(\eta_1^\top,\eta_2^\top)^\top$ partitioned conformably to $\theta=(\theta_1^\top,\theta_2^\top)^\top$,   $\eta=\int_{\mathcal{Z}} \frac{\partial \log f(z\mid\theta_{0})}{\partial\theta}\zeta(z)^\top\psi f_1(z\mid\theta_{1,0})f_2(z\mid\theta_{0}) \dt\mu(z)$, and $0\le\|\eta\|<\infty$.
	\end{assumption}
	\begin{remark}\normalfont
		Assumption \ref{ass:miss} is a device that will allow us to rigorously compare the cut and full posteriors when neither is clearly preferable.   The local misspecification framework ensures that, for any $n$, there remains some ambiguity about model misspecification—and thus which posterior to prefer. It provides an appropriate theoretical framework for analyzing cut and full posteriors when the choice between them is unclear. Assumption \ref{ass:miss} further maintains that the direction of misspecification, $\zeta(z)$, does not adversely affect the location of the cut posterior for $\theta_1$ but will impact the full posterior.  
	\end{remark}
	
	\begin{remark}\normalfont
		Assumption \ref{ass:miss} resembles, but is distinct from, the misspecification device employed in \cite{hjort2003frequentist} and \cite{claeskens2003focused} to construct methods that combine and choose between different frequentist point estimators. The approach outlined in Section 8 of \cite{claeskens2003focused} is not appropriate here since their framework can produce cut posterior inferences for $\theta_{1,0}$ that are less accurate than the full posterior, which contradicts the underling reasoning for using cut posteriors. The misspecification framework in Assumption \ref{ass:miss} also differs from the designs in \cite{pompe2021asymptotics}, and \cite{frazier2022cutting}, which are similar to Assumption \ref{ass:missgross}, and ensure that - with probability converging to one  - the researcher knows the model is misspecified. 
		
	\end{remark}

	{\color{black}
		Under Assumption \ref{ass:miss}, cut posterior inference for $\theta_1$ is not impacted by misspecification of the second module. To show this, we require further notation. First, note that 
		$$
		\ell(\theta)=\log f(\z\mid\theta)=\log f_1(\z\mid\theta_1)+\log f_2(\z\mid\theta)\equiv \ell_p(\theta_1)+\ell_c(\theta_1,\theta_2), 
		$$where $\ell_p(\theta_1)$ denotes the partial log-likelihood for $\theta_1$ used in the cut posterior, and $\ell_c(\theta)$ denotes the log-likelihood used in the conditional posterior.  Denote the derivative of the full log-likelihood as $\dot\ell(\theta):=\partial\ell(\theta)/\partial\theta$, and the second derivative by $\ddot\ell(\theta):=\partial^2\ell(\theta)/\partial\theta\partial\theta^\top$. For $j,k\in\{1,2\}$, define the first and second partial derivatives as $\dot\ell_{(j)}(\theta):=\partial \ell(\theta)/\partial\theta_j$, and $\ddot\ell_{(ij)}(\theta):=\partial^2 \ell(\theta)/\partial\theta_j\partial\theta_i^\top$. Similar notations will be use to denote derivatives of $\ell_{u}(\theta)$, for $u\in\{c,p\}$. Let $\E_{n}[g(z)]$ denote the expectation of $g:\mathcal{Z}\rightarrow \mathbb{R}^{d}$, under $h(z\mid \theta_0,\delta_n)$ in Assumption \ref{ass:miss}, and define the following information matrices: $\mathcal{I}_{}:=-\lim_nn^{-1}\E_{n}[\ddot\ell_{}(\theta_0)]$, $\mathcal{I}_{jk}:=-\lim_nn^{-1}\E_{n}[\ddot\ell_{(jk)}(\theta_0)]$, and $\mathcal{I}_{p(11)}:=-\lim_nn^{-1}\E_{n}[\ddot\ell_{p}(\theta_{1,0})]$. Let $\overline\theta_\cut:=\int_{\Theta}\theta\pi_\cut(\theta\mid\z)\dt\theta$,  $\overline\theta_\full=\int_{\Theta}\theta\pi_\full(\theta\mid\z)\dt\theta$, and let $\Rightarrow$ denote weak convergence under $P^{(n)}_0$. 
	}

	\begin{lemma}\label{lem:two}Under Assumption \ref{ass:miss}, and the regularity conditions in Appendix \ref{app:A}: 
		
		\noindent	(i)  $\sqrt{n}(\overline{\theta}_{\full}-\theta_0)\Rightarrow N(\mathcal{I}_{}^{-1}\eta,\mathcal{I}^{-1})$.
		
		\noindent	(ii)  $\sqrt{n}(\overline{\theta}_{1,\cut}-\theta_{1,0})\Rightarrow N(0,\mathcal{I}^{-1}_{p(11)})$.
		
		\noindent	(iii)  $\sqrt{n}(\overline{\theta}_{2,\cut}-\theta_{1,0})\Rightarrow N(\I_{22}^{-1}\eta_2,\mathcal{I}^{-1}_{22}+\I_{22}^{-1}\I_{21}\I_{p(11)}^{-1}\I_{12}\I_{22}^{-1})$.     
		
		\noindent (iv) Only credible sets calculated under $\pi_{\cut}(\theta_1\mid\z)$ \textbf{have valid frequentist coverage.}    
		
		\noindent	(v) The squared bias for $\theta_{1,0}$ and $\theta_{2,0}$ under $\pi_{\cut}(\theta\mid\z)$ is smaller than $\pi_{\mathrm{full}}(\theta \mid\z)$. 
	\end{lemma}	
	\begin{remark}\label{rem:cover}
		{\color{black}\normalfont
			Lemma \ref{lem:two} shows that inferences for $\theta_1$ using the cut posterior have no asymptotic bias, whereas the full posterior for $\theta_1$ has asymptotic bias and both posteriors have different biases for $\theta_2$.  The presence of asymptotic bias implies that, if $\eta\ne0$, \textit{only credible sets for $\pi_\cut(\theta_1\mid\z)$ correctly quantify uncertainty, i.e., only $\pi_\cut(\theta_1\mid\z)$ has calibrated credible sets.}  Since the bias due to misspecification, $\eta$, is unknown, it does seem feasible to determine whether $\pi_\cut(\theta\mid\z)$ or $\pi_\full(\theta\mid\z)$ more reliably quantifies uncertainty in general, except when $\eta=0$ where both methods accurately quantify uncertainty. }		
	\end{remark}
	
	{\color{black}Lemma \ref{lem:two} implies that the user is faced with a trade-off between conducting inference using the cut or full posterior:  inferences under $\pi_\full(\theta\mid\z)$ have the smallest variability possible (via Cramer-Rao) but exhibit a bias of unknown magnitude, whereas inferences under $\pi_\cut(\theta\mid\z)$ are guaranteed to have smaller bias than those under $\pi_\full(\theta\mid\z)$ but have (weakly) larger variability. Lemma \ref{lem:two} is the first result to formally show that a bias-variance trade-off exists between the cut and full posteriors. Since the bias due to misspecification ($\eta$) is unobservable, Lemma \ref{lem:two} exemplifies the situation where it is ambiguous as to whether we should base inferences on the posterior that exhibits low bias - the cut posterior - or the posterior that exhibits larger bias but which has much smaller variability - the full posterior.}

\section{Semi-Modular Inference}
{\color{black} Lemma \ref{lem:two} suggest that is we consider the accuracy of posteriors using a criteria that measures both the bias and variance of posterior inference, it should be possible to combine the cut and full posteriors to deliver inferences that are more accurate than using either by themselves. The goal of semi-modular inference (SMI) is to interpolate between the cut and full posteriors to reduce the uncertainty about $\theta$ while maintaining a tolerable level of bias. However, there are many ways to interpolate between two probability distributions, and there is no {\it a priori} reason to suspect one method of interpolation will deliver better results than others; see \cite{nicholls+lwc22} for a discussion on some of the possibilities. 
	
	Following \cite{chakraborty2022modularized}  and \cite{yu2021variational}, we focus on conducting SMI using linear opinion pools  (\citealp{stone1961opinion}), which produces a  semi-modular posterior (SMP) that is a convex combination between the cut and full posteriors: for $\omega\in[0,1]$, 
	\begin{flalign}
		\pi_\omega(\theta\mid z_{1:n})&:=\pi_\omega(\theta_1\mid z_{1:n})\pi(\theta_2\mid z_{1:n},\theta_1),\label{eq:pool}\\ \pi_\omega(\theta_1\mid z_{1:n})&:=(1-\omega)\pi_\cut(\theta_1\mid z_{1:n})+\omega \pi(\theta_1\mid z_{1:n}).\nonumber
	\end{flalign} 
	The pooling weight $\omega\in[0,1]$ in the SMP determines the level of interpolation between the posteriors, and  \cite{chakraborty2022modularized} suggest choosing $\omega$ using  prior-predictive conflict checks, while \cite{carmona+n22} propose to use out-of-sample predictive methods to select $\omega$. In contrast, we propose a novel choice of pooling weight that leverages the bias-variance-trade-off between cut and full posterior inferences. 
}

\subsection{Shrinkage-based semi-modular posteriors}\label{sec:gsmp}
{\color{black}To build intuition as to how  $\pi_\omega(\theta\mid z_{1:n})$ can utilize the bias-variance trade-off between $\pi_\cut$ and $\pi_\full$, we first focus on  the behavior of the SMP for $\theta_1$, i.e., $\pi_\omega(\theta_1\mid z_{1:n})$ in \eqref{eq:pool}, and analyze the behavior of $\pi_\omega(\theta\mid\z)$ in subsequent sections.\footnote{Differences between the cut, full, and SMP posteriors for $\theta_2$ are attributable to differences in the posterior for $\theta_1\mid\z$ since each posterior shares the same conditional posterior for $\theta_2$ given $\theta_1$.} Focusing on $\pi_\omega(\theta_1\mid z_{1:n})$ in \eqref{eq:pool}, note that the SMP point estimator for $\theta_1$ is
	\begin{flalign}\label{eq:pool_mean}
		\overline\theta_1(\omega)&:=(1-\omega)\overline\theta_{1,\cut}+\omega\overline\theta_{1,\mathrm{full}}\equiv\overline\theta_{1,\cut}-\omega(\overline\theta_{1,\cut}-\overline\theta_{1,\mathrm{full}}),
	\end{flalign}where $\overline\theta_{1,\cut}:=\int_{\Theta_1}\theta_1\pi_\cut(\theta_1\mid z_{1:n})\dt\theta_1$, and $\overline\theta_{1,\mathrm{full}}:=\int_{\Theta}\theta_1\pi_{\mathrm{full}}(\theta_1,\theta_2\mid z_{1:n})\dt\theta$. {\color{black}From Lemma \ref{lem:two}, the asymptotic mean of $\sqrt{n}(\overline\theta_{1,\cut}-\theta_{1,0})$ is zero, while that of $\sqrt{n}(\overline\theta_{1,\mathrm{full}}-\theta_{1,0})$ depends on the misspecification bias $\eta$. Hence, the SMP  posterior mean $\overline{\theta}_1(\omega)$ combines an asymptotically unbiased estimator with high variance,   $\overline\theta_{1,\cut}$,  and an asymptotically biased estimator with small variance, $\overline\theta_{1,\mathrm{full}}$.}	Therefore, if we are willing to tolerate some bias it will be possible to choose $\omega$ to deliver inferences on $\theta_{1,0}$ that are more accurate than either the cut or full posterior alone, at least so long as our measure of ``accuracy'' accounts for both bias and variance. 
	
	The idea of combining biased and unbiased estimators has a long history in statistics, with the most commonly encountered estimators of this kind being shrinkage and James-Stein estimators, see Chapter 5 of \cite{lehmann2006theory} for a textbook discussion. Indeed, the form of  $\overline{\theta}_1(\omega)$ in \eqref{eq:pool_mean} is reminiscent of those encountered in the shrinkage literature.  For two normally distributed estimators, one biased and the other unbiased, \cite{green1991james} demonstrated how to optimally combine such estimators to deliver a shrinkage estimator that is optimal in terms of expected squared error loss. The approach of \cite{green1991james} was extended to more general settings and loss functions by \cite{kim2001james}, and \cite{judge2004semiparametric}. 
	
	While  $\overline{\theta}_{1,\cut}$ and $\overline{\theta}_{1,\mathrm{full}} $ are not normally distributed, they are asymptotically normal, and so one could choose $\omega$ in the SMP using existing shrinkage estimation proposals. Following ideas similar to \cite{green1991james},  and \cite{kim2001james}, we could choose $\omega$ in \eqref{eq:pool} as
	\begin{equation}
		\widetilde\omega_+=\min\{1,\widetilde\omega\},\quad\frac{\gamma}{n(\overline\theta_{1,\cut}-\overline\theta_{1,\mathrm{full}})^\top \Upsilon (\overline\theta_{1,\cut}-\overline\theta_{1,\mathrm{full}})},\label{eq:general_shrink}
	\end{equation}
	for some  $\gamma>0$, and $\Upsilon$ a positive definite $(d_1\times d_1)$-matrix. Since $\overline{\theta}_{1,\mathrm{full}}$ is asymptotically biased, while $\overline{\theta}_{1,\cut}$ is not, using $\widetilde\omega_{+}$ within the SMP  would allow us to interpret the SMP as shrinking cut posterior inferences towards those of the biased full posterior, and so using $\widetilde{\omega}_{+}$ would deliver a type of ``shrinkage'' SMP (S-SMP).	
	
	Semi-modular Bayesian inference is clearly much more general than the linear Gaussian models analyzed in \cite{green1991james}, \cite{kim2001james}, and \cite{judge2004semiparametric}.  Nevertheless, applying shrinkage estimation ideas within semi-modular Bayesian inference should allow us to effectively combine the cut and full posteriors.
	In the following sections, we show that this is indeed the case: empirically and theoretically, SMPs based on weights similar to \eqref{eq:general_shrink} deliver inferences that can be shown to be optimal according to a general notion of asymptotic risk. Before presenting any formal analysis, we first illustrate the behavior of the S-SMP based on $	\widetilde\omega_+$ empirically in a simple example from the cutting feedback literature.
}	
\subsection{Example: Biased Mean}\label{sec:bias_mean} To demonstrate the benefits of the SMP, we consider a minor modification of the biased mean example in Section 2.1 of \cite{bayarri2009modularization}. We observe two datasets generated from independent random variables with the same unknown mean $\theta_1$, but where the assumed model
for the second dataset is incorrect resulting in biased estimation of the 
parameter of interest. 
%
The first dataset corresponds to $n_1$ observations on  a $(d_1\times1)$-dimensional, $d_1\ge1$, random vector that we assume is generated from the model: $$y_{1,i}=\theta_{1}+\epsilon_{1,i},$$ 
where $\theta_1=(\theta_{1,1},\dots,\theta_{1,d_1})^\top$ and $\epsilon_{1,i}$, $i=1,\dots, n_1$, are iid $N(0,I_{d_1})$. However, we also observe an additional dataset, comprised of $n_2> n_1$ observations which is assumed to be from the model: 
$$
y_{2,i} = \theta_{1}+\sigma\epsilon_{2,i},
$$
where $\epsilon_{2i}$, $i=1,\dots, n_2$, are iid $N(0,I_{d_1})$ for unknown $\sigma>0$.  
The prior density for $\theta_1$ is $N(0,I_{d_1})$, and for $\theta_2:=\sigma^2$ is
$\pi(\sigma^2)\propto 1/\sigma^2$.  The parameter of interest is $\theta_1$, and 
we wish to determine how much the second dataset should influence inference
about $\theta_1$, when its assumed model is incorrect, leading
to biased estimation of $\theta_1$.\footnote{The original example of \citet{bayarri2009modularization} is such that for even moderate values of the bias we always prefer the cut posterior. Given this, we have slightly modified the original example to ensure a meaningful trade-off exists between the cut and full posteriors. Without this modification, the S-SMP simply returns the cut posterior in the vast majority of cases.} 

Suppose that in the actual data-generating process
$\epsilon_{2,i}$ is not $N(0,I_{d_1})$, but 
$$
\epsilon_{2,i}\stackrel{iid}{\sim}\begin{cases}
	N(-0.25\cdot \iota_{d_1},0.10\cdot I_{d_1})&\text{ with prob. }\delta\\N(0,0.5 I_{d_1})&\text{ with prob. }1-\delta
\end{cases}
$$ for $i=1,\dots, n_2$, where $\iota_{d_1}$ denotes a $d_1\times 1$ dimensional vector of ones.  When $\delta=0$, 
this reduces to the assumed model, but when
$\delta>0$ we will obtain biased estimation of $\theta_1$ under the assumed model.

For our experiments, we assume $\sigma^2=1/2$ and use an equally spaced grid of values for the contamination  $\delta\in\{0:0.05:0.9\}$.  For each value in this grid, we generate 1000 replications from the above process in the case where $n_1=100$ and $n_2=1000$, and consider two different values of $d_1\in\{1,5\}$. For each dataset, the S-SMP {\color{black}is based on the following version of the weights in \eqref{eq:general_shrink}}: for $\widehat{\mathcal{I}}_{p(11)}^{-1}=\text{Cov}_{\pi_{\cut}(\theta_1\mid z_{1:n})}(\theta_1)$ and $\widehat{\mathcal{I}}_{11.2}^{-1}=\text{Cov}_{\pi_{}(\theta_1\mid z_{1:n})}(\theta_1)$  $$\widetilde\omega_+=\min\left\{1,\widetilde{\omega}\right\}\,\;\widetilde\omega=\frac{\tr[\widehat{\mathcal{I}}_{p(11)}^{-1}-\widehat{\mathcal{I}}_{11.2}^{-1}]}{\|\overline\theta_{1,\cut}-\overline\theta_{1,\mathrm{full}}\|^2}\mathbb{I}\left\{\tr[\widehat{\mathcal{I}}_{p(11)}^{-1}-\widehat{\mathcal{I}}_{11.2}^{-1}]>0\right\}.$$

{\color{black}To compare the impact of misspecification on the different modular posteriors we plot a Monte Carlo estimate of the expected squared error for the point estimators, based on 1000 replicated samples, across the grid of values for $\delta$. The results are presented in  Figure \ref{fig_risk1} and show that for relatively small levels of contamination, the full posterior has lower expected squared error than the cut posterior, due to having a much smaller variance, while at higher levels of contamination the cut posterior has lower expected squared error. However, the expected squared error for the S-SMP is always lower than the cut posterior, which demonstrates that the SMP is able to `trade-off' between the two posteriors to minimize squared error risk  across all  levels of contamination. However, we note that when $d_1=1$ and $\delta=0.90$, the S-SMP and cut posterior give very similar results.\footnote{Appendix D in the supplementary material contains additional experiments for this example. These results show that the S-SMP delivers reasonable results for all choices of $d_1$ and becomes more accurate as the dimension $d_1$ increases.}
}

\begin{figure}[h]
	\centerline{\includegraphics[width=160mm,height=65mm]{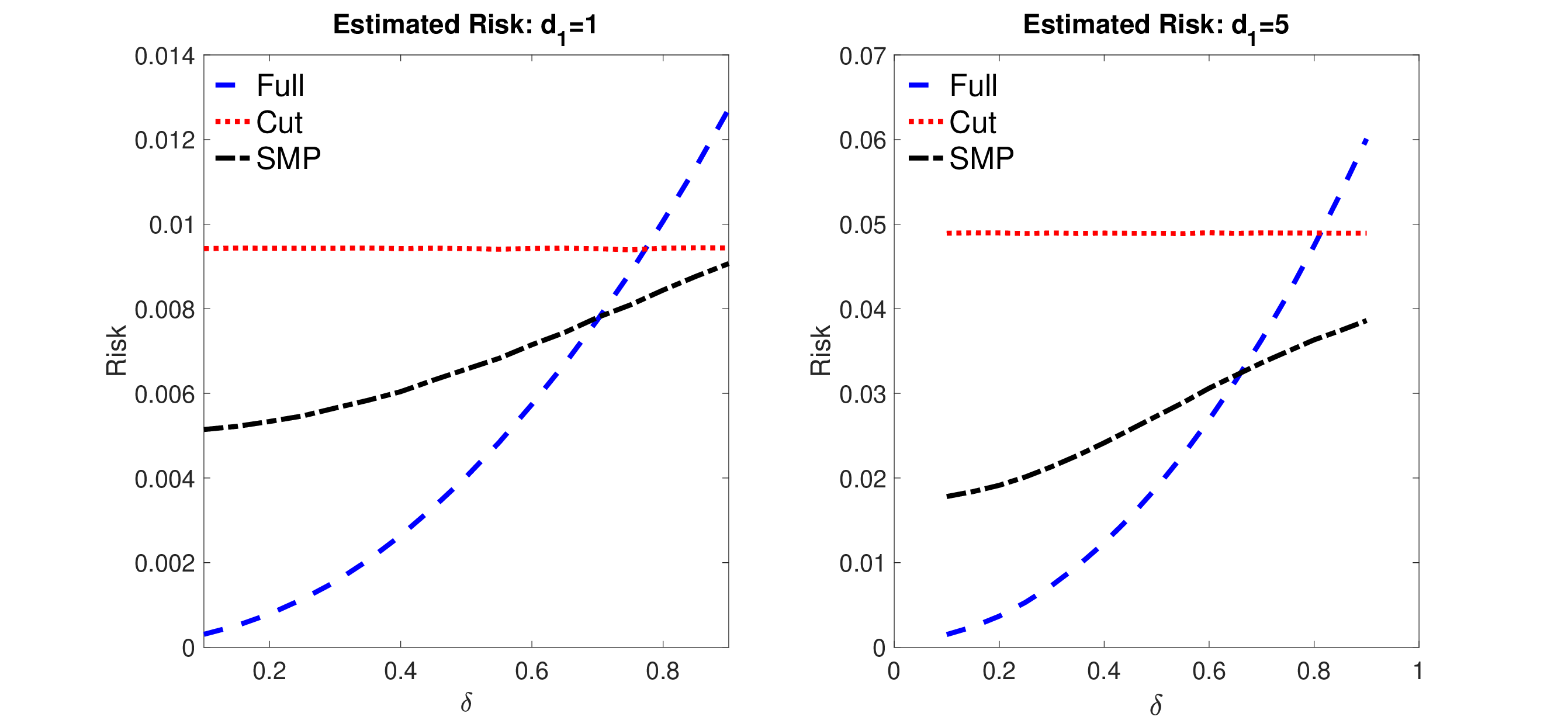}}
	\captionsetup{width=.99\linewidth}
	\caption{{Monte Carlo estimate of expected risk under squared error loss across different levels of contamination ($\delta$). Full refers to the expected risk for $\theta_{1,0}$ associated with the full posterior based on both sets of data; Cut is the P-risk associated to the cut posterior; SMP is the P-risk for the proposed semi-modular posterior. }}
	\label{fig_risk1}
\end{figure}

\section{Measuring posterior accuracy through risk}\label{sec:allrisk}
{\color{black}

	{The example in Section 	\ref{sec:bias_mean} suggests that the S-SMP can deliver more accurate inferences, in terms of expected squared error, than the cut or full posterior. This generally suggests that the S-SMP can deliver inferences that are accurate according to a criteria that measures both the bias and variance of posterior inferences. However, we remark that such a notion of accuracy is only one way to measure the accuracy of posterior inferences. In modular Bayesian inference, \cite{jacob2017better} and \cite{pompe2021asymptotics} have suggested choosing between full and cut posteriors using out-of-sample predictive accuracy, while \cite{carmona2020semi} have suggested a similar approach for calibrating an SMI tuning parameter. While such an approach is undoubtedly useful,  the empirical analysis in \cite{jacob2017better}, as well as the empirical and theoretical analysis in \cite{pompe2021asymptotics}, suggest that the preferred method according to such criteria is example specific, with neither method likely to be generally preferable. Further, Lemma \ref{lem:two} demonstrates that cut and full posterior inferences for $\theta_0$ exhibit asymptotic bias, and it is not clear how notions of predictive accuracy account for the impact of this bias on the resulting inferences.}
	
	In contrast to predictive approaches, we propose to measure the accuracy of modular and semi-modular posteriors through an inferential criterion that can accounts for both the bias and variance in the resulting inferences. 
	Given that the SMP includes the cut ($\omega=0$) and full posterior ($\omega=1$) as special cases, we evaluate the accuracy of different posteriors using the ``posterior risk'' associated to different $\omega$ values. This notion of risk has previously been used by \cite{castillo2014bayesian} and \cite{lee2018optimal} to choose between different prior classes in Bayesian inference, and is capable of capturing  both the bias and variance associated with posterior inferences.
	
	Given a user-chosen loss function $q:\Theta\times\Theta\rightarrow\mathbb{R}_+$, at the point $\theta_0\in\Theta$ we can measure the loss associated with $\omega\in[0,1]$ via the posterior risk $\mathbb{E}_{\pi_\omega}\{q(\theta,\theta_0)\mid\z\}:=\int_\Theta q(\theta,\theta_0)\pi_\omega(\theta\mid \z)\dt\theta$. 	For $g:\mathcal{Z}\rightarrow\mathbb{R}^d$ recall that $\E_{n}[g(z)]=\int_{\mathcal{Z}}g(z)h(z\mid\theta_0,\delta_n)\dt \mu(z)$, with $h(z\mid \theta_0,\delta_n)$ as in Assumption \ref{ass:miss}. The trimmed \textit{Posterior risk} of $\pi_\omega$ at $\theta_0$ (hereafter, referred to simply as P-risk) is defined as:\footnote{Trimming is necessary to ensure that the expectation exists, and can be disregarded in practical terms.}
	$$
	\mathrm{R}_q(\pi_\omega,\theta_{0}):=\lim_{\nu\rightarrow\infty}\liminf_{n\rightarrow\infty}\E_n\min\{\mathbb{E}_{\pi_\omega}\{n\cdot q(\theta,\theta_0)\mid\z\},\nu\}.
	$$ 
	The loss function $q:\Theta\times\Theta\rightarrow\mathbb{R}_+$ satisfies the following assumptions.
	\begin{assumption}\label{ass:loss}
		For all $\theta,\delta\in\Theta$, the loss function $q(\delta,\theta)$ satisfies: i) $q(\delta,\theta)\ge0$ and $q(\delta,\theta)=0$ if and only if $\theta=\delta$; ii) The loss is absolutely continuous in $\delta $, and three times continuously differentiable in $\delta$; iii) For $\delta$ in a neighbourhood of $\theta_{0}$, the matrix $\Upsilon(\delta)=\partial^2 q(\delta,\theta_{0})/\partial \delta\partial \delta^\top$ is continuous, and positive semi-definite; iv) For each $i=1,\dots,d$, and for all $x\in\mathbb{R}^d$, $\delta\in\Theta$, $|x^\top \{\partial \Upsilon(\delta)/\partial\delta_i\}x|\le M\|x\|^2$, where $\delta_i$ denotes the $i$-th direction of $\delta$, and $M>0$. 	
	\end{assumption}
	\begin{remark}\normalfont
		Assumption \ref{ass:loss} includes losses such as squared error loss $q(\delta,\theta)=\|\delta-\theta\|^2$, but also permits  intrinsic measures of accuracy for densities. The Kullback-Liebler divergence,
		$$
		q(\theta,\theta_0)=\int_{\mathcal{Z}}f(z\mid\theta_{0})\log \frac{f(z\mid\theta_{0})}{f_1(z\mid\theta_{1})f_2(z\mid\theta_2,\theta_{1})}\dt z,
		$$
		and various scoring rules, such as kernel scores or mean-variance scores (see \citealp{gneiting2007strictly}, Sections 4 and 5), will satisfy Assumption \ref{ass:loss}. Assumption \ref{ass:loss} does exclude discontinuous functions, such as those needed to measure the accuracy of quantiles. Extending our results to the case of discontinuous losses is the focus of subsequent work by the authors.
	\end{remark}
	
	P-risk is related to expected asymptotic risk, which is often used to gauge the accuracy of frequentist point estimators. Since $\mathrm{R}_q(\pi_\omega,\theta_{0})$ is calculated from a chosen posterior based on a decision made for $\omega$, we refer to this notion as P-risk to distinguish it from asymptotic risk. Risk has a long history in statistical analysis and we refer to Chapter 6 of \citet{lehmann2006theory}, and Chapter 8 of \citet{van2000asymptotic} for textbook treatments. The key benefit of using P-risk to measure different choices for $\omega$ is that,  for the chosen loss $q(\cdot,\cdot)$, P-risk can deliver a concrete ranking across inference procedures relative to this choice. Furthermore, our use of P-risk is not at odds with Bayesian inference, and has already been used by others, albeit in slightly different contexts. More generally, as argued by \cite{lehmann2006theory} (pg 310): ``The Bayesian paradigm is well suited for the construction of possible estimators, but is less well suited for their evaluation.'' Consequently, we follow this suggestion of \cite{lehmann2006theory} and carry out inference via Bayesian methods but evaluate the accuracy of these methods using our notion of asymptotic risk (i.e., P-risk).

	\subsection{The P-risk of SMPs}\label{sec:p_risk}
	
	As we have already seen in Section \ref{sec:gsmp},  weights of the form presented in \eqref{eq:general_shrink} deliver a SMP whose expected squared error was smaller than the cut posterior. Thus, in the remainder we focus our theoretical analysis on SMPs with pooling weights that are general versions of those in \eqref{eq:general_shrink}: for $\gamma_n$ a user-chosen sequence such that $\gamma_n\xrightarrow{P} \gamma$, $\gamma>0$, define the pooling weight
	\begin{flalign}\label{eq:pool_weight_all}
		\widehat\omega_+:=\min\{1,\widehat\omega\},\quad\widehat\omega=\frac{\gamma_n}{n(\overline\theta_{\cut}-\overline\theta_{\mathrm{full}})^\top\Upsilon(\overline\theta_{\cut}) (\overline\theta_{\cut}-\overline\theta_{\mathrm{full}})}.
	\end{flalign}
	In contrast to the weight $\widetilde\omega_{+}$ suggested in \eqref{eq:general_shrink}, the weight $\widehat\omega_+$ in \eqref{eq:pool_weight_all} depends on the entire vector $\theta$, and the curvature of the loss function, as captured by the matrix $\Upsilon(\theta)=\partial^2g(\theta,\delta)/\partial\delta\partial\delta^\top|_{\delta=\theta}$. Incorporating the curvature of the loss within the pooling weight is necessary for the SMPs to deliver inferences that are accurate according to the chosen loss; see Theorem \ref{thm:ests} of Appendix \ref{app:proofs} for further discussion. 
	
	The pooling weight in \eqref{eq:pool_weight_all} yields the shrinkage SMP (S-SMP): 
	\begin{flalign}\label{eq:smp}
		\pi_{\widehat\omega_+}(\theta\mid z_{1:n}):=\{(1-\widehat\omega_+)\pi_{\cut}(\theta_1\mid z_{1:n})+\widehat\omega_+\pi(\theta_1\mid z_{1:n})\}\pi(\theta_2\mid\theta_1,z_{1:n}).
	\end{flalign} Different choices of $\gamma_n$ in \eqref{eq:pool_weight_all} deliver different weights and ultimately different posteriors. However, the following result shows that across a range of choices for $\gamma_n$, the P-risk of $\pi_{\widehat\omega_{+}}$ in \eqref{eq:smp} is dominated by that of the cut posterior, and, under certain conditions, that of the full posterior. To state this result simply, define the $d_{}\times d_{}$-dimensional matrix $\Upsilon:=\Upsilon (\theta_{0})$, and denote the $(d_{2}\times d_2)$-block of $\Upsilon$ by $\Upsilon_{22}$ (see Assumption \ref{ass:loss}); let $\mathcal{M}$ be a $d\times d$ matrix with $(d_1\times d_1)$-block $W:=(\mathcal{I}_{p(11)}^{-1}-\I_{11.2}^{-1})$, $(d_2\times d_2)$-block $V:=\I_{22}^{-1}\I_{21}W\I_{12}I_{22}^{-1}$, and $(d_1\times d_2)$-block $-W\I_{12}\I_{22}^{-1}$, where $\I_{11.2}:=\I_{11}-\I_{12}\I_{22}^{-1}\I_{21}$. We remind the reader that $\I$, $\I_{jk}$ ($j,k\in\{1,2\}$) and $\I_{p(11)}$ are defined in  Section \ref{sec:formal}.


	\begin{theorem}\label{corr:bound}
		Assumptions \ref{ass:miss}-\ref{ass:loss}, and the regularity conditions in Appendix \ref{app:A} are satisfied. If 
		$\tr \Upsilon\M\ge2\|\Upsilon\M\|$, and $0<\gamma\le 2(\tr \Upsilon\M-2\|\Upsilon\M\|)$, then:
		
		\noindent(i) $
		\sup_{\eta:\|\eta\|<\infty}\{\mathrm{R}_q(\pi_{\widehat\omega_+},\theta_{0})-\mathrm{R}_q(\pi_{\cut},\theta_0)\}\le0;
		$
		
		\noindent(ii) $
		\sup_{\eta\in\mathcal{E}}\{\mathrm{R}_q(\pi_{\widehat\omega_+},\theta_{0}) -  \mathrm{R}_q(\pi_{\mathrm{full}},\theta_0)\}\le0
		$ for $\mathcal{E}:=\{\eta:\|\I^{-1}\eta\|_{\Upsilon}^2-\|\mathcal{I}_{22}^{-1}\eta_2\|^2_{\Upsilon_{22}}\ge \tr \Upsilon\M\}$.
	\end{theorem}
	\begin{remark}\normalfont
		Under the class of misspecified DGPs in Assumption \ref{ass:miss} and in terms of P-risk, Theorem \ref{corr:bound} demonstrates that the cut posterior delivers a point estimator that is \textit{inadmissable}, with a similar results also being true for the full posterior under additional conditions. Hence, from the standpoint of P-risk, the S-SMP is preferable to the cut posterior and under certain conditions the full posterior as well. However, it is unclear if the S-SMP has the lowest possible P-risk, i.e., if it is minimax. Answering this question requires deriving a local minimax efficiency bound for the class of DGPs in Assumption \ref{ass:miss} (see, e.g., Chapter 8 in \citealp{van2000asymptotic} for a discussion on asymptotic minimax estimators), which is outside the scope of the current paper, and is left for future research.
	\end{remark}
	\begin{remark}\normalfont
		As suggested by an anonymous referee, a Bayesian may also be interested in the uncertainty quantification of the S-SMP. As discussed in Remark \ref{rem:cover} after Lemma \ref{lem:two}, only credible sets based on $\pi_\cut(\theta_1\mid\z)$ deliver calibrated uncertainty quantification. Therefore, if $\widehat\omega_+>0$ the credible sets of the S-SMP will not be calibrated. 
		More generally, the credible sets calculated from the cut posterior, full posterior and S-SMP all depend, in different ways, on the magnitude of the bias induced via misspecification. Since $\eta$ is unknown in practice it is not obvious how one should theoretically compare the behavior of such credible sets. In light of this issue, we believe that measuring the accuracy of the posteriors using P-risk is the most direct approach. 
	\end{remark}
	
	\begin{remark}{\normalfont
			The condition $\tr \Upsilon\M\ge2\|\Upsilon\M\|$ in Theorem \ref{corr:bound} is related to the inefficiency that results from using the cut posterior relative to the full posterior. This condition is more likely to be satisfied when the efficiency gap between the cut and full posterior is large, or when $d$ - the dimension of $\theta$ - is large. That being said,  the condition $\tr \Upsilon\M\ge2\|\Upsilon\M\|$ is a sufficient condition and it is possible that the S-SMP may deliver smaller P-risk even when this condition is not satisfied.
		}	
	\end{remark}

	\begin{remark}\normalfont
		The example in Section \ref{sec:bias_mean} demonstrated that a S-SMP focused on inference for $\theta_1$ delivered smaller expected P-risk for $\theta_{1,0}$ than the cut posterior. However, Theorem \ref{corr:bound} makes clear that the S-SMP can also deliver smaller P-risk for $\theta_0$. Returning to the example in Section \ref{sec:bias_mean}, we now analyze the P-risk at $\theta_0$ for the S-SMP under the shrinkage weight
		$$
		\widehat\omega_+=\min\{1,\widehat\omega\},\quad\widehat\omega=\frac{\tr[\text{Cov}_{\pi_\cut}(\theta)-\text{Cov}_{\pi_\full}(\theta)]}{\|\overline\theta_{\cut}-\overline\theta_{\mathrm{full}}\|^2}\mathbb{I}\left\{\tr[\text{Cov}_{\pi_\cut}(\theta)-\text{Cov}_{\pi_\full}(\theta)]>0\right\}.
		$$ To this end, we repeat the Monte Carlo experiment in Section \ref{sec:bias_mean} under two different dimensions for $d_1\in\{1,5\}$, so that $d\in\{2,6\}$, and present the results in Figure \ref{fig_risk2}. These results show that the P-risk of the S-SMP is dominated by that of the cut posterior, and in certain cases that of the full posterior; as with the example in Section \ref{sec:bias_mean}, for $\delta=0.90$, and $d=2$, the S-SMP and cut posterior give very similar results. Theorem \ref{corr:bound} is asymptotic and given the sample sizes considered in this experiment it is not surprising that at large levels of contamination the S-SMP can perform slightly worse than the cut posterior (when $d=2$), since non-negligible weight is placed on the full posterior. However, at $\delta=0.90$ the median pooling weight across the replications is about $0.30$, so that most of the pooling weight corresponds to the cut posterior. As we shall see shortly, under higher levels of misspecification and as $n$ increases, the S-SMP resembles the cut posterior. 
		\begin{figure}[h]
			\centerline{\includegraphics[width=160mm,height=60mm]{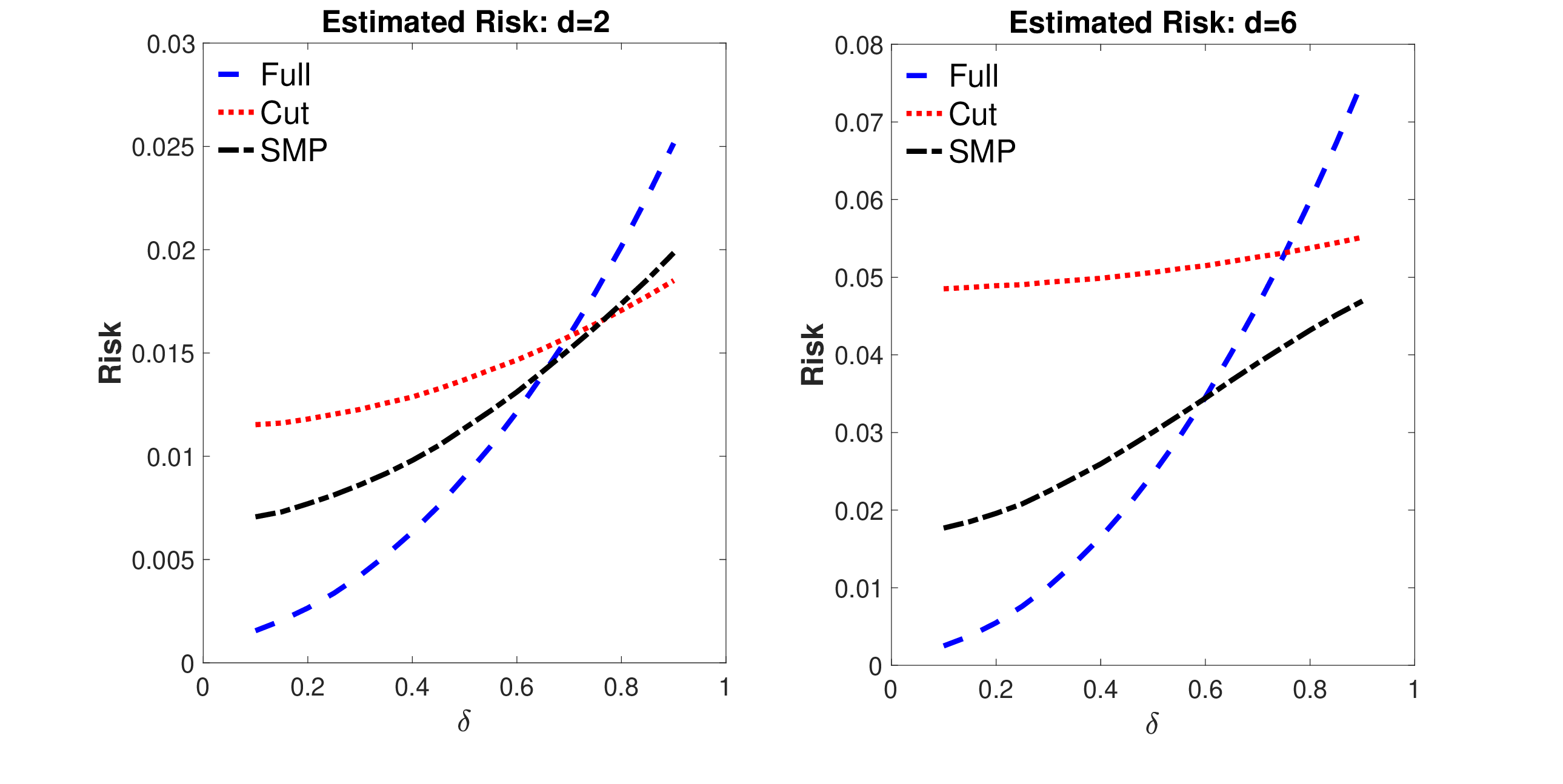}}
			\captionsetup{width=.99\linewidth}
			\caption{{Monte Carlo estimate of expected risk for $\theta_0$ under squared error loss across different levels of contamination ($\delta$). Please see Section \ref{sec:bias_mean} and Figure \ref{fig_risk1} for further details.}}
			\label{fig_risk2}
		\end{figure}
	\end{remark}
	
	Part one of Theorem \ref{corr:bound} implies that if the cut posterior is inefficient relative to the full posterior, as measured by $\tr \Upsilon\M\ge2\|\Upsilon\M\|$, then the S-SMP will be at least as accurate (in P-risk) as the cut posterior, and potentially more accurate than the full posterior.\footnote{Theorem \ref{corr:bound} applies even if $\eta=0$: when there is no misspecification bias the cut and full posterior means will be similar and the weight $\widehat{\omega}_+$ will be close to unity so that the S-SMP will resemble the full posterior.} The second part of Theorem \ref{corr:bound} gives a sufficient, but not necessary, condition which guarantees that the P-risk of the full posterior dominates that of the S-SMP. This condition is likely to be satisfied when the difference in posterior locations is larger than the difference in posterior variances.

	When $d>2$ it is also possible to obtain an analytic expression for the P-risk of the S-SMP if we set $q(\theta,\theta_0)=\frac{1}{2}(\theta-\theta_0)^\top\M^{-1}(\theta-\theta_0)$, so that $\Upsilon=\M^{-1}$. The requirement that $d>2$ is commonly encountered in the risk analysis of James-Stein estimators and is a consequence of the fact that $\pi_{\widehat\omega_+}(\theta\mid\z)$ can be viewed as performing a type of posterior shrinkage.
	\begin{theorem}\label{corr:bound2}
		Assumptions \ref{ass:miss}-\ref{ass:loss}, and the regularity conditions in Appendix \ref{app:A} are satisfied. If $\Upsilon=\M^{-1}$, $d_{}>2$, and $0<\gamma\le2(d_{}-2)$, then for any finite $\eta$,
		\begin{flalign*}
			\mathrm{R}_q(\pi_{\widehat\omega_+},\theta_{0})&=	\mathrm{R}_q(\pi_{\cut},\theta_0)-\frac{\gamma\{2(d_{}-2)-\gamma\}(d_{}-3)!}{2(d_{})!}{_{1}\mathrm{F}_1(d_{}-1;d_{1};\lambda)}\le \mathrm{R}_q(\pi_{\cut},\theta_0),
		\end{flalign*}where ${_{1}\mathrm{F}_1(k-1;k;\lambda)}$ is the confluent hyper-geometric function.
	\end{theorem}
	Theorem \ref{corr:bound2} is useful as it gives an exact bound on the P-risk, and an easily interpretable set of conditions for the value of  $\gamma$ in the weight $\widehat{\omega}_+$. Under $\Upsilon=\M^{-1}$, the condition $d_{}>2$ is necessary to guarantee that  $\pi_{\widehat\omega_{+}}(\theta\mid z_{1:n})$ has smaller P-risk than $\pi_\cut(\theta\mid z_{1:n})$. This condition is  related to Stein's phenomenon (see Ch. 6 of \citealp{lehmann2006theory} for a discussion) and implies that using the cut posterior by itself is sub-optimal (in terms of P-risk) when $d_{}>2$. We stress that this interpretation is only valid when $\Upsilon=\M^{-1}$ and that a similar phenomena does not necessarily extend to other loss functions. 
	
	Theorems \ref{corr:bound}-\ref{corr:bound2} demonstrate that, under certain conditions and in terms of P-risk, the S-SMP is more accurate than the cut posterior and possibly the full posterior.  However, Theorems \ref{corr:bound}-\ref{corr:bound2} implicitly require that the difference between the cut and full posterior locations do not diverge, which is a consequence of the asymptotic regime in Assumption \ref{ass:miss}. This begs the question of what happens to the S-SMP when we move from the case of local model misspecification (Assumption \ref{ass:miss}) to gross model misspecification (Assumption \ref{ass:missgross}). The following result demonstrates that under gross model misspecification the S-SMP converges to the cut posterior, and so is robust to either form of misspecification.
	\begin{corollary}\label{corr:gross}
		Assumption \ref{ass:missgross}, Assumption \ref{ass:loss}, and the regularity conditions in Appendix \ref{app:RegG} are satisfied. If $\|\overline\theta_{\cut}-\theta_{0}\|=O_p(n^{-1/2})$, and  $\|\overline\theta_{\mathrm{full}}-\theta_{\star}\|=O_p(n^{-1/2})$, then
		$
		\int_\Theta |\pi_{\widehat\omega_+}(\theta\mid z_{1:n})-\pi_\cut(\theta\mid z_{1:n})|\dt\theta=o_p(1).
		$
	\end{corollary}
}

\section{Additional examples}
\subsection{Normal-normal random effects model}
We first apply the S-SMP to the misspecified normal-normal random effects model presented in \cite{bayarri2009modularization}. The observed data is $z_{ij}$ comprising observations on  $i=1,\dots,N$ groups, with $j=1,\dots,J$ observations in each group, which we assume are generated from the model $z_{ij}\mid \beta_i,\varphi_i^2\stackrel{iid}{\sim} N(\beta_i,\varphi_i^2)$, with random effects $\beta_i\mid\nu \stackrel{iid}{\sim} N(0,\nu^2)$. The goal of the analysis is to conduct inference on the standard deviation of the random effects, $\nu$, and the residual standard deviation parameters $\varphi=(\varphi_1,\dots,\varphi_N)^\top$.  Below we write $\beta=(\beta_1,\dots, \beta_N)^\top$, and
$\zeta=(\nu,\beta^\top)^\top$.  

For $\bar{z}_i=J^{-1}\sum_{j=1}^{J}z_{ij}$ and $s_i^2=\sum_{j=1}^J (z_{ij}-\bar{z}_i)^2$, $i=1,\dots, N$, the likelihood for $\zeta,\varphi$ can be written to depend only on the sufficient statistics $\bar{z}=(\bar{z}_1,\dots,\bar{z}_N)^\top$ and $s^2=(s_1^2,\dots,s_N^2)^\top$, where independently for $i=1,\dots, N$, 
\begin{flalign*}
	& \bar{z}_i\mid\zeta,\varphi \sim N(\beta_i,\varphi_i^2/J),\quad s^2_i\mid\varphi\sim \text{Gamma}\left(\frac{J-1}{2},\frac{1}{2}\frac{1}{\varphi_i^2}\right).
\end{flalign*}Letting $\theta_1=\varphi$ and $\theta_2=\zeta$, the random effects model can then be written as a two-module system of the form shown in Figure \ref{two-module}:  module one depends on $(s^2,\theta_1)$, $X=s^2$, and module two depends on $(\bar{z},\theta_2,\theta_1)$, $Y=\bar{z}$.

Let $\text{Gamma}(x;A,B)$ denote the value
of the $\text{Gamma}(A,B)$ density evaluated at $x$, and 
$N(x;\mu,\sigma^2)$ denote the value of the $N(\mu,\sigma^2)$ density
evaluated at $x$. The first module has likelihood  $f_1(X\mid \theta_1)= \prod_{i=1}^{N}\text{Gamma}\left(s_i^2;\frac{J-1}{2},\frac{1}{2}\frac{1}{\theta_{1,i}^2}\right)$, while the second module has likelihood  $f_2(Y\mid \theta)=\prod_{i=1}^{N}N(\bar{Z}_i;\beta_i,\theta_{1,i}^2/J)$. 

When the Gaussian prior for the random effects term $\beta_i$ conflicts with the likelihood information, inferences for $\theta_{1,i}=\varphi_i$ can be adversely impacted. Such an outcome will happen when, for instance, {\color{black}a value of $\beta_i$ differs markedly from its assumed model. \cite{bayarri2009modularization} argue that the thin-tailed Gaussian assumption for the random effects can produce poor inferences for $\theta_{1,i}$ due to the feedback induced by the likelihood term $N(\bar{Z}_i;\beta_i,\varphi_i^2/J)$ in the second module. 
	To guard against this \cite{bayarri2009modularization}} propose cut posterior inference for $\theta_1\mid s^2$, which can be accommodated by simply updating the posterior for $\theta_1$ using only the information in the corresponding summary statistics $s^2$: given $\pi(\theta_{1,i}^2)\propto(\theta_{1,i}^2)^{-1}$, and independent across $i=1,\dots,N$, the cut posterior for $\theta_1^2$ (where this denotes the elementwise square of
$\theta_1$) is
$$
\pi_\cut(\theta_1^2\mid X)\propto \prod_{i=1}^{N}(\theta_{1,i}^2)^{-\frac{J+1}{2}}\exp\left\{-\frac{J\cdot s_i^2}{2\theta_{1,i}^2}\right\}.
$$
Summaries of the cut posterior for $\theta_1$ can be obtained by sampling from the cut posterior for $\theta_1^2$ and transforming
the samples.  Joint inferences for $(\zeta,\varphi)$ can be carried out using the cut posterior distribution 
\begin{flalign*}
	\pi_\cut(\zeta,\varphi\mid X,Y)	& = \pi_\cut(\varphi\mid X)\pi(\zeta\mid Y,\varphi)=\pi_\cut(\varphi\mid s^2)\pi(\zeta\mid \bar{Z},\varphi),\quad \varphi=\theta_1,
\end{flalign*}
where the conditional posterior $\pi(\zeta\mid \bar{Z},\varphi)$ is obtained from the joint posterior for $\theta_1,\theta_2$.

{\color{black}We now demonstrate that the S-SMP delivers inferences for $\theta_{1,i}$ that are more accurate than the cut posterior for different levels of misspecification. We generate 500 repeated samples from the normal-normal random effects model with $N=100$ groups each with random effect component $\beta_i$, $i=1,\dots,N$. For each group, we set $\theta_{1,i}:=\varphi_i=0.50$ and $\nu=1$. We induce model misspecification through the random effect term $\beta_1$. Following the design of \cite{liu+g20} we induce misspecification by forcing $\beta_1$ to be an outlier, however, unlike \cite{liu+g20} we consider that the magnitude of the outlier decreases as the number of individual observations in each group, $J$, increases. We set the random effect for the first group as $\beta_1=50/J$, and consider $J\in\{5,10,20,50\}$. When $J$ is small the cut posterior delivers more accurate inferences for $\theta_{1,1}$ than the full posterior as the feedback between this outlier and $\theta_{1,1}$ has been removed. As the magnitude of the outlier shrinks, the full posterior for $\theta_{1,1}$ become more accurate and a meaningful trade-off between the cut and full posteriors exists.}

{\color{black}Our goal is to measure the impact of misspecification on the inferences for $\theta_{1,1}$, and so we choose the weight in the S-SMP using squared error loss for this component only. This produces a pooling weight that is similar to that discussed in Section \ref{sec:bias_mean} but based on $\theta_{1,1}$ and not the entire vector of $\theta_1$.\footnote{Since only the first random effect component, $\beta_1$, is misspecified, inferences for $\theta_{1,2},\cdots,\theta_{1,N}$ are not impacted by misspecification. In this way, even if the pooling weight was estimated based on the entire vector for $\theta_1$ it would be the $\theta_{1,1}$ component that would drive the pooling weight since the posterior means for the cut and full posterior are very similar for the remaining components. }} 	{\color{black}We present Monte Carlo estimates (calculated over the 500 replicated datasets) of the corresponding P-risk under squared error loss for the posteriors in Table \ref{tab:cuts}. Across all values of $J$, the S-SMP has smaller P-risk than the cut posterior, and in many cases the full posterior as well. Further, the cut posterior is more accurate than the full posterior when $J$ is small but the full posterior becomes more accurate as $J$ increases. In this way, the weight in the S-SMP is close to unity with small $J$, but increases as $J$ increases. However, for $J$ large the cut and full posteriors behave similarly, and the S-SMP maintains most of the weight on the cut posterior. 
}
\begin{table}[H]
	\centering
	{\footnotesize		\begin{tabular}{lccccccccccc}
			\hline
			
			$J$  & {$J=5$} & {$J=10$} & {$J=20$}& {$J=50$}& {$J=100$}\\
			\midrule
			$\mathrm{S-SMP}$   &\textbf{10.5365}& \textbf{2.9705}& \textbf{1.1753}& \textbf{0.3723}& \textbf{0.1897}\\
			$\mathrm{Full}$  &8678.9763  &4.6894 &1.3124  &0.5491  &0.2634\\
			$\mathrm{Cut}$  &10.5365& 3.5806& 1.4721& 0.5102& 0.2581\\
			\hline		$\widehat{\omega}_+$  & 0.00& 0.18& 0.23& 0.15&0.15\\
			
			\hline
	\end{tabular}} \vspace{-0.1in}
	\captionsetup{width=.99\linewidth}
	\caption{P-risk values under squared error loss, multiplied by 100, for readability. Bold values indicate the lowest value of risk across the methods. $\widehat{\omega}_+$ is the average posterior weight across the replications. Misspecification decreases as the number of individual observations per-group (J) increases.}
	\label{tab:cuts}
\end{table}

\subsection{Archaeological example}
{\color{black}
	Our final example, discussed in \cite{styring17}, \cite{carmona2020semi}
	and \cite{yu2021variational}, involves data collected to evaluate an ``extensification hypothesis'' for
	early Mesopotamian agricultural practices. The hypothesis states that as cities grew, 
	agriculture extended over larger areas with less intensive cultivation, rather than 
	more intensively farming existing areas to meet food demands.  
	
	The analysis uses two data sources: an archaeological dataset and a modern 
	experimental dataset. Figure \ref{archaeology}, which
	is similar to Figure 6 from \cite{carmona2020semi}, shows a graphical representation
	of the model which comprises two modules. The first, the ``HM module'', is a Gaussian 
	linear regression model incorporating random effects. The second, the ``PO module'', 
	is a proportional odds model used to impute a missing categorical covariate for the HM module.
	
	\begin{figure}[ht]
		\begin{center}
			\includegraphics[width=90mm]{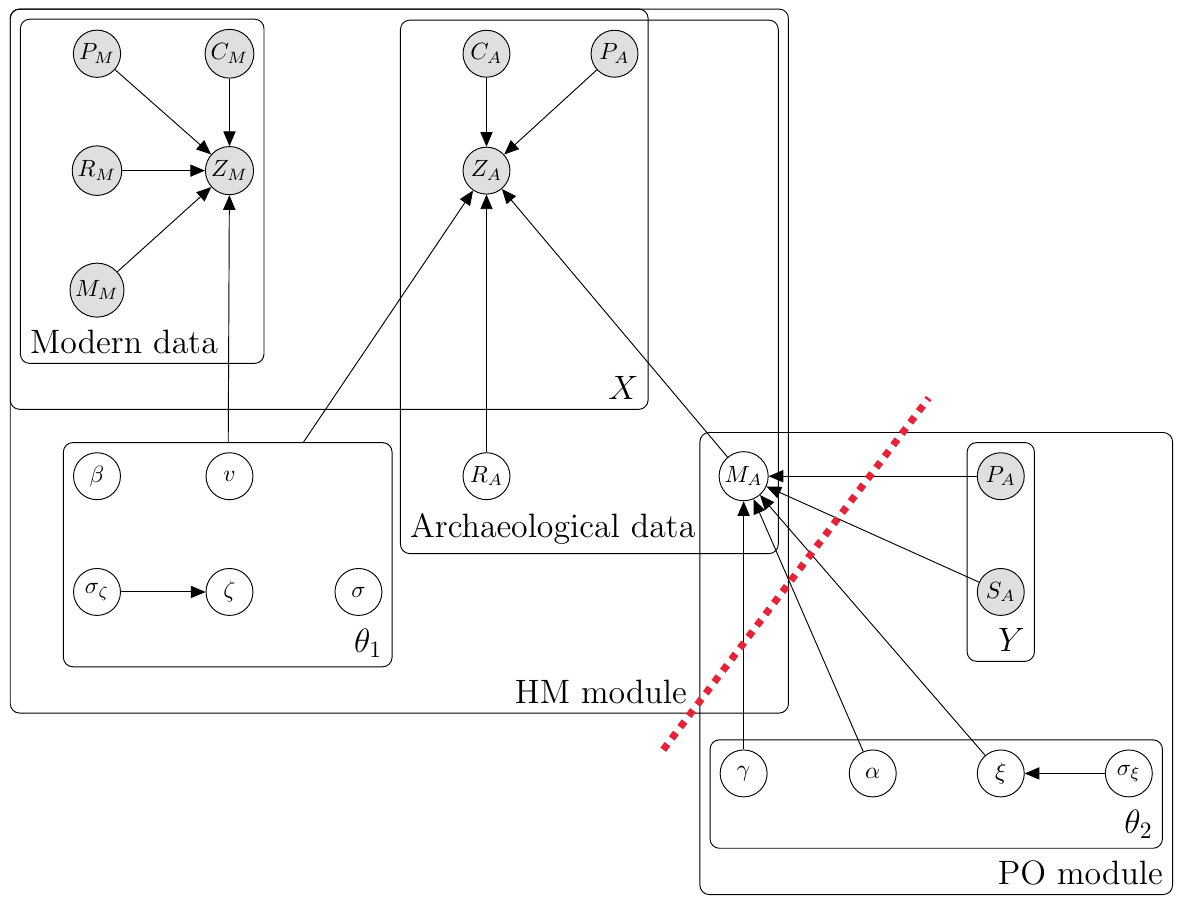}
		\end{center}
		\captionsetup{width=.99\linewidth}
		\caption{\label{archaeology} Graphical representation of model for the agricultural extensification example.}
	\end{figure}
	
	In the HM module's regression, the response is nitrogen level of cereal grains, 
	denoted $Z$.  We follow
	the notation of \cite{yu2021variational} and use subscripts of $A$ and $M$ 
	to denote archaeological and modern values of any variable. 
	So, for example, $Z_A$ and $Z_M$ are nitrogen
	levels of cereal grains for archaeological and modern data respectively.  
	For covariates in the HM module we have crop category $C\in \{\text{Wheat},\text{Barley}\}$, 
	site location $P$ (a categorical variable), site size $S$, rainfall $R$ and manure level
	$M\in \{m_{\text{low}}, m_{\text{medium}},m_{\text{high}}\}$.  
	Archaeological data for rainfall and manure level, $R_A$ and $M_A$, are missing. 
	The HM module is a linear regression model with fixed effects for rainfall and 
	manure level, a random effect for site location, and error variance based on crop category. 
	
	The imputation model in the PO module imputes the missing
	manuring level covariates for the archaeological data with parameters 
	$\theta_2=(\gamma,\alpha,\xi,\sigma_\xi)^\top$.  
	The prior on $M_A$ is a proportional odds model with covariates site size and site location.  
	The parameter $\gamma$ is the site size coefficient; a negative $\gamma$
	supports the extensification hypothesis.  The parameter
	$\xi$ is a vector of random effects for five archaeological site locations in the proportional odds 
	model, $\sigma_\xi$ is the standard deviation of random effects, and $\alpha$ is 
	a vector of two threshold parameters.  
	Further details on the model and priors are available in Appendix B.1 of \cite{yu2021variational}.
	Bayesian modular inference is relevant 
	in this example because 
	the PO module may be poorly specified. Therefore, we can cut feedback so that $M_A$
	is imputed based solely on the hierarchical model for cereal grain nitrogen levels 
	(HM module in Figure \ref{archaeology}), ensuring that 
	imputation of $M_A$ and the interpretation of $\gamma$
	are unaffected by any misspecification in the PO module.  
	
	In Figure \ref{archaeology}, the red line indicates a ``cut'' between the modules. 
	Figure \ref{archaeology-simple} provides a simplified model structure. 
	\begin{figure}[ht]
		\begin{center}
			\includegraphics[width=40mm]{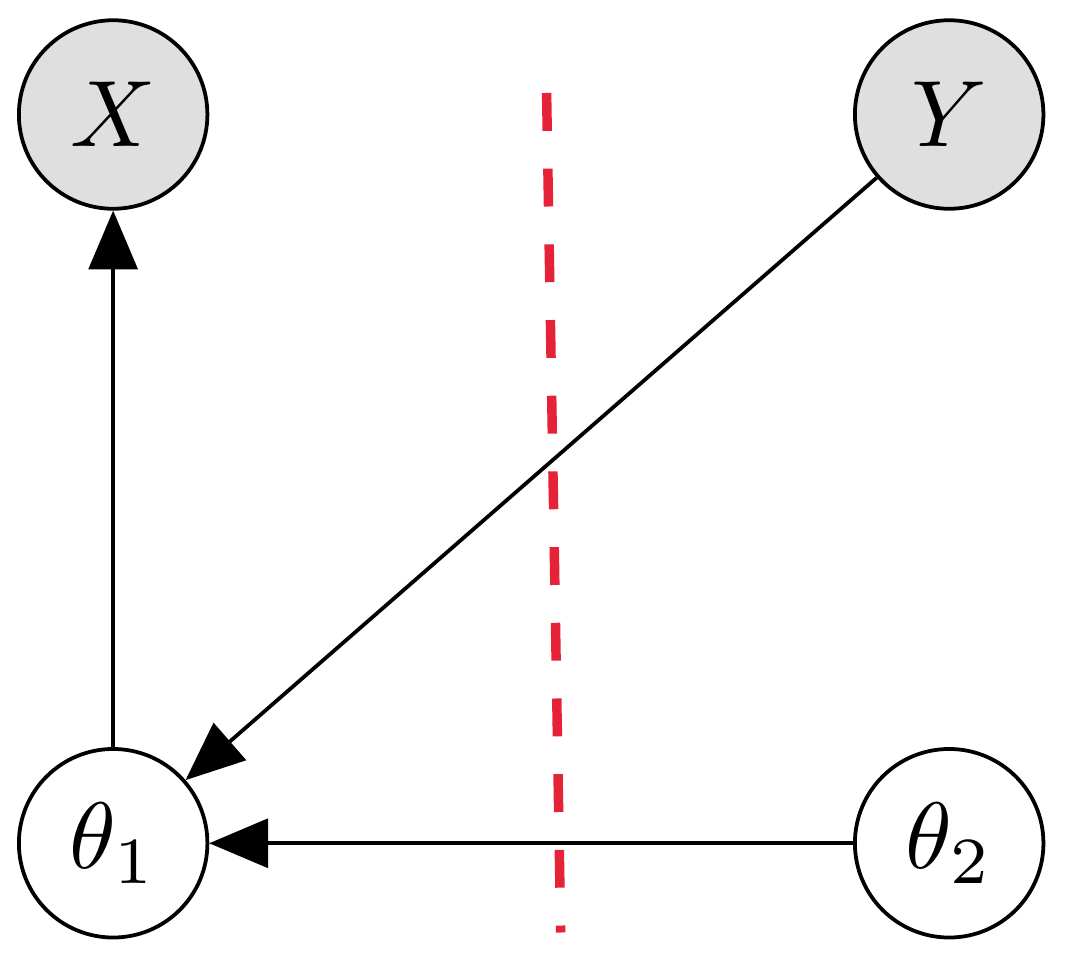}
		\end{center}	  
		\captionsetup{width=.99\linewidth}
		\caption{\label{archaeology-simple} Simplified graphical representation of the model
			for the agricultural extensification example.}
	\end{figure}
	Although it looks 
	different from the two-module system in Figure \ref{two-module}, the cut posterior has the same form, 
	allowing cut and semi-modular inference to proceed similarly.
	Here we consider how semi-modular inference changes according to the choice
	of the loss function.  For a given scalar parameter $\tau$, we will consider a S-SMP posterior
	using mixing weight 
	$$\widetilde{\omega}_+=\min\{1,\widetilde{\omega}(\tau)\},\;\;\;
	\widetilde{\omega}(\tau)=\frac{\sigma_{\tau,\text{cut}}^2-\sigma_{\tau,\full}^2}{(\overline{\tau}_{\text{cut}}-\overline{\tau}_{\full})^2}\mathbb{I}(\sigma_{\tau,\text{cut}}^2-\sigma_{\tau,\full}^2>0),$$
	where $\sigma_{\tau,\text{cut}}^2$ and $\sigma_{\tau,\full}^2$ are the cut and full marginal
	posterior variances for $\tau$, and $\overline{\tau}_{\text{cut}}$ and $\overline{\tau}_{\full}$ are 
	the cut and full marginal posterior means for $\tau$.  This is similar to the S-SMP in Section 3.2, 
	but based on the marginal full and cut posteriors for $\tau$.   
	
	Figure \ref{smi-archaeology} shows the cut, full and S-SMP posteriors for $\gamma$ (the parameter
	of main interest) and the proportional odds regression random effects $\xi_1,\dots, \xi_5$.  
	When basing the mixing weight on $\gamma$ (top left), we do a full cut, while basing the
	mixing weights on $\xi_1,\dots, \xi_5$ results in mixing weights varying between $0.2$ and $1$.     
	In this example, 
	the shrinkage weight can vary a great deal depending on what scalar parameter is being
	targeted in the loss function, and so the use of an appropriately defined loss function for
	the application is crucial.  The S-SMP for $\gamma$ shows weaker evidence for the extensification hypothesis than the standard posterior, in the sense that the posterior probabililty of $\gamma<0$ is smaller. 
	Details of the MCMC approach for generating samples from the cut posterior, as well as an
	SMC method to generate samples from the full posterior, are given in \cite{yu2021variational}.
	\begin{figure}[ht]
		\begin{center}
			\includegraphics[width=120mm]{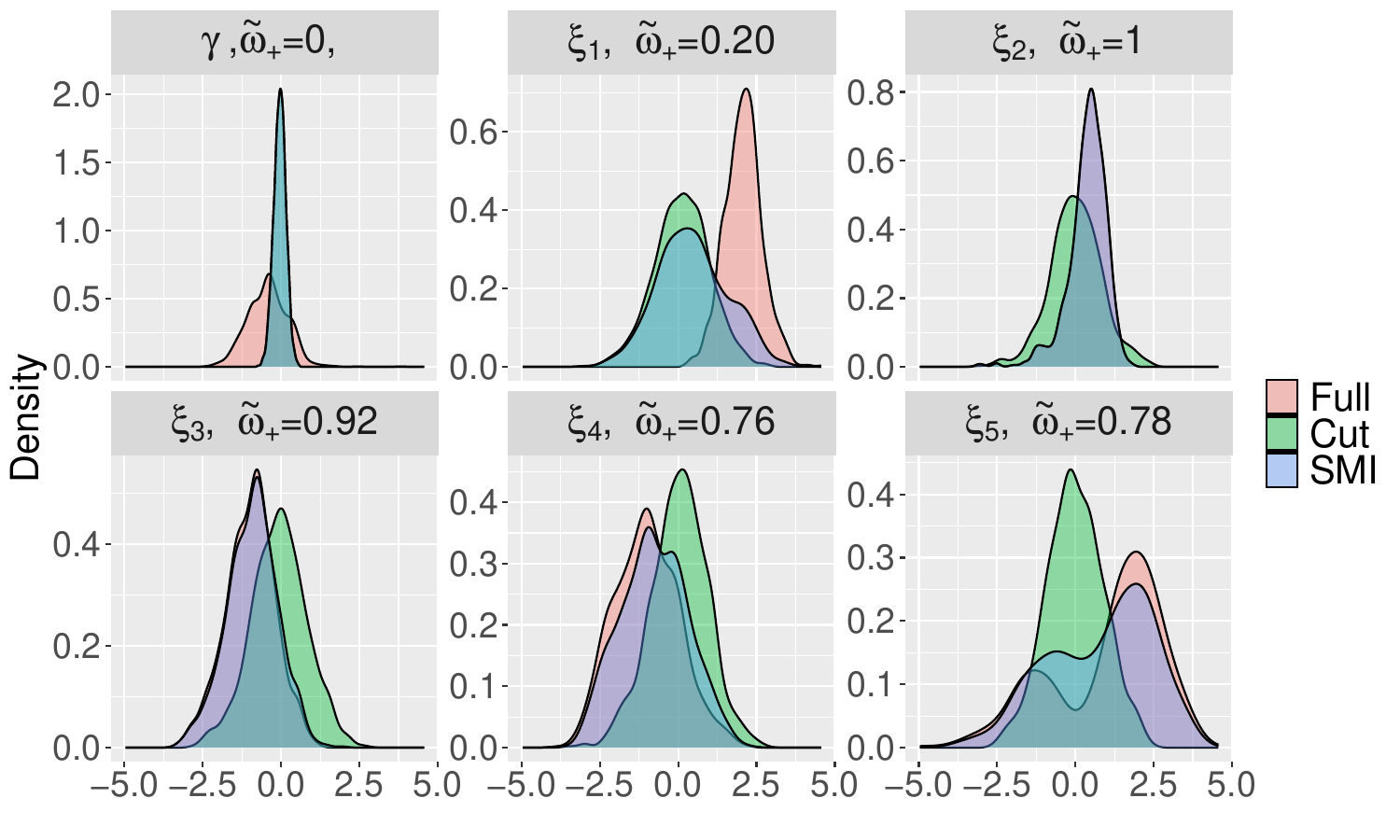}
		\end{center}
		\captionsetup{width=.99\linewidth}
		\caption{\label{smi-archaeology} Marginal cut, conventional and S-SMP posteriors for 
			different parameters in the proportional odds regression model for the archaeological example.
			The title for each graph shows the parameter used, and the S-SMP mixing weight used for
			that parameter.}
	\end{figure}
}

\section{Discussion}

{\color{black}Choosing between the cut and full posteriors is difficult when model misspecification is not severe and in such cases
	semi-modular posteriors (SMPs) proposed in \cite{carmona2020semi} are an 
	attractive alternative.  While SMPs are motivated by the presence of a bias-variance trade-off between cut and full posterior inferences, this paper is the first to formalize
	the existence of such a trade-off. Using SMPs based on linear opinion pooling, we devise a novel pooling weight that allows the SMP to leverage this bias-variance trade-off. Our proposed shrinkage SMP is simple to implement and possess useful theoretical guarantees that other SMP approaches do not posses: the posterior risk of our shrinkage SMP is dominated by that of the cut posterior, and, under certain conditions, that of the full posterior. An interesting future direction would be to determine if our theoretical results can be extended to other types of SMPs, such as those of \cite{carmona2020semi} and \cite{nicholls+lwc22}.}

{\color{black}As suggested by a referee,  the notion of asymptotic risk we consider is only one criterion with which to judge the accuracy of posterior inferences, with posterior predictive accuracy and the validity of posterior credible sets being alternative measures. While assessing the accuracy of different methods based on posterior predictive accuracy is empirically feasible, it is not obvious that it is feasible to deduce a ranking across different modular Bayesian methods under our maintained assumptions. Further, the random weighting of the SMP ensures that determining the asymptotic shape of this posterior, and thus the behavior of its credible sets, is not straightforward. We leave these interesting topics for future study.}

\subsection*{Acknowledgments}
David T. Frazier gratefully acknowledges support by the Australian Research Council through grant DE200101070. 
David Nott's research was supported by the Ministry of Education, Singapore, under the Academic Research Fund Tier 2 (MOE-T2EP20123-0009). We thank seminary participants at the Weierstrass Institute
for Applied Analysis and Stochastics, and the Computational methods for unifying multiple statistical analyses (Fusion) workshop for helpful comments. In addition, we thank Pierre Jacob for helpful comments on some of the stated results. The authors also thank the associate editor and referees for very helpful comments that significantly improved the paper.

\spacingset{1.0}
{\footnotesize
	\bibliographystyle{chicago}
	\bibliography{mod_bib_semi}
}
\appendix
\section{Supplementary Material}
This supplementary material contains the regularity conditions used to obtain the results in the main text, proofs of all stated results and several lemmas used to prove the main results. The regularity conditions and proofs are broken up into two sections that depend on whether the analysis is conducted under gross model misspecification (Assumption \ref{ass:missgross}), or local model misspecification (Assumption \ref{ass:miss}). In addition, this material contains further
details of the HPV and cervical cancer incidence example introduced in Section 2.1 of the main text, and additional experiments for the biased means example in Section 3.2.

\section{Gross Misspecification: Assumption \ref{ass:missgross}}\label{app:G}
\setcounter{assumption}{0}
\renewcommand{\theassumption}{\Alph{section}\arabic{assumption}}

\subsection{Regularity Conditions}\label{app:RegG}
The regularity conditions used to prove Lemma \ref{lemma:gross_miss} are similar to those used to deduce posterior concentration rates in generalized Bayesian methods; see, e.g., \cite{shen2001rates}, as well as \cite{miller2021asymptotic}. We state the assumptions separately for the cut posterior and the full posterior.     Recall, $\ell_{p}(\theta_1)=\log f_1(\z\mid\theta_1)$, and rewrite the cut posterior for $\theta_1$ as 
$$
\pi_{\cut}(\theta_1\mid \z)=\frac{\exp\{\ell_p(\theta_1)\}\pi(\theta_1)}{\int_{\Theta_1}\exp\{\ell_p(\theta_1)\}\pi(\theta_1)\dt\theta_1}
=\frac{\exp\{\ell_p(\theta_1)-\ell_p(\theta_{1,0})\}\pi(\theta_1)}{\int_{\Theta_1}\exp\{\ell_p(\theta_1)-\ell_p(\theta_{1,0})\}\pi(\theta_1)\dt\theta_1}.
$$For $A\subseteq\Theta_1$, write $M_{1,n}(A)=\int_A \exp\{\ell_p(\theta_1)-\ell_p(\theta_{1,0})\}\pi(\theta_1)\dt\theta_1$ so that 
$
\Pi_{\cut}(\theta_1\in A\mid \z)=\frac{M_{1,n}(A)}{M_{1,n}(\Theta)}
$. For $\varepsilon>0$, define $\Theta_{1}(\sqrt{n\varepsilon^2}):=\{\theta_1\in\Theta_1:\|\theta_1-\theta_{1,0}\|\le \sqrt{n\varepsilon^2}\}$. 
\begin{assumption}\label{ass:cut_gross}The following are satisfied. 
\begin{enumerate}
	\item[(i)]  For any $\delta>0$ there exists $\varepsilon>0$ and a sufficiently large $K>0$ such that 
	$$
	P^{(n)}_0\left[\sup_{\|\theta_1-\theta_{1,0}\|\ge\delta}\frac{1}{n}\left\{\ell_p(\theta_1)-\ell_p(\theta_{1,0})\right\}\ge -K\varepsilon\right]=o(1).
	$$
	\item[(ii)] For all $n\ge1$, $\varepsilon>0$, and $b_n=(1/2)\Pi\{\Theta_{1}(\sqrt{n\varepsilon^2})\}e^{-2n\epsilon^2}$, $$P^{(n)}_0\left\{M_{1,n}(\Theta_{1})\le b_n\right\}\le 2/n\epsilon^2.$$
	\item[(iii)] For $\varepsilon>0$, $c>0$, $\Pi\{\Theta_{1}(\sqrt{n\varepsilon^2})\}\gtrsim \exp\{-c n\varepsilon^2\}$.
	\item[(iv)] For any $\theta_1,\theta_1'\in\Theta_1$, if $\theta_1\ne\theta_1'$, then $ f_1(z\mid\theta_1)\ne  f_1(z\mid\theta_1')$ with positive probability.
\end{enumerate}
\end{assumption}
\begin{remark}\normalfont
Assumptions \ref{ass:cut_gross} parts (i), (iii) are identical to those maintained in \cite{shen2001rates} but for the partial log-likelihood $\ell_p(\theta_1)$, while the bound on the posterior denominator in Assumption \ref{ass:cut_gross}(ii) is maintained to simply the proofs and can be removed at the cost of additional technicalities; for instance, using arguments similar to those of Lemma 1 in \cite{shen2001rates}. 
\end{remark}

From the definition of  $\mathrm{KL}\{h(z\mid\theta_{1,0},\delta_0)\|f(z\mid\theta)\}$, 
\begin{flalign*}
\mathrm{KL}\{h(z\mid\theta_{1,0},\delta_0)\|f(z\mid\theta)\}=&\int_{\mathcal{Z}} \log \frac{f_1(z\mid \theta_{1,0})\delta_0(z)}{f_1(z\mid \theta_1)f_2(z\mid\theta)}f_1(z\mid\theta_{1,0})\delta_0(z)\dt z\\=&\int_{\mathcal{Z}} \log \frac{f_1(z\mid \theta_{1,0})}{f_1(z\mid \theta_1)}f_1(z\mid\theta_{1,0})\delta_0(z)\dt z\\+&\int_{\mathcal{Z}} \log \frac{\delta_0(z)}{f_2(z\mid\theta)}f_1(z\mid\theta_{1,0})\delta_0(z)\dt z.
\end{flalign*}Setting $\theta_1=\theta_{1,0}$ minimizes the first part of $\mathrm{KL}\{h(z\mid\theta_{1,0},\delta_0)\|f(z\mid\theta)\}$, but does not minimize both components. Hence, under Assumption \ref{ass:missgross},  $\theta_\star:=\argmin_{\theta\in\Theta}\mathrm{KL}\{h(z\mid\theta_{1,0},\delta_0)\|f(z\mid\theta)\}$ is the value we would expect the full posterior to concentrate onto as $n\rightarrow\infty$. Thus, conducting joint Bayesian inference  on $\theta$ under Assumption \ref{ass:missgross} results in posteriors for which $\theta_1$ will not concentrate onto $\theta_{1,0}$. To formally prove this result, recall the definition $\ell(\theta)=\log f(\z\mid\theta)$, and consider the following regularity conditions, which are equivalent version of Assumption \ref{ass:cut_gross} but for   $\ell(\theta)$, and $\pi(\theta)$. 

\begin{assumption}\label{ass:exact_gross}The following are satisfied. 
\begin{enumerate}
	\item[(i)]  For any $\delta>0$ there exists $\varepsilon>0$ and a sufficiently large $K>0$ such that, for some $\theta_\star\in\Theta$,  
	$$
	P^{(n)}_0\left[\sup_{\|\theta-\theta_{\star}\|\ge\delta}\frac{1}{n}\left\{\ell(\theta)-\ell(\theta_\star)\right\}\ge -K\varepsilon\right]=o(1).
	$$
	\item[(ii)] For all $n\ge1$, $\varepsilon>0$, $\Theta({\sqrt{n\varepsilon^2}}):=\{\theta\in\Theta:\|\theta-\theta_\star\|\le \sqrt{n\varepsilon^2}\}$, and $b_n=(1/2)\Pi\{\Theta_{}(\sqrt{n\varepsilon^2})\}e^{-2n\varepsilon^2}$, $$P^{(n)}_0\left\{M_{n}(\Theta_{})\le b_n\right\}\le 2/n\varepsilon^2.$$
	\item[(iii)] For $\varepsilon>0$, $c>0$, and  $\Pi\{\Theta({\sqrt{n\varepsilon^2}})\}\gtrsim \exp\{-2c n\varepsilon^2\}$.
	\item[(iv)] For any $\theta,\theta'\in\Theta$, if $\theta\ne\theta'$, then $ f(z\mid\theta)\ne f(z\mid\theta')$ with positive probability.
\end{enumerate}
\end{assumption}

\subsection{Proofs of Main Results: Gross Misspecification}\label{app:proofsG}
\begin{proof}[Proof of Lemma \ref{lemma:gross_miss}]
We prove the two cases separately, starting with the cut posterior. 

\medskip 

\noindent\textbf{Part 1: Cut posterior.} 

\medskip 

\noindent 	For $\varepsilon>0$, recall $\Theta_1(\varepsilon):=\{\theta\in\Theta:\|\theta_1-\theta_{1,0}\|\le\varepsilon \}$, and consider 
\begin{flalign*}
	\Pi_\cut\{\Theta_{1}(\varepsilon)^{c}\mid \z\}&=\int_{\Theta_1(\varepsilon)^c}\frac{\exp\{\ell_p(\theta_1)-\ell_p(\theta_{1,0})\}\pi(\theta_1)}{\int_{\Theta_1}\exp\{\ell_p(\theta_1)-\ell_p(\theta_{1,0})\}\pi(\theta_1)\dt\theta_1}=\frac{M_{1,n}\{\Theta_1(\varepsilon)^c\}}{M_{1,n}\{\Theta_1\}}.
\end{flalign*}	
Apply Assumption \ref{ass:cut_gross}(ii) to see that, with probability at least $2/(\sqrt{n\varepsilon^2})^2$, 
\begin{flalign}
	\Pi_\cut\{\Theta_{1}(\varepsilon)^{c}\mid \z\}&\le\frac{M_{1,n}\{\Theta_{1}(\varepsilon)^{c}\}}{b_n}
	\nonumber  \\&\le 2 \frac{e^{2n\epsilon^2}}{\Pi\{\Theta_{1}(\sqrt{n\varepsilon^2})\}}M_{1,n}\{\Theta_1({\varepsilon})^{c}\}\nonumber\\&\lesssim  {e^{4n\epsilon^2}}M_{1,n}\{\Theta_1({\varepsilon})^{c}\}\label{eq:brackets},
\end{flalign}where the last term follows by Assumption \ref{ass:cut_gross}(iii).

Focus on the term in brackets in \eqref{eq:brackets}. Since $\Theta_{1}(\varepsilon)$ is bounded for any finite  $n$, the log-likelihood ratio $\ell_p(\theta_1)-\ell_p(\theta_{1,0})$ is also bounded for any $n$, and so by Assumption \ref{ass:cut_gross}(i),
\begin{flalign*}
	M_{1,n}\{\Theta_{1}(\varepsilon)^{c}\}&=\int \mathrm{1}\left\{\Theta_{1}(\varepsilon)^{c}\right\}\exp\{\ell_p(\theta_1)-\ell_p(\theta_{1,0})\}\pi(\theta_1)\dt\theta_1\\&\le \exp\{-nK\varepsilon^2\}\Pi\{\Theta_{1}(\varepsilon)^{c}\}\\&\le \exp\{-nK\varepsilon^2\},
\end{flalign*}with probability converging to one (since $\Pi(A)\le 1$ for all $A\subseteq\Theta$). 

Placing this bound into equation \eqref{eq:brackets}, and taking $K=8c$, we obtain 
\begin{flalign*}
	\Pi_\cut\{\theta\in \Theta_{1}(\varepsilon)^{c}\mid \z\}&\lesssim  {e^{4n\epsilon^2}}\left[M_{1,n}\{\Theta_1({\varepsilon})^{c}\}\right]\\&\lesssim \exp\{c4n\epsilon^2-cKn\varepsilon^2\}\\&=\exp\{-4cn\epsilon^2\}.
\end{flalign*}For any $\varepsilon\le\log(n)/\sqrt{n}$, the stated result follows. 

\medskip

\noindent\textbf{Part 2: Exact posterior.} Repeating similar arguments to those above, but for the set $\Theta(\varepsilon):=\{\theta\in\Theta:\|\theta-\theta_\star\|\le\varepsilon\}$, proves that, with probability converging to one 
$$
\Pi_{\full}\{\theta\in \Theta_\star(\varepsilon)\mid \z\}\gtrsim 1-\exp\{-c4n\varepsilon^2\}.
$$However, defining $\Theta_{1,\star}(\varepsilon):=\{\theta_1\in\Theta:\|\theta_1-\theta_{1,\star}\|\le\varepsilon\}$, since $\Theta(\varepsilon)\subset \Theta_{1,\star}(\varepsilon)\times\Theta_2$, it follows that
\begin{flalign*}
	1-C\exp\{-c4n\varepsilon^2\}\le \Pi_{\full}\{\theta\in 
	\Theta(\varepsilon)\mid \z\}&=\int_{ \Theta(\varepsilon)}\pi_{\full}(\theta_1,\theta_2\mid \z)\dt\theta_2\dt\theta_1\\&\le  \int_{ \Theta_{1,\star}(\varepsilon)}\int_{\Theta_2}\pi_{\full}(\theta_1,\theta_2\mid \z)\dt\theta_2\dt\theta_1\\&\le\int_{ \Theta_{1,\varepsilon}^\star} \pi_{\full}(\theta_1\mid \z)\dt\theta_1.
\end{flalign*}Hence, with probability converging to 1, $\int_{ \Theta_{1,\star}(\varepsilon)} \pi_{\full}(\theta_1\mid \z)\dt\theta_1\rightarrow 1$. Since $\theta_{1,0}\ne\theta_{1,\star}$ under Assumption \ref{ass:missgross}, there exists some $\varepsilon>0$, such that $\theta_{1,0}\not\in \Theta_{1,\star}(\varepsilon)$, so that for any $\widetilde\varepsilon\le\varepsilon$, $\int_{ \Theta_{1}(\widetilde\varepsilon)} \pi(\theta_1\mid \z)\dt\theta_1\rightarrow 0$ in probability. 
\end{proof}

\begin{proof}[Proof of Corollary \ref{corr:gross}]
Write the SMP as 
\begin{flalign}
	\pi_\omega(\theta\mid\z)&=\pi_{\cut}(\theta\mid \z)+\omega\{\pi_{\full}(\theta\mid\z)-\pi_{\cut}(\theta\mid\z)\}\nonumber\\&=\pi_{\cut}(\theta\mid \z)+\omega\{\pi_{\full}(\theta_1\mid \z)-\pi_{\cut}(\theta_1\mid \z)\}\pi(\theta_2\mid \theta_1, \z)\label{eq:eq_cut}
\end{flalign}where $\pi_{\full}(\theta_1\mid\z)=\int_{\Theta_2}\pi_{\full}(\theta\mid\z)\dt\theta_2$. Apply \eqref{eq:eq_cut}, and Fubini,  to obtain 
\begin{flalign*}
	\int_\Theta|\pi_{\widehat\omega_+}(\theta\mid \z)-\pi_{\cut}(\theta\mid \z)|\dt\theta
	&=\int_\Theta|\widehat\omega_+\{\pi_{\full}(\theta_1\mid \z)-\pi_{\cut}(\theta_1\mid \z)\}\pi(\theta_2\mid\theta_1, \z)|\dt\theta\\&=\int_{\Theta_1}\int_{\Theta_2}|\widehat\omega_+\{\pi_{\full}(\theta_1\mid \z)-\pi_{\cut}(\theta_1\mid \z)\}|\pi(\theta_2\mid\theta_1, \z)\dt\theta\\&=\widehat\omega_+\int_{\Theta_1}|\pi_{\full}(\theta_1\mid \z)-\pi_{\cut}(\theta_1\mid \z)|\dt\theta_1.
\end{flalign*}We can then write
\begin{flalign*}
	\int_\Theta|\pi_{\widehat\omega_+}(\theta\mid \z)-\pi_{\cut}(\theta\mid \z)|\dt\theta \le \widehat\omega_+\left[\int_{\Theta_1}\pi_{\full}(\theta_1\mid \z)\dt\theta_1+\int_{\Theta_1}\pi_\cut(\theta_1\mid \z)\dt\theta_1\right]=2\widehat\omega_+.
\end{flalign*}The stated result now follows if $\widehat\omega_+=o_p(1)$ as $n\rightarrow+\infty$.

To show this, let $X_{n,\cut}:=\sqrt{n}(\overline\theta_{\cut}-\theta_{0})$, $X_{n,\full}:=\sqrt{n}(\overline\theta_{\full}-\theta_{\star})$, $X_n:=X_{\full,n}-X_{\cut,n}$ and $Y_n=\sqrt{n}(\theta_{0}-\theta_{\star})$. Then, for ${\Upsilon_n}=\Upsilon(\overline\theta_\cut)$,
$$
\widehat{\omega}=\frac{\gamma_n}{n\|\overline\theta_{\cut}-\overline\theta_{\full}\|^2_{\Upsilon_n}}
=\frac{\gamma_n}{\|X_{n.\cut}-X_{n,\full}+Y_n\|^2_{\Upsilon_n}}=\frac{\gamma_n}{\|Y_n-X_n\|^2_{\Upsilon_n}}.
$$ By the reverse triangle inequality
\begin{flalign*}
	\widehat\omega \le \frac{\gamma_n}{|\|X_n\|_{\Upsilon_n}^2-\|Y_n\|_{\Upsilon_n}^2|}.
\end{flalign*}By the hypothesis of the result, $\|X_{n}\|=O_p(1)$, while under Assumption \ref{ass:missgross}, $\theta_{1,0}\ne\theta_{1,\star}$, so that $\|Y_n\|\rightarrow+\infty$ as $n\rightarrow+\infty$, and the stated result follows. 
\end{proof}

\section{Local Misspecification: Assumption \ref{ass:miss}}

\subsection{Regularity Conditions: Local Misspecification}\label{app:A}

\setcounter{assumption}{0}
\renewcommand{\theassumption}{\Alph{section}\arabic{assumption}}
Before stating the regularity conditions we maintain in this section, we recall several notations previously defined in the main text. Let $\ell(\theta)=\log f(\z\mid\theta)$, and denote the joint log-likelihood for the $i$-th observation as $\ell(z_i\mid\theta)=\log f(z_i\mid\theta)$.  Denote the full derivative of the log-likelihood as $\dot\ell(\theta):=\partial\ell(\theta)/\partial\theta$, and denote the second derivative as $\ddot\ell(\theta):=\partial^2\ell(\theta)/\partial\theta\partial\theta^\top$. For $j,k\in\{1,2\}$, define the partial derivatives $\dot\ell_{(j)}(\theta)=\partial \ell(\theta)/\partial\theta_j$, and the second partial derivatives $\ddot\ell_{(jk)}(\theta)=\partial^2 \ell(\theta)/\partial\theta_j\partial\theta_k^\top$. For a function $g:\mathcal{Z}\rightarrow \mathbb{R}^{d}$, let $\E_{n}[g(z)]$ denote the expectation of $g(z)$ under $h(z\mid \theta_0,\delta_n)$ in Assumption \ref{ass:miss}; i.e., $\E_n[g(z)]=\int_{\mathcal{Z}}g(z)h(z\mid\theta_0,\delta_n)\dt \mu(z)$. Define the matrices $\mathcal{I}_{}:=-\lim_nn^{-1}\E_{n}[\ddot\ell_{}(\theta_0)]$, and $\mathcal{I}_{jk}:=-\lim_nn^{-1}\E_{n}[\ddot\ell_{(jk)}(\theta_0)]$. Recall that $\eta=(\eta_1^\top,\eta_2^\top)^\top$ in Assumption \ref{ass:miss} is partitioned conformably
with $\theta=(\theta_1^\top,\theta_2^\top)^\top$.

In addition,  note that
$$\ell(\theta)=\log f_1(z_{1:n}\mid\theta_1)+\log f_2(z_{1:n}\mid\theta)=\ell_p(\theta_1)+\ell_c(\theta),$$ 
where $\ell_p(\theta_1):=\log f_1(z_{1:n}\mid\theta_1)$ signifies the `partial log-likelihood' term, and $\ell_c(\theta):=\log f_2(z_{1:n}\mid\theta_1,\theta_2)$ signifies the log-likelihood term that is used in cut inference to construct the conditional posterior for $\theta_2$ given $\theta_1$. Define the partial derivatives of $\ell_p(\theta)$ as $\dot\ell_{p(1)}(\theta_1)=\partial \ell_{p}(\theta_1)/\partial\theta_1$, $\ddot\ell_{p(11)}(\theta_1)=\partial^2 \ell_{p(11)}(\theta_1)/\partial\theta_j\partial\theta_k^\top$, and recall  $\mathcal{I}_{p(11)}:=-\lim_nn^{-1}\E_{n}[\ddot\ell_{p(11)}(\theta_{1,0})]$. For $\ell_c(\theta)$ and $j,k\in\{1,2\}$, define $\dot\ell_{c(jk)}(\theta):=\partial \ell_{c(j)}(\theta)/\partial\theta_j$ and  $\ddot\ell_{c(jk)}(\theta):=\partial^2 \ell_{c(jk)}(\theta)/\partial\theta_j\partial\theta_k^\top$. From the structure of the log-likelihood $\ell(\theta)$, note that
\begin{flalign*}
\mathcal{I}_{12}&:=-\lim_{n\rightarrow\infty}n^{-1}\E_{n}[\ddot\ell_{(12)}(\theta_0)]=-\lim_{n\rightarrow\infty}n^{-1}\E_{n}[\ddot\ell_{c(12)}(\theta_0)],\\ \mathcal{I}_{21}&:=-\lim_{n\rightarrow\infty}n^{-1}\E_{n}[\ddot\ell_{(21)}(\theta_0)]=-\lim_{n\rightarrow\infty}n^{-1}\E_{n}[\ddot\ell_{c(21)}(\theta_0)],\\ \mathcal{I}_{22}&:=-\lim_{n\rightarrow\infty}n^{-1}\E_{n}[\ddot\ell_{(22)}(\theta_0)]=-\lim_{n\rightarrow\infty}n^{-1}\E_{n}[\ddot\ell_{c(22)}(\theta_0)].
\end{flalign*}

To formalize the impact of the misspecification in Assumption \ref{ass:miss}, we impose the following regularity conditions on the  density $h(z\mid \theta,\delta_n)$.\footnote{We eschew measurability conditions and assume that all objects written are measurable.}
\begin{assumption}\label{ass:regular}
For $\upsilon:=(\theta^\top,\psi^\top)^\top$, let $\upsilon_0:=(\theta_0^\top,\psi_0^\top)^\top$, be elements of the interior of $\Theta\times\Delta$, where $\Delta\subset\mathbb{R}^{d}$ and compact. The function $h_n(z\mid \upsilon)=f_1(z\mid\theta)f_2(z\mid\theta)\{1+\psi^\top\zeta(z)/\sqrt{n}\}$ is twice continuously differentiable in $v$ for almost all $z\in\mathcal{Z}$. There exist positive functions $a(z)$, $b(z)$ such that for all $z\in\mathcal{Z}$, except on sets of measure zero, $\ell_n(z\mid\upsilon)=\log h_n(z\mid\upsilon)$, satisfies the following: $\exp\ell_n(z\mid\upsilon)\le a(z)$, and  for all $\|\upsilon-\upsilon_0\|\le \nu_0/\sqrt{n}$,  and some $\nu_0>0$, each of the following $|\ell_n(z\mid\upsilon)|$, $\|\dot \ell_n(z\mid\upsilon)\|^2$, $\|\ddot\ell_n(z\mid\upsilon)\|^2$ are less than $b(z)$.  Further, $\E_n[a(z)],\E_n[b(z)], \E_n[a(z)b(z)]<+\infty$, and the set $\{z\in\mathcal{Z}:h_n(z\mid \upsilon)>0\}$ does not depend on $\upsilon$.
\end{assumption}
In addition, we maintain the following assumptions about the assumed model $f(z\mid \theta)$.

\begin{assumption}\label{ass:regular2}
For any $\theta,\theta'\in\Theta$, if $\theta\ne\theta'$, then $\ell(\theta)\ne\ell(\theta')$ with positive probability. 
\end{assumption}

\begin{remark}\normalfont
Assumption \ref{ass:regular} is similar to the regularity conditions employed by \cite{claeskens2003focused} to deduce  large sample theory for frequentist model averaging estimators.  Assumption \ref{ass:regular2} is a classical identification condition.
\end{remark}
\numberwithin{lemma}{section}
\numberwithin{corollary}{section}
\numberwithin{theorem}{section}

\subsection{Preliminary Results}\label{sec:prelim_results}

The following intermediate results are used to state and prove our main results. 
\begin{lemma}\label{lem:matrix_equal1}
If Assumptions \ref{ass:miss}, \ref{ass:regular} and \ref{ass:regular2} are satisfied, then the following results are satisfied.
\begin{enumerate}
	\item $\lim_n n^{-1}\E_{n}[-\ddot\ell_{p(11)}(\theta_{1,0})]=\lim_n n^{-1}\E_{n}[\dot\ell_{p(1)}(\theta_{1,0})\dot\ell_{p(1)}(\theta_{1,0})^\top ]$.
	\item $\lim_n n^{-1}\E_{n}[-\ddot\ell_{c(11)}(\theta_{0})]=\lim_n n^{-1}\E_{n}[\dot\ell_{ c(1)}(\theta_{0})\dot\ell_{c(1)}(\theta_{0})^\top]$.
	\item $\lim_n n^{-1}\E_{n}[\dot\ell_{p(1)}(\theta_{1,0})\dot\ell_{c(1)}(\theta_0)^\top]=0$.
	\item $\lim_n n^{-1}\E_{n}[\dot\ell_{p(1)}(\theta_{1,0})\dot\ell_{c(2)}(\theta_{0})^\top]=0$.
\end{enumerate}
\end{lemma}
The following Lemma is used to prove the results under the drifting sequences of DGPs constructed in Assumption \ref{ass:miss}. 

\begin{lemma}\label{lem:newey}
Assumptions \ref{ass:miss}, \ref{ass:regular}, and \ref{ass:regular2}  are satisfied. 
Then $\phi(\theta, \delta)=\int_{\mathcal{Z}}\dot\ell(z\mid\theta) h(z\mid\theta,\delta)\dt\mu(z)$ exists and is continuous on $\Theta \times \Delta$, and for all $\varepsilon>0$, and any compact $\tilde\Theta\subset\Theta$,
$$\quad \lim_{n\rightarrow+\infty} \operatorname{Pr}\left(\sup _{\theta\in\tilde\Theta}\left|n^{-1}\dot\ell(\theta)-\phi\left(\theta, \delta_n\right)\right| \geqslant \varepsilon\right)=0.$$
\end{lemma}
The following result is a consequence of Lemmas \ref{lem:matrix_equal1} and  \ref{lem:newey}. 
\begin{lemma}\label{lem:CLTmiss}
If Assumptions \ref{ass:miss}, \ref{ass:regular} and \ref{ass:regular2} are satisfied, then the following results are satisfied.
\begin{enumerate}
	\item  $\lim_n \E_n\left[\dot\ell(\theta_0)/\sqrt{n}\right]=\eta$.
	\item $\dot\ell(\theta_0)/\sqrt{n}\Rightarrow N(\eta,\mathcal{I})$.
	\item $\dot\ell_{p(1)}(\theta_0)/\sqrt{n}\Rightarrow N(0,\mathcal{I}_{p(11)})$.
	\item $\dot\ell_{c(2)}(\theta_0)/\sqrt{n}\Rightarrow N(\eta_2,\mathcal{I}_{(22)}).$
\end{enumerate}
\end{lemma}
The following is a useful extension of Stein's Lemma. 
\begin{lemma}[Lemma 2 of \citealp{hansen2016efficient}]\label{lem:stein}
If $\xi\sim N(0,V)$ is an $m\times1$ vector, and $\Psi$ is $m\times m$ matrix, for $\varphi(X):\mathbb{R}^m\rightarrow\mathbb{R}^m$ continuously differentiable, then for $h\in\mathbb{R}^m$,
\begin{equation}
	\mathbb{E}\left\{\varphi(\xi+{h})^{\top} \Psi\xi\right\}=\mathbb{E} \tr\left\{\frac{\partial}{\partial  {x}} \varphi(\xi+ {h})^{\top} \Psi V\right\}.
\end{equation}
\end{lemma}

To simplify proofs of our results, we use the following block-matrix notations: for $j,k\in\{1,2\}$,  let $\mathbf{0}_{d_{j}\times d_{k}}$ denote a $d_{j}\times d_{k}$ matrix of zeros, and define 
\begin{equation*}
\begin{aligned}
	\Gamma_{1,\cut}&:=(\mathcal{I}_{p(11)}^{-1}:\mathbf{0}_{d_{1}\times d_{1}}:\mathbf{0}_{d_{1}\times d_{2}}),\quad& \Gamma_{1,\full}:=(\mathbf{0}_{d_{1}\times d_{1}}:\mathcal{I}_{11.2}^{-1}:-\mathcal{I}_{11.2}^{-1}\mathcal{I}_{12}\mathcal{I}_{22}^{-1}),\\
	\Gamma_{2,\cut}&:=(-\I_{22}^{-1}\I_{21}\mathcal{I}_{p(11)}^{-1}:\mathbf{O}_{d_2\times d_1}:\I_{22}^{-1}),\quad&\Gamma_{2,\full}:=(\mathbf{O}_{d_2\times d_1}:-\I_{22}^{-1}\I_{21}\mathcal{I}_{11.2}^{-1}:\I_{22}^{-1}).\\			
\end{aligned}
\end{equation*}
In addition, define 
\begin{flalign*}
Z_n:=\frac{1}{\sqrt{n}}\begin{pmatrix}
	\dot\ell_{p(1)}(\theta_{1,0})\\\dot\ell_{p(1)}(\theta_{1,0})+\dot\ell_{c(1)}(\theta_{0})\\\dot\ell_{c(2)}(\theta_{1,0})\end{pmatrix},\quad	
\Gamma_{\cut}:=\begin{pmatrix}
	\Gamma_{1,\cut}\\\Gamma_{2,\cut}
\end{pmatrix},\quad\Gamma_{\full}:=\begin{pmatrix}
	\Gamma_{1,\full}\\\Gamma_{2,\full}
\end{pmatrix},	
\end{flalign*}
where we note that $\Gamma_{\cut}$ and $\Gamma_{\full}$ have dimension $(d_1+d_2)\times(2d_1+d_2)$. Recall
$$
W=(\I_{p(11)}^{-1}-\I_{11.2}^{-1}),\quad \mathcal{M}=\begin{pmatrix}W&-W\I_{12}\I_{22}^{-1}
\\-\I_{22}^{-1}\I_{21}	W&\I_{22}^{-1}\I_{21}W\I_{12}\I_{22}^{-1}
\end{pmatrix}.
$$ 
When no confusion is likely to result, when terms are evaluated at $\theta_0$, we drop their dependence on this value; e.g., for $j,k\in\{1,2\}$ we will often write $\ell_{(jk)}=\ell_{(jk)}(\theta_0)$.

\begin{lemma}\label{lem:matrix_equal2}Under Assumptions \ref{ass:miss}, \ref{ass:regular} and \ref{ass:regular2}, the following are satisfied. 
\begin{enumerate}
	\item $Z_n\Rightarrow\xi+\tau$, where $\tau=(\mathbf{0}_{d_{1}\times1}^\top,\eta^\top)^\top$, and where  $\xi$ is a $2d_{1}+d_{2}$ dimensional normal random variable with mean zero and variance
	\begin{flalign*}
		\Omega&:=\lim_{n\rightarrow+\infty}n^{-1}\E_n[Z_nZ_n^\top]\\&=\lim_{n\rightarrow+\infty}n^{-1}\begin{pmatrix}\E_n[\dot\ell_{p(1)}\dot\ell_{p(1)}^\top]&\E_n[\dot\ell_{p(1)}\dot\ell_{p(1)}^\top]&\mathbf{0}_{d_{1}\times d_{2}}\\\E_n[\dot\ell_{p(1)}\dot\ell_{p(1)}^\top]&\E_n[\{\dot\ell_{p(1)}+\dot\ell_{c(1)}\}\{\dot\ell_{p(1)}+\dot\ell_{c(1)}\}^\top]&\E_n[\dot\ell_{c(1)}\dot\ell_{c(2)}^\top]\\\mathbf{0}_{d_{2}\times d_{1}} &\E_n[\dot\ell_{c(2)}\dot\ell_{c(1)}^\top] &\E_n[\dot\ell_{c(2)}\dot\ell_{c(2)}^\top]
		\end{pmatrix}\\&=\begin{pmatrix}
			\I_{p(11)}&\I_{p(11)}&\mathbf{0}_{d_{1}\times d_{2}}\\\I_{p(11)}&\I_{11}&\I_{12}\\\mathbf{0}_{d_{2}\times d_{1}}&\I_{21}&\I_{22}
		\end{pmatrix}.
	\end{flalign*}
	\item 	$(\Gamma_{\cut}-\Gamma_{\full})\Omega(\Gamma_{\cut}-\Gamma_{\full})^\top=\mathcal{M}$.
	\item $\Gamma_{\cut}\Omega(\Gamma_{\cut}-\Gamma_{\full})^\top=\mathcal{M}$.	
\end{enumerate}	
\end{lemma}

To analyze the behavior of the S-SMP,
we must first deduce the behavior of $\overline{\theta}_{\cut}$ and $\overline{\theta}_{\full}$. To this end, define the statistics 
$$
T_{n,\mathrm{full}}=\Gamma_{\full}Z_n/\sqrt{n}=\begin{pmatrix}
T_{1,\full,n}\\T_{2,\full,n}
\end{pmatrix},\quad T_{n,\mathrm{cut}}=\Gamma_{\cut}Z_n/\sqrt{n}=\begin{pmatrix}
T_{1,\cut,n}\\T_{2,\cut,n}
\end{pmatrix},
$$ which are partitioned conformably with $\theta=(\theta_1^\top,\theta_2^\top)^\top$. Define
\begin{eqnarray*}
&t:=\sqrt{n}(\theta-\theta_0-T_{n,\mathrm{full}})=(t_1^\top,t_2^\top)^\top,\;\;	\mathcal{T}:=\{t=\sqrt{n}(\theta-\theta_0-T_{n,\mathrm{full}}):\theta\in\Theta\},\\&\vartheta:=\sqrt{n}(\theta-\theta_0-T_{n,\mathrm{cut}})=(\vartheta_1^\top,\vartheta_2^\top)^\top,\;\;\mathcal{V}:=\{\vartheta=\sqrt{n}(\theta-\theta_0-T_{n,\mathrm{cut}}):\theta\in\Theta\}.
\end{eqnarray*}
The cut and full posteriors for $\vartheta$ and $t$ are given by
\begin{eqnarray*}
&\pi_{\cut}(\vartheta\mid\z)=\frac{1}{\sqrt{n}^{d_\theta}}\pi_{\cut}\left(\theta_0+\frac{\vartheta}{\sqrt{n}}+{T_{n,\mathrm{cut}}}{}\mid\z\right),\\
&	\pi_{\mathrm{full}}(t\mid\z)=\frac{1}{\sqrt{n}^{d_\theta}}\pi_{\full}\left(\theta_0+\frac{t}{\sqrt{n}}+T_{n,\mathrm{full}}\mid\z\right).
\end{eqnarray*}

\begin{theorem}\label{thm:exact}
If Assumption \ref{ass:miss} in the main text, and Assumptions \ref{ass:regular} and \ref{ass:regular2} are satisfied, then $\int_{\mathcal{T}}|\pi_{\mathrm{full}}(t\mid\z)-N(t;0,\mathcal{I}^{-1})|\dt t=o_p(1)$ and $\sqrt{n}(\overline\theta_{\mathrm{full}}-\theta_0)\Rightarrow N(\mathcal{I}^{-1}\eta,\mathcal{I}^{-1})$.
\end{theorem}

\begin{theorem}\label{thm:cut}
If Assumption \ref{ass:miss} in the main text, and Assumptions \ref{ass:regular} and \ref{ass:regular2} are satisfied, then
\begin{flalign*}
&\int_{\mathcal{V}_1}|\pi_{\mathrm{cut}}(\vartheta_1\mid\z)-N(\vartheta_1;0,\mathcal{I}_{p(11)}^{-1})|\dt\vartheta_1=o_p(1),\\	&\int_{\mathcal{V}_2}|\pi_{\mathrm{cut}}(\vartheta_2\mid\z)-N(\vartheta_2;0,\mathcal{I}^{-1}_{22}+\mathcal{I}^{-1}_{22}\mathcal{I}_{21}\mathcal{I}_{p(11)}^{-1}\mathcal{I}_{12}\mathcal{I}_{22}^{-1})|\dt\vartheta_2=o_p(1).	
\end{flalign*}				
In addition, 
\begin{flalign*}
&\sqrt{n}(\overline\theta_{1,\mathrm{cut}}-\theta_{1,0})\Rightarrow N(0,\mathcal{I}_{p(11)}^{-1}),\\	
&\sqrt{n}(\overline\theta_{2,\cut}-\theta_{2,0})\Rightarrow N(\mathcal{I}_{(22)}^{-1}\eta_2,\mathcal{I}^{-1}_{22}+\mathcal{I}^{-1}_{22}\mathcal{I}_{21}\mathcal{I}_{p(11)}^{-1}\mathcal{I}_{12}\mathcal{I}_{22}^{-1}).
\end{flalign*}
\end{theorem}
\begin{corollary}\label{corr:bias_compare}
If Assumption \ref{ass:miss} in the main text, and Assumptions \ref{ass:regular} and \ref{ass:regular2} are satisfied, then, for $j\in\{1,2\}$, 
$$
\lim_{n\rightarrow+\infty}\left([\E_{n}\left\{\sqrt{n}(\overline{\theta}_{j,\full}-\theta_{j,0})\right\}]^2-[\E_{n}\left\{\sqrt{n}(\overline{\theta}_{j,\cut}-\theta_{j,0})\right\}]^2\right)\ge 0.
$$	
\end{corollary}

\begin{lemma}\label{lem:bias}If Assumptions \ref{ass:miss}, and \ref{ass:regular}-\ref{ass:regular2} are satisfied, then 
\begin{enumerate}
\item $
\sqrt{n}\begin{pmatrix}
	\overline\theta_{\cut}-\theta_{0}	
	\\	\overline\theta_{\mathrm{full}}-\theta_{0}
\end{pmatrix}=\begin{pmatrix}
	\Gamma_{\cut}\\\Gamma_{\full} 
\end{pmatrix}Z_n+o_p(1)\Rightarrow \begin{pmatrix}\Gamma_{\cut}\\\Gamma_{\full}
\end{pmatrix}(\xi+\tau)
$.
\item $\widehat\omega_+\Rightarrow\overline\omega:=\min\left\{1,\frac{\gamma}{(\xi+\tau)^\top(\Gamma_{\cut}-\Gamma_{\full})^\top\Upsilon(\Gamma_{\cut}-\Gamma_{\full})(\xi+\tau)}\right\}	$, where $\Upsilon=\partial^2 q(\theta,\theta_0)/\partial\theta\partial\theta^\top|_{\theta=\theta_0}$.
\item $\sqrt{n}\{\overline\theta(\widehat{\omega}_{+})-\theta_{0}\}\Rightarrow \Gamma_{\cut}(\xi+\tau)-\overline\omega\{(\Gamma_{\cut}-\Gamma_{\full})(\xi+\tau)\}$.
\end{enumerate}	
\end{lemma}

\begin{remark}\normalfont
We do not explicitly consider that the cut and full posteriors are computed from samples of different sizes; e.g., $n_2$ for the exact and $n_1$ for the cut. While useful, this difference in sample sizes will not have a significant impact on the resulting behavior of the SMP so long as $n_1/n_2\rightarrow \alpha\in(0,\infty)$. If one wishes to impose such a condition, the only result will be a slight change in the definition of the matrix $\Omega$ in Lemma \ref{lem:matrix_equal2} to account for the fact that $\lim_nn_1/n_2\ne1$. As such, all results presented herein can be extended to this case at the cost of minor additional technicalities. To see this, let $n_2$ be the larger sample size associated with the full posterior, and $n_1$ the smaller sample size associated with the cut posterior, which satisfies $\lim_nn_1/n_2= \alpha$ for some $\alpha<1$. Then, our results go through with $n=n_2$ since 
\begin{flalign*}
\sqrt{n_2}(\overline\theta_{\cut}-\theta_{0})&=\alpha^{-1/2}\sqrt{n_1}(\overline\theta_{\cut}-\theta_{0})+o_p\left\{\left(\frac{1}{n_1/n_2}-\frac{1}{\alpha}\right)\|\sqrt{n_1}(\overline\theta_{\cut}-\theta_{0})\|\right\}\\
&=\alpha^{-1/2}\sqrt{n_1}(\overline\theta_{\cut}-\theta_{0})+o_p(1),
\end{flalign*}where the second equality follows from Theorem \ref{thm:cut} (when $\pi_\cut(\theta\mid \z)$ is based on $n_1$ observations).
\end{remark}

\subsection{Proofs of Main Results}\label{app:proofs}

Recall the definitions $\I_{11.2}=\I_{11}-\I_{12}\I_{22}^{-1}\I_{21}$, and 
\begin{flalign}
W=\I_{p(11)}^{-1}-\I_{11.2}^{-1},\quad \mathcal{M}=\begin{pmatrix}
W&-W\I_{12}\I_{22}^{-1}\\-\I_{22}^{-1}\I_{21}W&=\I_{22}^{-1}\I_{21}W\I_{12}\I_{22}^{-1}
\end{pmatrix}.\label{eq:m_eq}
\end{flalign}
To prove Theorem \ref{corr:bound} in the main text, we first prove the following result. 
\begin{theorem}\label{thm:ests}
Consider that Assumptions \ref{ass:miss}-\ref{ass:loss}, and the regularity conditions in Assumptions \ref{ass:regular}-\ref{ass:regular2} are satisfied. 	If $\tr \Upsilon\mathcal{M}\ge2\|\Upsilon\mathcal{M}\|$, and $0\le\gamma\le 2(\tr \Upsilon\mathcal{M}-2\|\Upsilon\mathcal{M}\|)$, then
$$
\mathrm{R}_{q}(\pi_{\widehat\omega_{+}},\theta_{0})\le \mathrm{R}_{q}(\pi_{\cut},\theta_{0})-\E_{}\left[\frac{\gamma\{2(\tr \Upsilon\mathcal{M}-2\|\Upsilon\mathcal{M}\|)-\gamma
\}}{(\xi+\tau)^\top(\Gamma_{\cut}-\Gamma_{\full})^\top{\Upsilon_{}}(\Gamma_{\cut}-\Gamma_{\full})(\xi+\tau)}\right].
$$
\end{theorem} 
\begin{proof}[Proof of Theorem \ref{thm:ests}.]
For $\Upsilon=[\partial^2 q(\delta,\theta_{0})/\partial\delta\partial\delta^\top]_{\delta=\theta_{0}}$, and $\|X\|_\Upsilon^2=X^\top\Upsilon X$, following arguments similar to those in Theorem 1 of \cite{rousseau1997asymptotic}, it can be shown that 
\begin{flalign*}
\mathrm{R}_q(\pi_{\cut},\theta_0)&:=	\lim_{\nu\rightarrow+\infty}\liminf_{n\rightarrow+\infty}\E_n \min\left\{\int_\Theta nq(\theta,\theta_{0})\pi_{\cut}(\theta\mid  \z)\dt\theta,\nu\right\}\\&=\lim_{\nu\rightarrow+\infty}\liminf_{n\rightarrow+\infty}\E_n\min\left\{nq(\overline\theta_\cut,\theta_0),\nu\right\}\\&=\lim_{\nu\rightarrow+\infty}\liminf_{n\rightarrow+\infty}\E_n\min\left\{\|\sqrt{n}(\overline{\theta}_\cut-\theta_0)\|_\Upsilon^2,\nu\right\},
\end{flalign*}and similarly  (under Assumption \ref{ass:miss}), 
\begin{flalign*}
\mathrm{R}_q(\pi_{\mathrm{full}},\theta_0)&:=	\lim_{\nu\rightarrow+\infty}\liminf_{n\rightarrow+\infty}\E_n \min\left\{\int_\Theta nq(\theta,\theta_{0})\pi_{\mathrm{full}}(\theta\mid  \z)\dt\theta,\nu\right\}\\&=\lim_{\nu\rightarrow+\infty}\liminf_{n\rightarrow+\infty}\E_n\min\left\{nq(\overline\theta_{\mathrm{full}},\theta_0),\nu\right\}\\&=\lim_{\nu\rightarrow+\infty}\liminf_{n\rightarrow+\infty}\E_n\min\left\{\|\sqrt{n}(\overline{\theta}_{\mathrm{full}}-\theta_0)\|_\Upsilon^2,\nu\right\},
\end{flalign*} which together imply that 
$$
\mathrm{R}_q(\pi_\omega,\theta_{0})=\lim_{\nu\rightarrow+\infty}\lim_{n\rightarrow+\infty}\E_{n}\min\{\|\sqrt{n}\{\bar\theta_{}(\omega)-\theta_{0}\}\|^2_\Upsilon,\nu\} .
$$

Write 
\begin{flalign*}
\sqrt{n}\{\bar\theta(\widehat\omega_+)-\theta_{0}\}&=\sqrt{n}\{(1-\widehat\omega_+)\overline\theta_{\cut}+\widehat\omega_+\overline\theta_{\mathrm{full}}-\theta_{0}\}
\\&=\sqrt{n}(\bar\theta_{\cut}-\theta_{0})-\widehat\omega_{+}\{\sqrt{n}(\bar\theta_{\cut}-\theta_{0})-\sqrt{n}(\bar\theta_{\mathrm{full}}-\theta_{0})\}
\end{flalign*}
By Lemma \ref{lem:bias}, 
\begin{flalign*}
\sqrt{n}\{\bar\theta(\widehat\omega_{+})-\theta_{0}\}& \Rightarrow \Psi:= \Gamma_{\cut}(\xi+\tau)-\overline{\omega}(\Gamma_{\cut}-\Gamma_{\full}) (\xi+\tau),
\end{flalign*}where $\xi\sim N(0,\Omega)$,  and $\tau=(0^\top,\eta^\top)^\top$ are defined in Lemma \ref{lem:matrix_equal2}, and where 
$$
\overline{\omega}=\frac{\gamma}{{(\xi+\tau)^\top P (\xi+\tau)}},\quad P=(\Gamma_{\cut}-\Gamma_{\full})^\top{\Upsilon_{}}(\Gamma_{\cut}-\Gamma_{\full}).
$$
For $\mathcal{Q}_n:=\|\sqrt{n}\{\bar\theta(\widehat\omega_{+})-\theta_{0}\}\|^2_{\Upsilon_{}}$, and $\nu\ge0$,  let $$\Psi_{n,\nu}=\sqrt{n}\{\bar\theta(\widehat\omega_{+})-\theta_{0}\}\cdot\1[\mathcal{Q}_{n}\le\nu]+\nu\cdot \1[\mathcal{Q}_n>\nu].$$ By Theorem 1.8.8 of \cite{lehmann2006theory}, 
$$
\liminf_{n\rightarrow\infty}\E\left[\|\Psi_{n,\nu}\|_{{\Upsilon_{}}}^2\right]=\E\left[\|\Psi\|_{{\Upsilon_{}}}^2\1(\|\Psi\|_{\Upsilon_{}}^2\le \nu)\right]+\nu^2\text{Pr}(\|{\Upsilon_{}}\|^2_{\Upsilon_{}}>\nu).
$$For $\nu\rightarrow\infty$, the RHS of the above converges to $\E\left[\Psi^\top {\Upsilon_{}}\Psi\right]$, and we have that $$		\mathrm{R}_q(\pi_{\widehat\omega_{+}},\theta_{0})=\E\left[\Psi^\top {\Upsilon_{}}\Psi\right].$$ Expand $\mathrm{R}_q(\pi_{\widehat\omega_{+}},\theta_{0})$: 
\begin{flalign*}
\mathrm{R}_q(\pi_{\widehat\omega_{+}},\theta_{0})=&\E\left[\Psi^\top {\Upsilon_{}}\Psi\right]\nonumber\\=&\E\left[(\xi+\tau)^\top \Gamma_{\cut}^\top{\Upsilon_{}} \Gamma_{\cut}(\xi+\tau)\right]+\gamma^2\E\frac{(\xi+\tau)^\top P (\xi+\tau)}{[(\xi+\tau)^\top P (\xi+\tau)]^2}\\&-2\gamma\E\frac{(\xi+\tau)^\top (\Gamma_{\cut}-\Gamma_{\full})^\top{\Upsilon_{}} \Gamma_{\cut}(\xi+\tau)}{(\xi+\tau)^\top P (\xi+\tau)}\nonumber\\=&\E\left[(\xi+\tau)^\top \Gamma_{\cut}^\top{\Upsilon_{}} \Gamma_{\cut}(\xi+\tau)\right]+\gamma^2\E\frac{1}{[(\xi+\tau)^\top P (\xi+\tau)]}\\&-2\gamma\E\frac{(\xi+\tau)^\top (\Gamma_{\cut}-\Gamma_{\full})^\top{\Upsilon_{}} \Gamma_{\cut}(\xi+\tau)}{(\xi+\tau)^\top P (\xi+\tau)}.
\end{flalign*}From the definitions of $\Gamma_{\cut},\xi,\tau$, 
\begin{equation*}
\E \left[(\xi+\tau)^\top \Gamma_{\cut}^\top{\Upsilon_{}} \Gamma_{\cut}(\xi+\tau)\right]=\mathrm{R}_{q}(\pi_{\cut},\theta_{0}).
\end{equation*}

Let us concentrate our attention on the last term in $\mathrm{R}_0(\theta_0,\pi_{\widehat{\omega}_+})$. Define the mapping $\varphi(x)=x^\top/(x^\top {\Upsilon_{}} x)$ and note that
\begin{flalign*}
\E\frac{(\xi+\tau)^\top (\Gamma_{\cut}-\Gamma_{\full})^\top{\Upsilon_{}} \Gamma_{\cut}(\xi+\tau)}{(\xi+\tau)^\top P (\xi+\tau)}&=\E\{ \varphi(\xi+\tau)(\Gamma_{\cut}-\Gamma_{\full})^\top{\Upsilon_{}}\Gamma_{\cut}(\xi+\tau)\}.
\end{flalign*}
Recall that $\E(\xi)=0$, and $\Omega=\E(\xi\xi^\top)$. Using Lemma \ref{lem:stein}, and for $K:=(\Gamma_{\cut}-\Gamma_{\full})^\top{\Upsilon_{}} \Gamma_{\cut}\Omega$,	we have
\begin{flalign}
\E\frac{(\xi+\tau)^\top (\Gamma_{\cut}-\Gamma_{\full})^\top{\Upsilon_{}} \Gamma_{\cut}(\xi+\tau)}{(\xi+\tau)^\top P (\xi+\tau)}&=\E \tr\{\partial \varphi(x)/\partial x|_{x=X}\}K\nonumber\\&=\E\frac{\tr K}{(\xi+\tau)^\top P (\xi+\tau)}-2\E\frac{\tr P (\xi+\tau)(\xi+\tau)^\top K}{[{(\xi+\tau)^\top P (\xi+\tau)}]^2}\label{eq:stein} .
\end{flalign}		
From the properties of $\tr(\cdot)$,  and Lemma \ref{lem:matrix_equal2},
\begin{flalign*}
\tr K=\tr {\Upsilon_{}} \Gamma_{\cut}\Omega(\Gamma_{\cut}-\Gamma_{\full})^\top =\tr{\Upsilon_{}}\mathcal{M}.
\end{flalign*} For the second term in \eqref{eq:stein}, we have that 
\begin{flalign*}
\tr P (\xi+\tau) (\xi+\tau)^\top K&=	\tr (\xi+\tau)^\top K P(\xi+\tau) \\
&= (\xi+\tau)^\top (\Gamma_{\cut}-\Gamma_{\full})^\top {\Upsilon_{}}\Gamma_{\cut} \Omega(\Gamma_{\cut}-\Gamma_{\full})^\top {\Upsilon_{}} (\Gamma_{\cut}-\Gamma_{\full})(\xi+\tau)\\&=(\xi+\tau)^\top (\Gamma_{\cut}-\Gamma_{\full})^\top {\Upsilon_{}}  \mathcal{M} {\Upsilon_{}} (\Gamma_{\cut}-\Gamma_{\full})(\xi+\tau),
\end{flalign*}	and where the last equality again follows from Lemma \ref{lem:matrix_equal2}. Consequently, for ${\Upsilon_{}}$ positive semi-definite, 
\begin{flalign}
\tr P (\xi+\tau) (\xi+\tau)^\top K&=(\xi+\tau)^\top (\Gamma_{\cut}-\Gamma_{\full})^\top {\Upsilon_{}}  \mathcal{M} {\Upsilon_{}} (\Gamma_{\cut}-\Gamma_{\full})(\xi+\tau)\nonumber\\&\le \|{\Upsilon_{}}^{1/2}\mathcal{M} {\Upsilon_{}}^{1/2}\|\left\{(\xi+\tau)^\top P (\xi+\tau)\right\}\nonumber\\&=\|\Upsilon\mathcal{M}\|\left\{(\xi+\tau)^\top P (\xi+\tau)\right\}\label{eq:bound1},
\end{flalign}where the last inequality holds for any matrix norm $\|\cdot\|$ such that if $A$ and $B$ are positive semi-definite, then $\|AB\|=\|BA\|$; e.g., the Frobenius norm.  Applying \eqref{eq:bound1} into \eqref{eq:stein}, we have 
\begin{flalign}
&-\E\frac{(\xi+\tau)^\top (\Gamma_{\cut}-\Gamma_{\full})^\top{\Upsilon_{}} \Gamma_{\cut}(\xi+\tau)}{(\xi+\tau)^\top P (\xi+\tau)}\nonumber\\&\le -\E_{}[\frac{\tr \Upsilon\mathcal{M}}{(\xi+\tau)^\top P (\xi+\tau)}]+2\|\Upsilon\mathcal{M}\|\E_{}\left[\frac{\left\{(\xi+\tau)^\top P (\xi+\tau)\right\}}{\left\{(\xi+\tau)^\top P (\xi+\tau)\right\}^2}\right]\nonumber	\\&=-\left(\tr \Upsilon\mathcal{M}-2\|\Upsilon\mathcal{M}\|\right)\E_{}\left[\frac{1}{(\xi+\tau)^\top P (\xi+\tau)}\right]\label{eq:bound2}.
\end{flalign}
Applying \eqref{eq:bound2} into the last equation for $\mathrm{R}_0(\theta_0,\pi_{\widehat{\omega}_+})$,  and collecting terms yields 
\begin{flalign}
\mathrm{R}_q(\pi_{\widehat\omega_+},\theta_{0})&\le \mathrm{R}_{q}(\pi_{\cut},\theta_{0})+\gamma^2\E\left[\frac{1}{\{(\xi+\tau)^\top P(\xi+\tau)\}}\right]-2\gamma\E\left[\frac{(\tr \Upsilon\mathcal{M}-2\|\Upsilon\mathcal{M}\|)}{(\xi+\tau)^\top P (\xi+\tau)}\right]\nonumber\\&=
\mathrm{R}_{q}(\pi_{\cut},\theta_{0})-\gamma\E\left[\frac{\{2(\tr \Upsilon\mathcal{M}-2\|\Upsilon\mathcal{M}\|)-\gamma\}}{(\xi+\tau)^\top P (\xi+\tau)}\right]
\label{eq:loss2}.
\end{flalign}
\end{proof}

Theorem \ref{thm:ests} yields the following corollary.
\begin{corollary}\label{corr:bound_gen}
Consider that Assumptions \ref{ass:miss}-\ref{ass:loss}, and the regularity conditions in Assumptions \ref{ass:regular}-\ref{ass:regular2} are satisfied. If $\tr \Upsilon\mathcal{M}>2\|\Upsilon\mathcal{M}\|$, and $0<\gamma< 2(\tr \Upsilon\mathcal{M}-2\|\Upsilon\mathcal{M}\|)$, then for any $\eta$ such that $\|\eta\|<\infty$, 
\begin{flalign}
\mathrm{R}_q(\pi_{\widehat\omega_+},\theta_{0})&\le \mathrm{R}_{q}(\pi_{\cut},\theta_{0})-\frac{\gamma\{2(\tr \Upsilon\mathcal{M}-2\|\Upsilon\mathcal{M}\|)-\gamma
	\}}{\tr {\Upsilon_{}}\mathcal{M}+\tau^\top(\Gamma_{\cut}-\Gamma_{\full})^\top {\Upsilon_{}}(\Gamma_{\cut}-\Gamma_{\full}) \tau}<  \mathrm{R}_{q}(\pi_{\cut},\theta_{0})\label{eq:bound3}.
\end{flalign}	
Furthermore, if $\tau=(0^\top,\eta^\top)^\top$ is such that $\tau^\top(\Gamma_{\cut}-\Gamma_{\full})^\top{\Upsilon_{}}(\Gamma_{\cut}-\Gamma_{\full})\tau\ge\tr\Upsilon\mathcal{M}$, then 
$$
\mathrm{R}_q(\pi_{\widehat\omega_+},\theta_{0})\le \mathrm{R}_q(\pi_{\mathrm{full}},\theta_{0})
$$
\end{corollary}
\begin{proof}[Proof of Corollary \ref{corr:bound_gen}.]Since $\xi$ is Gaussian with mean zero and variance $\Omega$,
\begin{equation}
\E[{(\xi+\tau)^\top P (\xi+\tau)}]=\tau^\top P\tau+\tr P \Omega\label{eq:mom1}.
\end{equation}
From equation \eqref{eq:loss2} in the proof of Theorem \ref{thm:ests}, 
\begin{flalign*}
\mathrm{R}_q(\pi_{\widehat\omega_+},\theta_0)&\le \mathrm{R}_{q}(\pi_{\cut},\theta_{0})-\gamma\E\frac{[2(\tr \Upsilon\mathcal{M}-2\|\Upsilon\mathcal{M}\|)-\gamma]}{[{(\xi+\tau)^\top P (\xi+\tau)}]}\nonumber\\&\le \mathrm{R}_{q}(\pi_{\cut},\theta_{0})-\gamma\frac{[2(\tr \Upsilon\mathcal{M}-2\|\Upsilon\mathcal{M}\|)-\gamma]}{\E[{(\xi+\tau)^\top P (\xi+\tau)}]}\nonumber\\&\le\mathrm{R}_{q}(\pi_{\cut},\theta_{0}) -\gamma\frac{[2(\tr \Upsilon\mathcal{M}-2\|\Upsilon\mathcal{M}\|)-\gamma]}{\tau^\top P \tau+\tr P \Omega}\\&< \mathrm{R}_{q}(\pi_{\cut},\theta_{0}),
\end{flalign*}where the second inequality follows by Jensen's inequality, and the second to last by plugging in the moment of the quadratic form in equation \eqref{eq:mom1}, while the last (strict) inequality follows since  $0<\gamma< 2(\tr \Upsilon\mathcal{M}-2\|\Upsilon\mathcal{M}\|)$.

Since under our hypotheses, $\mathrm{R}_q(\pi_{\widehat\omega_+},\theta_0)<\mathrm{R}_{q}(\pi_{\cut},\theta_{0})$, to prove the second part of the result we need only prove that 
$$
\mathrm{R}_q(\pi_{\mathrm{full}},\theta_0)- \mathrm{R}_q(\pi_{\mathrm{cut}},\theta_0)\ge0.
$$Taking $\omega=0$  and $\omega=1$ in the proof of Theorem \ref{thm:ests}, we see that 
\begin{flalign*}
\mathrm{R}_q(\pi_{\mathrm{cut}},\theta_0)&= \E_n\{(\xi+\tau)^\top\Gamma_{\cut}^\top\Upsilon\Gamma_{\cut}(\xi+\tau)\}=\tau^\top\Gamma_{\cut}^\top{\Upsilon_{}}\Gamma_{\cut}\tau+\tr \Upsilon\Gamma_{\cut}\Omega\Gamma_{\cut}^\top\\&=\|\I_{22}^{-1}\eta_2\|^2_{\Upsilon_{22}}+\tr \Upsilon\Gamma_{\cut}\Omega\Gamma_{\cut}^\top;
\\\mathrm{R}_q(\pi_{\mathrm{full}},\theta_0)&=\E_n\{(\xi+\tau)^\top\Gamma_{\full}^\top\Upsilon\Gamma_{\full}(\xi+\tau)\}=\tau^\top\Gamma_{\full}^\top{\Upsilon_{}}\Gamma_{\full}\tau+\tr \Upsilon\Gamma_{\full}\Omega\Gamma_{\full}^\top\\&=\|\I_{}^{-1}\eta\|^2_{\Upsilon_{22}}+\tr \Upsilon\Gamma_{\full}\Omega\Gamma_{\full}^\top	.	
\end{flalign*}Consequently, 
\begin{flalign*}
\mathrm{R}_q(\pi_{\mathrm{cut}},\theta_0)-\mathrm{R}_q(\pi_{\mathrm{full}},\theta_0)&= \|\I_{22}^{-1}\|_{\Upsilon_{22}}^{2}-\|\I^{-1}\eta\|_{\Upsilon}^{2}+\tr \Upsilon\{\Gamma_{\cut}\Omega\Gamma_{\cut}^\top-\Gamma_{\full}\Omega\Gamma_{\full}^\top\}
\end{flalign*}
From the definitions of $\Gamma_{\cut},\Gamma_{\full}$, $\Omega$, and $\mathcal{M}$ in \eqref{eq:m_eq}, 
\begin{flalign*}
&\Gamma_{\cut}\Omega\Gamma_{\cut}^\top-\Gamma_{\full}\Omega\Gamma_{\full}^\top\\&=\begin{pmatrix}
	\I_{p(11)}^{-1}&-\I_{p(11)}^{-1}\I_{12}\I_{22}^{-1}\\-\I_{22}^{-1}\I_{21}\I_{p(11)}^{-1}&\I_{22}^{-1}+\I_{22}^{-1}\I_{21}\I_{p(11)}^{-1}\I_{12}\I_{22}^{-1}
\end{pmatrix}-\begin{pmatrix}
	\I_{11.2}^{-1}&-\I_{11.2}^{-1}\I_{12}\I_{22}^{-1}\\-\I_{22}^{-1}\I_{21}\I_{11.2}^{-1}&\I_{22}^{-1}+\I_{22}^{-1}\I_{21}\I_{11.2}^{-1}\I_{12}\I_{22}^{-1}
\end{pmatrix}
\\&=\mathcal{M}.
\end{flalign*}
Hence, 
$\mathrm{R}_q(\pi_{\mathrm{cut}},\theta_0)-\mathrm{R}_q(\pi_{\mathrm{full}},\theta_0)= \|\I_{22}^{-1}\|_{\Upsilon_{22}}^{2}-\|\I^{-1}\eta\|_{\Upsilon}^{2}+\tr \Upsilon\mathcal{M}$
and
$$
\mathrm{R}_q(\pi_{\mathrm{cut}},\theta_0)-\mathrm{R}_q(\pi_{\mathrm{full}},\theta_0)\le0 ,
$$ when $\|\I^{-1}\eta\|_{\Upsilon}^{2}\ge \tr \Upsilon\mathcal{M}+\|\I_{22}^{-1}\|_{\Upsilon_{22}}^{2}$, as maintained in the stated result. 
\end{proof}


\begin{proof}[Proof of Theorem \ref{corr:bound} (main text).]
Theorem \ref{corr:bound} is a direct consequence of Corollary \ref{corr:bound_gen}. To see this, note that Corollary \ref{corr:bound_gen} is true pointwise for any $\eta$ such that $\|\eta\|<\infty$. Note that $\mathrm{R}_{q}(\pi_{\cut},\theta_{0})$ and $\mathrm{R}_q(\pi_{\full},\theta_0)$ are convex for all $\eta$, and that the RHS of the bound for $\mathrm{R}_q(\pi_{\widehat{\omega}_{+}},\theta_0)$ in equation \eqref{eq:bound3} is also convex as a function of $\eta$. Hence,  Corollary \ref{corr:bound_gen} holds uniformly, which verifies Theorem \ref{corr:bound}. 
\end{proof}

\begin{proof}[Proof of Theorem \ref{corr:bound2}]Recall the expanded expression for $\mathrm{R}_q(\pi_{\widehat{\omega}_{+}},\theta_0)$ in the proof of Theorem \ref{thm:ests}. Recall equation \eqref{eq:stein} derived in that proof of Theorem \ref{thm:ests}: 
\begin{flalign*}
	\E\frac{(\xi+\tau)^\top (\Gamma_{\cut}-\Gamma_{\full})^\top\Upsilon \Gamma_{\cut}(\xi+\tau)}{(\xi+\tau)^\top P (\xi+\tau)}&=\E\frac{\tr K}{(\xi+\tau)^\top P (\xi+\tau)}-2\E\frac{\tr P (\xi+\tau)(\xi+\tau)^\top K}{[{(\xi+\tau)^\top P (\xi+\tau)}]^2} ,
\end{flalign*}where $K=(\Gamma_{\cut}-\Gamma_{\full})^\top \Upsilon \Gamma_{\cut}\Omega$. Under our choice of loss, we can show that
\begin{flalign*}
	\tr K=\tr \Upsilon\Gamma_{\cut}\Omega(\Gamma_{\cut}-\Gamma_{\full})^\top =\tr\Upsilon\mathcal{M}=d_{}.
\end{flalign*}
For the second term in \eqref{eq:stein}, 
\begin{flalign*}
	\tr P (\xi+\tau) (\xi+\tau)^\top K&=(\xi+\tau)^\top K P(\xi+\tau)\\&=(\xi+\tau)^\top (\Gamma_{\cut}-\Gamma_{\full})^\top \mathcal{M}^{-1} \Gamma_{\cut}\Omega(\Gamma_\cut-\Gamma_{\full})^\top\mathcal{M}^{-1}(\Gamma_{\cut}-\Gamma_{\full})(\xi+\tau)\\&=(\xi+\tau)^\top(\Gamma_{\cut}-\Gamma_{\full})^\top \mathcal{M}^{-1} (\Gamma_{\cut}-\Gamma_{\full}) (\xi+\tau).
\end{flalign*}Applying the above,  the second term in \eqref{eq:stein} becomes
\begin{flalign}
	\frac{\tr P (\xi+\tau)(\xi+\tau)^\top K}{[{(\xi+\tau)^\top P (\xi+\tau)}]^2} &=\frac{(\xi+\tau)^\top (\Gamma_{\cut}-\Gamma_{\full})^\top \mathcal{M}^{-1} (\Gamma_{\cut}-\Gamma_{\full})(\xi+\tau)}{\{(\xi+\tau)^\top (\Gamma_{\cut}-\Gamma_{\full})^\top \mathcal{M}^{-1} (\Gamma_{\cut}-\Gamma_{\full})(\xi+\tau)\}^2}\nonumber\\&=\frac{1}{(\xi+\tau)^\top (\Gamma_{\cut}-\Gamma_{\full})^\top \mathcal{M}^{-1} (\Gamma_{\cut}-\Gamma_{\full})(\xi+\tau)}\label{eq:invchi}.
\end{flalign}

From Lemma \ref{lem:matrix_equal2}, $\text{Cov}[(\Gamma_{\cut}-\Gamma_{\full})(\xi+\tau)]=\mathcal{M}.$ Hence, the inverse of the random variable in \eqref{eq:invchi} is distributed as non-central chi-squared with $\kappa=d$ degrees of freedom, and non-centrality parameter $\lambda=\tau^\top P\tau$, which, for $d>2$, has density function
\begin{equation*}
	f_{\chi}(z):=\exp(-\lambda)\sum_{j=0}^{\infty} \frac{\lambda^{j}}{j !} \frac{z^{\frac{1}{2}(\kappa+2 j)-1} \exp({-\frac{1}{2} z})}{2^{\frac{1}{2}(\kappa+2 j)} \Gamma[(\kappa+2 j) / 2]}.
\end{equation*}
To calculate $\E[Z^{-1}]$, first  we see that 
\begin{flalign*}
	\E[Z^{-1}]=\int z^{-1}f_{\chi}(z)\dt \mu(z)=\exp(-\lambda)\sum_{j=0}^{\infty}\int_0^\infty \frac{\lambda^{j}}{j !} \frac{z^{\frac{1}{2}(\kappa+2 j)-2} \exp({-\frac{1}{2} z})}{2^{\frac{1}{2}(\kappa+2 j)} \Gamma[(\kappa+2 j) / 2]}\dt \mu(z),
\end{flalign*}	and note that, since $\kappa/2>1$, when $d>2$, for all $j\ge0$, 
$$
\int z^{\frac{1}{2}(\kappa+2 j)-r-1}\exp(-z/2)\dt \mu(z)=2^{(\kappa/2+j)-1}\Gamma[(\kappa/2+j)-1], 
$$so that we can rewrite the expectation as 
\begin{flalign*}
	\E[ Z^{-1}] &=2^{-1} \frac{\Gamma(\kappa / 2-1)}{\Gamma(\kappa / 2)} \exp(-\lambda) \sum_{j=0}^{\infty} \frac{(\kappa / 2-1)_{j}}{(\kappa / 2)_{j}} \left(\frac{\lambda^{j}}{j !}\right) \\&=2^{-1} \frac{\Gamma(\kappa / 2-1)}{\Gamma(\kappa / 2)} \exp(-\lambda){ }_{1} F_{1}(\kappa / 2-1 ; \kappa / 2 ; \lambda).
\end{flalign*}Taking $\kappa=d_{}$ then yields the result. 		
\end{proof}

\subsection{Proofs of Preliminary Results in Section \ref{sec:prelim_results}}
When no confusion is likely to result, quantities that are evaluated at $\theta_0$ will have their dependence on $\theta_0$ suppressed; i.e., we write $\mathcal{I}_{p(11)}(\theta_0)=\mathcal{I}_{p(11)}$, etc.

\begin{proof}[Proof of Lemma \ref{lem:matrix_equal1}.]
Note that, under the regularity conditions in Assumption \ref{ass:regular}, we can exchange the order of integration and differentiation. Furthermore, note that for $\delta_n$ as in Assumption \ref{ass:miss} and $z\in\mathcal{Z}$, we have that 
\begin{flalign*}
	\lim_{n\rightarrow+\infty}n^{-1}\E_{n}[\dot\ell(\theta_0)\dot\ell(\theta_0)^\top]&=\lim_{n\rightarrow+\infty}\int_{\mathcal{Z}}\frac{\partial \log f(z\mid\theta_0)}{\partial\theta}\frac{\partial \log f(z\mid\theta_0)}{\partial\theta^\top}h(z\mid\theta_0,\delta_n)\dt \mu(z)	\\&=\int_{\mathcal{Z}}\frac{\partial \log f(z\mid\theta_0)}{\partial\theta}\frac{\partial \log f(z\mid\theta_0)}{\partial\theta^\top}h(z\mid\theta_0,0)\dt \mu(z)\\&=\int_{\mathcal{Z}}\frac{\partial \log f(z\mid\theta_0)}{\partial\theta}\frac{\partial \log f(z\mid\theta_0)}{\partial\theta^\top}f(z\mid\theta_0)\dt \mu(z)\\&=\E[\dot\ell(z_i\mid\theta_0)\dot\ell(z_i\mid\theta_0)^\top ],
\end{flalign*}where recall that
$$
\dot\ell(z_i\mid\theta_0)=\partial \log f(z_i\mid\theta_0)/\partial\theta, \text{ and } \ddot\ell(z_i\mid\theta_0)=\partial^2 \log f(z_i\mid\theta_0)/\partial\theta\partial\theta^\top.
$$
Similarly, from the structure of the score equations, and the information matrix equality, we have that 
{\small
	\begin{flalign*}
		-\E\ddot\ell(z_i\mid\theta_0)&=\E[\dot\ell(z_i\mid\theta_0)\dot\ell(z_i\mid\theta_0)^\top ]	\\-\E\begin{pmatrix}
			\ddot\ell_{p(11)}+\ddot\ell_{c(11)}&\ddot\ell_{c(12)}\\\ddot\ell_{c(21)}&\ddot\ell_{c(22)}
		\end{pmatrix}&=\E\begin{pmatrix}
			\{\ddot\ell_{p(1)}+\ddot\ell_{c(1)}\}\{\ddot\ell_{p(1)}+\ddot\ell_{c(1)}\}^\top&\{\ddot\ell_{p(1)}+\ddot\ell_{c(1)}\}\ddot\ell_{c(2)}^\top\\\ddot\ell_{c(2)}\{\ddot\ell_{p(1)}+\ddot\ell_{c(1)}\}^\top&\ddot\ell_{c(2)}\ddot\ell_{c(2)}^\top
		\end{pmatrix}\\\begin{pmatrix}
			\mathcal{I}_{p(11)}+\mathcal{I}_{c(11)}&\mathcal{I}_{(12)}\\\mathcal{I}_{(21)}&\mathcal{I}_{(22)}
		\end{pmatrix}&=\E\begin{pmatrix}\dot\ell_{p(1)}\dot\ell_{p(1)}^\top+\dot\ell_{c(1)}\dot\ell_{c(1)}^\top+2\dot\ell_{p(1)}\dot\ell_{c(1)}&\ddot\ell_{p(1)}\ddot\ell_{c(2)}^\top+\ddot\ell_{c(1)}\ddot\ell_{c(2)}^\top\\\ddot\ell_{c(2)}\ddot\ell_{p(1)}^\top+\ddot\ell_{c(2)}\ddot\ell_{c(1)}^\top&\ddot\ell_{c(2)}\ddot\ell_{c(2)}^\top
		\end{pmatrix}
\end{flalign*}}

Analysing the above we see that  
\begin{flalign}
	\mathcal{I}_{p(11)}=&-\int_{\mathcal{Z}} \frac{\partial^2 \log f_1(z\mid\theta_{1})}{\partial\theta_1\partial\theta_1^\top} f(z\mid\theta_0)\dt \mu(z)\nonumber\\=&\int_{\mathcal{Z}}\frac{1}{f_1(z\mid\theta_{1,0})^2}\frac{\partial   f_1(z\mid\theta_{1,0})}{\partial \theta_1^\top}\frac{\partial f_1(z\mid\theta_{1,0})}{\partial \theta_1^\top}f(z\mid\theta_0)\dt \mu(z)\nonumber\\&\hspace{5mm}-\int_{\mathcal{Z}}\frac{\partial^2  f_1(z\mid\theta_{1,0})}{\partial\theta_1\partial \theta_1^\top}f(z\mid\theta_0)\dt \mu(z)\nonumber\\=&\E\dot\ell_p(z_i\mid\theta_{1,0})\dot\ell_p(z_i\mid\theta_{1,0})^\top,\label{eq:equal_p}
\end{flalign}where the last line follows from rewriting $\frac{1}{f_1(z\mid\theta_{1,0})}\frac{\partial f_1(z\mid\theta_{1,0})}{\partial \theta^\top}=\dot\ell_{p(1)}(z\mid\theta_{1,0})$, and from exchanging integration and differentiation to note that the second term is zero. This proves part one of the result. Repeating the argument for $\mathcal{I}_{p(11)}$ for $\mathcal{I}_{c(11)}$ yields
\begin{flalign*}
	\mathcal{I}_{c(11)}=\E\dot\ell_{c(1)}(z_i\mid\theta_{0})\dot\ell_{c(1)}(z_i\mid\theta_{0})^\top,
\end{flalign*}which proves the second part of the result. 

Now, let us investigate the equivalence of each term. For the first term we have that 
\begin{flalign}
	\mathcal{I}_{p(11)}+\mathcal{I}_{c(11)}&=\E\dot\ell_{p(1)}(z_i\mid\theta_{1,0})\dot\ell_{p(1)}(z_i\mid\theta_{1,0})^\top+\E\dot\ell_{c(1)}(z_i\mid\theta_{0})\dot\ell_{c(1)}(z_i\mid\theta_{0})^\top\nonumber\\&+2\E\dot\ell_{p(1)}(z_i\mid\theta_{1,0})\dot\ell_{c(1)}(z_i\mid\theta_{0})^\top\label{eq:equal_one}.
\end{flalign}
Applying the above equations for $\mathcal{I}_{j(11)}$, $j\in\{p,c\}$, in equation \eqref{eq:equal_one} we have 
\begin{flalign*}
	\mathcal{I}_{p(11)}+\mathcal{I}_{c(11)}&=\E\dot\ell_p(z_i\mid\theta_{1,0})\dot\ell_p(z_i\mid\theta_{1,0})^\top+\E\dot\ell_{c(1)}(z_i\mid\theta_{0})\dot\ell_{c(1)}(z_i\mid\theta_{0})^\top\\&=\E\dot\ell_{p(1)}(z_i\mid\theta_{1,0})\dot\ell_{p(1)}(z_i\mid\theta_{1,0})^\top+\E\dot\ell_{c(1)}(z_i\mid\theta_{0})\dot\ell_{c(1)}(z_i\mid\theta_{0})^\top\\&+2\E\dot\ell_{p(1)}(z_i\mid\theta_{1,0})\dot\ell_{c(1)}(z_i\mid\theta_{0})^\top,
\end{flalign*}which is satisfied if and only if $\E\dot\ell_{p(1)}(z_i\mid\theta_{0})\dot\ell_{c(1)}(z_i\mid\theta_{0})^\top=0$. Hence, we have part three of the stated result.

Lastly, we investigate the term $\mathcal{I}_{(21)}$. Again, using the definitions of this term we have 
\begin{flalign}
	\mathcal{I}_{21}=&-n^{-1}\int_{\mathcal{Z}} \frac{\partial^2\ell_c(z\mid \theta_0)}{\partial\theta_2\partial\theta_1^\top} f(z\mid\theta_0)\dt \mu(z)\nonumber\\=&-\int_{\mathcal{Z}} \frac{1}{f_2(z\mid\theta_0)^2}\frac{\partial f_2(z\mid\theta)}{\partial\theta_2}\frac{\partial f_2(z\mid\theta)}{\partial\theta_1^\top} f(z\mid\theta_0)\dt \mu(z)-\int_{\mathcal{Z}} \frac{\partial^2f_2(z\mid\theta)}{\partial\theta_2\partial\theta_1^\top} f(z\mid\theta_0)\dt \mu(z)\nonumber\\=&\E\dot\ell_{c(2)}(z\mid\theta_0)\dot\ell_{c(1)}(z\mid\theta_0)^\top,\label{eq:equal_three}
\end{flalign}where the last equality again follows from exchanging integration and differentiation of the second term, and rewriting the derivatives in the first term. Therefore, from equation \eqref{eq:equal_three}, and the general matrix information equality, we have the equality
\begin{flalign*}
	\mathcal{I}_{21}=\E\dot\ell_{c(2)}(z\mid\theta_0)\dot\ell_{c(1)}(z\mid\theta_0)^\top=\E\dot\ell_{c(2)}(z\mid\theta_0)\dot\ell_{c(1)}(z\mid\theta_0)^\top+\E\dot\ell_{c(2)}(z\mid\theta_0)\dot\ell_{p(1)}(z\mid\theta_{1,0})^\top,	
\end{flalign*}which is satisfied if and only if $\E\dot\ell_{c(2)}(z\mid\theta_0)\dot\ell_{p(1)}(z\mid\theta_{1,0})^\top=0$, which proves the last result. 			
\end{proof}
\begin{proof}[Proof of Lemma \ref{lem:newey}]To simplify the proof, let $\delta_n(z)=\psi^\top \xi(z)/\sqrt{n}$. By Assumption \ref{ass:regular}, under Assumption \ref{ass:miss}, we see that $\|\dot\ell(z\mid\theta)\| h(z\mid\theta,\delta_n)\le b(z)\{1+\delta_n(z)\}f_1(z\mid\theta_{1,0})f_2(z\mid\theta_0)+o(b(z)/\sqrt{n})$. By Assumption \ref{ass:miss}, there exists some $a_1(z)\ge0$ with $\sup_{z\in\mathcal{Z}}a_1(z)<\infty$ such that $\delta_n(z)\le a_1(z)$. Hence, for each $n\ge1$, $h(z \mid \theta,\delta_n) \leqslant \{1+a_1(z)/\sqrt{n}\}f_1(z\mid\theta_{1,0})f_2(z\mid\theta_0)$ and $|\dot\ell(z\mid\theta)h(z \mid \theta,\delta_n)| \leqslant b(z)\{1+a_1(z)/\sqrt{n}\}f_1(z\mid\theta_{1,0})f_2(z\mid\theta_0)$. From Assumption \ref{ass:regular}, we have that $\int_{\mathcal{Z}}a_1(z)b(z)f_1(z\mid\theta_{1,0})f_2(z\mid\theta_0)\dt\mu(z)<\infty$, and $\int_{\mathcal{Z}}b(z)f_1(z\mid\theta_{1,0})f_2(z\mid\theta_0)\dt\mu(z)<\infty$. By the dominated convergence theorem, we then have that $\phi(\theta,\delta_n)$ exists for each $\delta_n\in\Delta$. 

To prove the second part of the result, we note that, by Assumption \ref{ass:miss}, 
\begin{flalign}
	&\lim_{n\rightarrow\infty} \sup_{\theta\in\tilde\Theta}|\phi(\theta, \delta_n)-\phi(\theta, 0)|=\lim_{n\rightarrow\infty} \frac{1}{\sqrt{n}}\sup_{\theta\in\tilde\Theta}\left|\int_{\mathcal{Z}}\dot\ell(z\mid\theta)\delta_n(z)f_1(z\mid\theta_{1,0})f_2(z\mid\theta_0)\dt\mu(z)\right|\nonumber\\&\le \lim_{n\rightarrow\infty}\frac{1}{\sqrt{n}}\left|\int_{\mathcal{Z}}	b(z)a_1(z)f_1(z\mid\theta_{1,0})f_2(z\mid\theta_0)\dt\mu(z)+o\left\{\frac{1}{\sqrt{n}}\int_{\mathcal{Z}}	b(z)f_1(z\mid\theta_{1,0})f_2(z\mid\theta_0)\dt\mu(z)\right\}\right|\nonumber\\&=o(1/\sqrt{n}),\label{eq:newlem1}
\end{flalign}where the second line follows since by Assumption \ref{ass:regular} $\|\dot\ell(z\mid\theta)\|\le b(z)$.

Now, by the triangle inequality, 
\begin{flalign*}
	\sup_{\theta\in\tilde\Theta}\left|n^{-1}\dot\ell(\theta)-\phi\left(\theta, 0\right)\right| \le &\sup_{\theta\in\tilde\Theta}\left|n^{-1}\dot\ell(\theta)-\phi\left(\theta, \delta_n\right)\right|+\sup_{\theta\in\tilde\Theta}\left|\phi\left(\theta, \delta_n\right)-\phi\left(\theta, 0\right)\right|.
\end{flalign*}	From equation \eqref{eq:newlem1}, for any $\varepsilon>0$, the there exists an $n_1$ large enough so that for all $n\ge n_1$, $\sup_{\theta\in\widetilde\Theta}\left|\phi\left(\theta, \delta_n\right)-\phi\left(\theta, 0\right)\right|\le \varepsilon/2$. To handle the first term, note that by Assumption \ref{ass:regular}, for any $\theta,\theta'\in\Theta$, and some $\overline\theta\in\Theta$, 
$$
\|\dot\ell(z\mid\theta)-\dot\ell(z\mid\theta')\|\le \|\ddot\ell(z\mid\overline\theta)\|\|\theta-\theta'\|\le b(z)\|\theta-\theta'\|.
$$Since $\E[b(z)]<\infty$, it follows directly from Theorem 21.10 of \cite{davidson1994stochastic}, that
$$
\sup_{\theta\in\tilde\Theta}\left|n^{-1}\dot\ell(\theta)-\phi\left(\theta, \delta_n\right)\right|\le \varepsilon/2
$$with probability converging to one. Hence, for any $\varepsilon>0$,
\begin{flalign*}
	\lim_{n\rightarrow+\infty}\text{Pr}\left[\sup_{\theta\in\tilde\Theta}\left|n^{-1}\dot\ell(\theta)-\phi\left(\theta, 0\right)\right|\ge \varepsilon\right]\le &\lim_{n\rightarrow+\infty}\text{Pr}\left[\sup_{\theta\in\tilde\Theta}\left|n^{-1}\dot\ell(\theta)-\phi\left(\theta, \delta_n\right)\right|\ge \varepsilon/2\right]\\&+\lim_{n\rightarrow+\infty}\text{Pr}\left[\sup_{\theta\in\tilde\Theta}\left|\phi\left(\theta, \delta_n\right)-\phi\left(\theta, 0\right)\right|\ge \varepsilon/2\right]\\&=o(1).
\end{flalign*}

\end{proof}

\begin{proof}[Proof of Lemma \ref{lem:CLTmiss}.]
The result uses  Lemma \ref{lem:newey} , the nature of the model misspecification in Assumption \ref{ass:miss}, and the separable structure of the likelihood.  

\smallskip 

\noindent\textbf{Result 1.} Recall the expectation
$$
\phi(\theta,\delta_n):=\int_{\mathcal{Z}} \frac{\partial\log f(z\mid\theta)}{\partial\theta}  h(z\mid\theta,\delta_n)\dt\mu(z).
$$From Lemma \ref{lem:newey}, and the regularity conditions on $h(z\mid\theta,\delta)$ in Assumption \ref{ass:regular}, $\phi(\theta,\delta_n)$ exists and is continuous on $\Theta\times\Delta$. The second portion of Lemma \ref{lem:newey} implies that, for any $0<\varepsilon<+\infty$,
$$ 
\sup_{\theta\in\Theta(\varepsilon)}|n^{-1}\dot\ell(\theta)-\phi(\theta,0)|=o_p(1).
$$ 

Define $\phi_n:=\phi(\theta_0,\delta_n)$ and note that, since $\phi(\cdot,\delta)$ is continuous in $\delta$, as $n\rightarrow\infty$,  
$$
\phi_n\rightarrow\phi(\theta_0,\delta_0).
$$However, since we can exchange differentiation and integration (Assumption \ref{ass:regular}), and since $h(\cdot\mid\theta,\delta)$ is continuous in $\delta$, for each $\theta\in\Theta$, we have 
\begin{flalign}
	\phi(\theta_0,\delta_n)&=\int_{\mathcal{Z}} \frac{\partial\log f(z\mid\theta_0)}{\partial\theta}  h(z\mid\theta_0,\delta_n)\dt\mu(z)\nonumber
	\\&=\int_{\mathcal{Z}} \frac{\partial\log f(z\mid\theta_0)}{\partial\theta}  \{1+\delta_n(z)/\sqrt{n}\}f(z\mid\theta_0)\dt\mu(z)\{1+o(1/\sqrt{n})\}\nonumber \\&=\int_{\mathcal{Z}} \frac{\partial\log f(z\mid\theta_0)}{\partial\theta}  f(z\mid\theta_0)\dt\mu(z)+\frac{1}{\sqrt{n}}\int_{\mathcal{Z}} \frac{\partial\log f(z\mid\theta_0)}{\partial\theta} \delta_n(z)\dt\mu(z)\{1+o(1/\sqrt{n})\}.
\end{flalign}The first term in the above is zero since 
$$\int_{\mathcal{Z}}\frac{\partial \log f(z\mid \theta_0)}{\partial\theta} f(z\mid\theta_0)\dt\mu(z)=\int_{\mathcal{Z}}\frac{\partial}{\partial\theta}f(z\mid\theta_0)\dt\mu(z)=\frac{\partial}{\partial\theta}\int_{\mathcal{Z}}f(z\mid\theta_0)\dt\mu(z)=0.
$$Considering the second term, using Assumption \ref{ass:miss} in the main text and Assumption \ref{ass:regular}, we see that the integral portion of the second term satisfies
$$
\left\|\int_{\mathcal{Z}} \frac{\partial\log f(z\mid\theta_0)}{\partial\theta}  \delta_n(z)\dt\mu(z)\right\|\le \int_{\mathcal{Z}}b(z)a(z)\dt\mu(z)<\infty.
$$Hence, as $n\rightarrow\infty$, 
$
\lim_{n\rightarrow+\infty}		\sqrt{n}\phi(\theta_0,\delta_n)=\eta,
$ which yields the first part of the stated result. 

\bigskip 

\noindent\textbf{Result 2.}
To deduce the second result, we need only establish a CLT for 
$
\{n^{-1/2}\dot\ell(\theta_0)-\sqrt{n}\phi(\theta_0,\delta_n)\}
$. This can be established using arguments similar to those given in in Lemma 2.1 of \cite{newey1985maximum}. In particular, let $\lambda\in\mathbb{R}^{d}$ be a non-zero vector, with $\|\lambda\|<\infty$, and define $Y_n(z):=\lambda^{\top}\left[\dot\ell(z\mid\theta_0)-\phi_n\right]$ and $Y_{i,n}=Y_n\left(z_i\right)(i=1, \ldots, n)$. For each $n\ge1$, the random variable $Y_{i,n}$ is mean-zero and and has a strictly positive variance $\lambda^{\top} \Omega_n\lambda$, for all $n$ large enough, where $\Omega_n:=\E_{n}[\dot\ell(z\mid\theta_0)\dot\ell(z\mid\theta_0)^\top]$. For $\varepsilon>0$ let $A_n(\varepsilon)=\left\{z:\left|\dot\ell\left(z, \theta_0\right)-\phi_n\right|>\varepsilon \sqrt{n \lambda^{\top} \Omega_n\lambda}\right\}$. For each $n, Y_{i,n}(i=1, \ldots, T)$ are identically distributed, so that by Assumption \ref{ass:regular}, for any $\varepsilon>0$,
\begin{flalign*}
	&{\left[1 /\left(n \lambda^{\top} \Omega_n \lambda\right)\right] \cdot \sum_{i=1}^n \int_{\left|Y_{i,n}\right| \geqslant \varepsilon \sqrt{n \lambda^{\top} \Omega_n\lambda}}\left|Y_{i,n}\right|^2 h\left(z_i \mid \theta_0,\delta_n\right) \dt \mu\left(z_i\right)} \\
	&=\left(1 / \lambda^{\top} \Omega_n \lambda\right) \cdot \int_{A_n(\varepsilon)}\left|Y_n(z)\right|^2 h\left(z \mid \theta_0,\delta_n\right) \dt \mu(z) \\
	&\le \left(1 / \lambda^{\top} \Omega_n \lambda\right)\|\lambda\|^2 \int_{A_n(\varepsilon)}\|\dot\ell(z\mid\theta_0)-\phi_n\|^2
	h\left(z \mid \theta_0,\delta_n\right) \dt \mu(z)\\&\le  \left(1 / \lambda^{\top} \Omega_n \lambda\right)\|\lambda\|^2\left\{\|\phi_n\|^2+ \int_{A_n(\varepsilon)}\|\dot\ell(z\mid\theta_0)\|^2
	h\left(z \mid \theta_0,\delta_n\right) \dt \mu(z)\right\}
	\\& \leq \left(1 / \lambda^{\top} \Omega_n \lambda\right)  \|\lambda\|\left[\left\|\phi_n\right\|^2+\int_{A_n(\varepsilon)} b(z) a(z) \dt \mu(z)\right]
\end{flalign*}
Since $h(z\mid\theta_0,\delta_n)\rightarrow h(z\mid\theta_0,\delta_0)\equiv f(z\mid\theta_0)$, by continuity $\Omega_n\rightarrow\E[\dot\ell(z\mid\theta_0)\dot\ell(z\mid\theta_0)^\top]$, and $\lim \lambda^{\top} \Omega_n \lambda=\lambda^{\top} \Omega \lambda>0$, so that by Lemma A.1 of \cite{newey1985maximum}, the set $A_n(\varepsilon)\rightarrow\emptyset$, as $n\rightarrow+\infty$. Hence, by Assumption \ref{ass:regular}, $\lim_n \int_{A_n(\varepsilon)} b(z) a(z) \dt \mu(z)=0$. Lastly, since $\phi_n\rightarrow0$, we have that 
\begin{flalign*}
	&{\left[1 /\left(n \lambda^{\top} \Omega_n \lambda\right)\right] \cdot \sum_{i=1}^n \int_{\left|Y_{i,n}\right| \geqslant \varepsilon \sqrt{ } n\lambda^{\top} \Omega_{n }\lambda}\left|Y_{i,n}\right|^2 h\left(z_i \mid \theta_0,\delta_n\right) \dt \mu\left(z_i\right)}=o(1),
\end{flalign*}
which implies that the  Lindberg-Feller condition for the central limit theorem is satisfied, so that by the Cramer-Wold device: $$ \{n^{-1/2}\dot\ell\left(\theta_0\right)-\sqrt{n}\phi_n\} {\Rightarrow} N(0, \Omega)$$ and the second part of the stated result is satisfied. 

\bigskip 

\noindent\textbf{Result 3.} First, define
$$
\phi_{1}(\theta,\delta_n):=\int_{\mathcal{Z}}\frac{\partial \log f_1(z\mid\theta_1)}{\partial\theta_1} h(z\mid\theta,\delta_n)\dt\mu(z).
$$Applying the definition of $h(z\mid\theta,\delta_n)$ in Assumption \ref{ass:miss} yields, up to an $o(1/\sqrt{n})$ term,  
\begin{flalign}
	\phi_{1}(\theta,\delta_n)&=\phi_1(\theta,\delta_0)+\partial\phi_1(\theta_0,\bar\delta)/\partial\delta\nonumber\\&=	\int_{\mathcal{Z}}\frac{\partial \log f_1(z\mid\theta_{1,0})}{\partial\theta_1}  f(z\mid\theta_{0})\dt\mu(z) +\int_{\mathcal{Z}}\frac{\partial \log f_1(z\mid\theta_{1,0})}{\partial\theta_1}\delta_n(z)f(z\mid\theta_0)\dt\mu(z)\frac{1}{\sqrt{n}}.\label{eq:r3} 
\end{flalign}For the first term in \eqref{eq:r3}, we have 
\begin{flalign*}
	\int_{\mathcal{Z}}\frac{\partial \log f_1(z\mid\theta_{1,0})}{\partial\theta_1}  f(z\mid\theta_0)\dt\mu(z)&=\int_{\mathcal{Z}}\frac{\partial f_1(z\mid\theta_{1,0})}{\partial\theta_1}  f_2(z\mid\theta_0)\dt\mu(z)\\&=\frac{\partial }{\partial\theta}\int f_1(z\mid\theta_{1,0})f_2(z\mid\theta_{0})\dt\mu(z)\\&=0.
\end{flalign*}
By Assumption \ref{ass:miss}(iii), the second term in \eqref{eq:r3} satisfies
\begin{flalign*}
	\frac{1}{\sqrt{n}}\int_{\mathcal{Z}}\frac{\partial \log f_1(z\mid\theta_{1,0})}{\partial\theta_1}\delta_n(z)f_1(z\mid\theta_{1,0})f_2(z\mid\theta_0)\dt\mu(z)=0.
\end{flalign*}

Repeating the same arguments as in the proof of \textbf{Result 2}, but only for $\dot\ell_p(z\mid\theta_1)=\partial \log f_1(z\mid\theta_1)/\partial\theta_1$, we have that
$$ n^{-1/2}\dot\ell_p\left(\theta_{1,0}\right) {\Rightarrow} N\{0, \E(\dot\ell_{p(1)}\dot\ell_{p(1)}^\top)\}\equiv N(0,\I_{p(11)}).$$

\bigskip 

\noindent\textbf{Result 4.} Define
$$
\phi_{2}(\theta,\delta_n):=\int_{\mathcal{Z}}\frac{\partial \log f_2(z\mid\theta)}{\partial\theta_2} h(z\mid\theta,\delta_n)\dt\mu(z)
$$ and apply the definition of $h(z\mid\theta,\delta_n)$ in Assumption \ref{ass:miss} to see that, up to an $o(1/\sqrt{n})$ term,  
\begin{flalign*}
	\phi_{2}(\theta,\delta_n)=&\int_{\mathcal{Z}}\frac{\partial \log f_2(z\mid\theta_{0})}{\partial\theta_2}  f(z\mid\theta_{0})\dt\mu(z)+\int_{\mathcal{Z}}\frac{\partial \log f_2(z\mid\theta_{0})}{\partial\theta_2}\delta_n(z)f(z\mid\theta_0)\dt\mu(z)\frac{1}{\sqrt{n}}.
\end{flalign*}
Similar arguments to those used to prove \textbf{Result 3} yield 
\begin{flalign*}
	\phi_{2}(\theta_0,\delta_n)&=0+\eta/\sqrt{n}.
\end{flalign*}
The distributional result then follows similarly to \textbf{Result 2}. 
\end{proof}
\begin{proof}[Proof of Lemma \ref{lem:matrix_equal2}.] We prove each result in turn. 

\bigskip 

\noindent\textbf{Result 1.} Write $\tau_n=(\mathbf{0}_{d_1\times1}^\top, \phi_n^\top)^\top$, where $\phi_n=\phi(\theta_0^\top,\delta_n^\top)^\top$ was defined in the proof of Lemma \ref{lem:CLTmiss} and recall that $\tau=(\mathbf{0}_{d_1\times1}^\top,\eta^\top)^\top$.  The first and second results in Lemma \ref{lem:CLTmiss} together imply that 
$$
\{Z_n-\sqrt{n}\tau_n\}\Rightarrow N(0,\Omega),
$$where {\small
	\begin{flalign*}
		\Omega&:=\lim_{n\rightarrow+\infty} \E[Z_nZ_n^\top]\\&=
		\begin{pmatrix}\E[\dot\ell_{p(1)}\dot\ell_{p(1)}^\top]&\E[\dot\ell_{p(1)}\{\dot\ell_{p(1)}+\dot\ell_{c(1)}\}^\top]&\E[\dot\ell_{p(1)}\dot\ell_{c(2)}^\top]\\\E[\{\dot\ell_{p(1)}+\dot\ell_{c(1)}\}\dot\ell_{p(1)}^\top]&\E[\{\dot\ell_{p(1)}+\dot\ell_{c(1)}\}\{\dot\ell_{p(1)}+\dot\ell_{c(1)}\}^\top]&\E[\{\dot\ell_{p(1)}+\dot\ell_{c(1)}\}\dot\ell_{c(2)}^\top]\\E[\dot\ell_{c(2)}\dot\ell_{p(1)}^\top]	&\E[\dot\ell_{c(2)}\{\dot\ell_{p(1)}+\dot\ell_{c(1)}\}^\top]	&\E[\dot\ell_{c(2)}\dot\ell_{c(2)}^\top]
		\end{pmatrix}	
\end{flalign*}}
However, from Lemma \ref{lem:matrix_equal1}, we know that $$\lim_{n\rightarrow+\infty}n^{-1}\E_n\left[\dot\ell_{p(1)}\dot\ell_{c(1)}^\top\right]=0,\text{ and } \lim_{n\rightarrow+\infty}n^{-1}\E_n\left[\dot\ell_{p(1)}\dot\ell_{c(2)}^\top\right]=0.$$ Applying these relationships, we then have the general form of $\Omega$ as stated in the result. Since $\sqrt{n}\tau_n\rightarrow\tau$ in probability, the first stated result follows.

The second part of the result follows by applying the results of Lemma \ref{lem:matrix_equal1} to see that we can rewrite $\Omega$ as 
\begin{flalign*}
	\Omega=\begin{pmatrix}
		\mathcal{I}_{p(11)}&\mathcal{I}_{p(11)}&0\\\mathcal{I}_{p(11)}&\mathcal{I}_{11}&\mathcal{I}_{12}\\0&\mathcal{I}_{21}&\mathcal{I}_{22}
	\end{pmatrix},
\end{flalign*}where 
$$
\I_{11}=\I_{p(11)}+\I_{c(11)},\quad\I_{22}=\I_{c(22)},\quad\I_{21}=\I_{c(21)}.
$$

\bigskip 

\noindent\textbf{Result 2.}		
Define
$$
V_{22}:=\I_{22}^{-1}-\I_{22.1}^{-1},\quad \I_{22.1}^{-1}:=\I_{22}^{-1}+\I_{22}^{-1}\I_{21}\I_{11.2}^{-1}\I_{12}\I_{22}^{-1},
$$ and recall the definitions of $\Gamma_{\cut}$ and $\Gamma_{\full}$ to see that
$$
\Gamma_{\cut}-\Gamma_{\full}=\begin{pmatrix}
	\I_{p(11)}^{-1}&-\I_{11.2}^{-1}&\I_{11.2}^{-1}\I_{12}\I_{22}^{-1}\\-\I_{22}^{-1}\I_{21}\I_{p(11)}^{-1}&\I_{22}^{-1}\I_{21}\I_{11.2}^{-1}&V_{22}		
\end{pmatrix}.
$$ 
Firstly,  
\begin{flalign*}
	&(\Gamma_{\cut}-\Gamma_{\full})\Omega=
	\begin{pmatrix}
		\I_{p(11)}^{-1}&-\I_{11.2}^{-1}&\I_{11.2}^{-1}\I_{12}\I_{22}^{-1}\\-\I_{22}^{-1}\I_{21}\I_{p(11)}^{-1}&\I_{22}^{-1}\I_{21}\I_{11.2}^{-1}&\I_{22}^{-1}-\I_{22.1}^{-1}		
	\end{pmatrix}
	\begin{pmatrix}
		\mathcal{I}_{p(11)}&\mathcal{I}_{p(11)}&0\\\mathcal{I}_{p(11)}&\mathcal{I}_{11}&\mathcal{I}_{12}\\0&\mathcal{I}_{21}&\mathcal{I}_{22}
	\end{pmatrix}
	\\&=\begin{pmatrix}
		I_{d_1}-\I_{11.2}^{-1}\I_{p(11)}&I-\I_{11.2}^{-1}\I_{11}+\I_{11.2}^{-1}\I_{12}\I_{22}^{-1}\I_{21}&\mathbf{O}_{d_1\times d_2}\\-\I_{22}^{-1}\I_{21}+\I_{22}^{-1}\I_{21}\I_{11.2}^{-1}\I_{p(11)}&-\I_{22}^{-1}\I_{21}+\I_{22}^{-1}\I_{21}\I_{11.2}^{-1}\I_{11}+V_{22}\I_{21}&\I_{22}^{-1}\I_{21}\I_{11.2}^{-1}\I_{12}+{V_{22}}\I_{22}	
	\end{pmatrix}.	
\end{flalign*}Therefore,  
\begin{flalign}
&(\Gamma_{\cut}-\Gamma_{\full})\Omega(\Gamma_{\cut}-\Gamma_{\full})^\top=\nonumber\\
&\begin{pmatrix}
	I_{d_1}-\I_{11.2}^{-1}\I_{p(11)}&I-\I_{11.2}^{-1}\I_{11}+\I_{11.2}^{-1}\I_{12}\I_{22}^{-1}\I_{21}&\mathbf{O}_{d_1\times d_2}\\-\I_{22}^{-1}\I_{21}+\I_{22}^{-1}\I_{21}\I_{11.2}^{-1}\I_{p(11)}&-\I_{22}^{-1}\I_{21}+\I_{22}^{-1}\I_{21}\I_{11.2}^{-1}\I_{11}+V_{22}\I_{21}&\I_{22}^{-1}\I_{21}\I_{11.2}^{-1}\I_{12}+{V_{22}}\I_{22}	
\end{pmatrix}\times\nonumber\\&\begin{pmatrix}
	\I_{p(11)}^{-1}&-\I_{p(11)}^{-1}\I_{12}\I_{22}^{-1}\\-\I_{11.2}^{-1}&\I_{11.2}^{-1}\I_{12}\I_{22}^{-1}\\\I_{22}^{-1}\I_{21}\I_{11.2}^{-1}&V_{22}
\end{pmatrix}
	.\label{eq:gross11}
\end{flalign}
To derive the stated result, we analyze each of the corresponding blocks in the $d\times d$ matrix. 

\medskip

\noindent\textbf{$(d_1\times d_1)$-Block.} Let $\mathcal{M}_{11}$ denote the  $(d_1\times d_1)$-block of $\mathcal{M}$. Multiply the first row and first column of \eqref{eq:gross11} to obtain 
\begin{flalign*}
	\mathcal{M}_{11}&=\I_{p(11)}^{-1}-\I_{11.2}^{-1}-\I_{11.2}^{-1}+\I_{11.2}\I_{11}\I_{11.2}^{-1}-\I_{11.2}^{-1}\I_{12}\I_{22}^{-1}\I_{21}\I_{11.2}^{-1}\\&=W-\I_{11.2}^{-1}+\I_{11.2}^{-1}\underbrace{(\I_{11}-\I_{12}\I_{22}^{-1}\I_{21})}_{=\I_{11.2}}\I_{11.2}^{-1}\\&=W-\I_{11.2}^{-1}+\I_{11.2}^{-1}\\&=W,
\end{flalign*}where we recall that $W=\I_{p(11)}^{-1}-\I_{11.2}^{-1}$. 

\medskip 

\noindent\textbf{$(d_1\times d_2)$-Block.} Let $\mathcal{M}_{12}$ denote the  $(d_1\times d_2)$-block of $\mathcal{M}$. Multiply the first row and second column of \eqref{eq:gross11} to obtain 
\begin{flalign*}
	\mathcal{M}_{12}&=-\I_{p(11)}^{-1}\I_{12}\I_{22}^{-1}+\I_{11.2}^{-1}\I_{12}\I_{22}^{-1}\\&+\I_{11.2}^{-1}\I_{12}\I_{22}^{-1}-\I_{11.2}^{-1}\I_{11}\I_{11.2}^{-1}\I_{12}\I_{22}^{-1}+\I_{11.2}^{-1}\I_{12}\I_{22}^{-1}\I_{21}\I_{11.2}^{-1}\I_{12}\I_{22}^{-1}\\&=-W\I_{12}\I_{22}^{-1}\\&+\I_{11.2}^{-1}\I_{12}\I_{22}^{-1}-\I_{11.2}^{-1}(\I_{11}-\I_{12}\I_{22}^{-1}\I_{21})\I_{11.2}^{-1}\I_{12}\I_{22}^{-1}\\&=-W\I_{12}\I_{22}^{-1}.
\end{flalign*}

\medskip 

\noindent\textbf{$(d_2\times d_2)$-Block.} Let $\mathcal{M}_{22}$ denote the  $(d_2\times d_2)$-block of $\mathcal{M}$. Multiply the second row and second column of \eqref{eq:gross11}, and use the fact that $\I_{22.1}^{-1}=\I_{22}^{-1}+\I_{22}^{-1}\I_{21}\I_{11.2}^{-1}\I_{12}\I_{22}^{-1}$, so that $V_{22}=\I_{22}^{-1}-\I_{22.1}^{-1}=-\I_{22}^{-1}\I_{21}\I_{11.2}^{-1}\I_{12}\I_{22}^{-1}$, to obtain 
\begin{flalign*}
	\mathcal{M}_{22}&=\I_{22}^{-1}\I_{21}\I_{p(11)}^{-1}\I_{12}\I_{22}^{-1}-\I_{22}^{-1}\I_{21}\I_{11.2}^{-1}\I_{12}\I_{22}^{-1}\\&-\I_{22}^{-1}\I_{21}\I_{11.2}^{-1}\I_{12}\I_{22}^{-1}+\I_{22}^{-1}\I_{21}\I_{11.2}^{-1}\I_{11}\I_{11.2}^{-1}\I_{12}\I_{22}^{-1}\\&+V_{22}\I_{21}\I_{11.2}^{-1}\I_{12}\I_{22}^{-1}+\I_{22}^{-1}\I_{21}\I_{11.2}^{-1}\I_{12}V_{22}+V_{22}\I_{22}V_{22}\\&=\I_{22}^{-1}\I_{21}W\I_{12}\I_{22}^{-1}-\I_{22}^{-1}\I_{21}\I_{11.2}^{-1}\I_{12}\I_{22}^{-1}+\I_{22}^{-1}\I_{21}\I_{11.2}^{-1}\I_{11}\I_{11.2}^{-1}\I_{12}\I_{22}^{-1}\\&-\I_{22}^{-1}\I_{21}\I_{11.2}^{-1}\I_{12}\I_{22}^{-1}\I_{21}\I_{11.2}^{-1}\I_{12}\I_{22}^{-1}+\I_{22}^{-1}\I_{21}\I_{11.2}^{-1}\I_{12}V_{22}+V_{22}\I_{22}V_{22}\\&=\I_{22}^{-1}\I_{21}W\I_{12}\I_{22}^{-1}-\I_{22}^{-1}\I_{21}\I_{11.2}^{-1}\I_{12}\I_{22}^{-1}+\I_{22}^{-1}\I_{21}\I_{11.2}^{-1}\I_{12}\I_{22}^{-1}\\&+\I_{22}^{-1}\I_{21}\I_{11.2}^{-1}\I_{12}V_{22}+V_{22}\I_{22}V_{22}\\&=\I_{22}^{-1}\I_{21}W\I_{12}\I_{22}^{-1}-\I_{22}^{-1}\I_{21}\I_{11.2}^{-1}\I_{12}\I_{22}^{-1}\I_{21}\I_{11.2}^{-1}\I_{12}\I_{22}^{-1}+\I_{22}^{-1}\I_{21}\I_{11.2}^{-1}\I_{12}\I_{22}^{-1}\I_{22}\I_{22}^{-1}\I_{21}\I_{11.2}^{-1}\I_{12}\I_{22}^{-1}\\&=\I_{22}^{-1}\I_{21}W\I_{12}\I_{22}^{-1}.
\end{flalign*}

\medskip 

\noindent Applying each of $\mathcal{M}_{11},\mathcal{M}_{12},\mathcal{M}_{22}$ into equation \eqref{eq:gross11}  yields
\begin{flalign*}
	&(\Gamma_{\cut}-\Gamma_{\full})\Omega(\Gamma_{\cut}-\Gamma_{\full})^\top=\begin{pmatrix}
		W&-W\I_{12}\I_{22}^{-1}\\-\I_{22}^{-1}\I_{21}W&\I_{22}^{-1}\I_{21}W\I_{12}\I_{22}^{-1}
	\end{pmatrix}=\mathcal{M}.
\end{flalign*}	

\bigskip 

\noindent\textbf{Result 3.} Similar to \textbf{Result 2}, 
\begin{flalign*}
	&\Gamma_{\cut}\Omega(\Gamma_{\cut}-\Gamma_{\full})^\top=\\&
	\begin{pmatrix}
		\I_{p(11)}^{-1}&\mathbf{O}&\mathbf{O}\\-\I_{22}^{-1}\I_{21}\I_{p(11)}^{-1}&\mathbf{O}&\I_{22}^{-1}		
	\end{pmatrix}\begin{pmatrix}
		\mathcal{I}_{p(11)}&\mathcal{I}_{p(11)}&0\\\mathcal{I}_{p(11)}&\mathcal{I}_{11}&\mathcal{I}_{12}\\0&\mathcal{I}_{21}&\mathcal{I}_{22}
	\end{pmatrix}\times (\Gamma_{\cut}-\Gamma_{\full})^\top
	\\&=\begin{pmatrix}
		I_{d_1}&I_{d_1}&\mathbf{O}\\-\I_{22}^{-1}\I_{21}&\mathbf{O}&I_{d_2}
	\end{pmatrix}\times \begin{pmatrix}
		\I_{p(11)}^{-1}&-\I_{p(11)}^{-1}\I_{12}\I_{22}^{-1}\\-\I_{11.2}^{-1}&\I_{11.2}^{-1}\I_{12}\I_{22}^{-1}\\\I_{22}^{-1}\I_{21}\I_{11.2}^{-1}&V_{22}
	\end{pmatrix}\\&=\begin{pmatrix}
		W&-W\I_{12}\I_{22}^{-1}\\-\I_{22}^{-1}\I_{12}W&\I_{22}^{-1}\I_{21} W\I_{12}\I_{22}^{-1}
	\end{pmatrix}
	\\&=\mathcal{M}.
\end{flalign*}
\end{proof}

\begin{proof}[Proof of Theorem \ref{thm:exact}.]
Recall $\overline\theta_{\mathrm{full}}:=\int_{\Theta}\theta\pi_{\full}(\theta\mid \z)\dt\theta$, and decompose
\begin{flalign*}
	\sqrt{n}(\overline\theta_{\mathrm{full}}-\theta_0)&=\sqrt{n}(\overline\theta_{\mathrm{full}}-\theta_0-T_{n,\full})+\sqrt{n}T_{n,\full}\\&=\int_{\Theta}\sqrt{n}(\theta-\theta_0-T_{n,\full})\pi(\theta\mid \z)\dt\theta+\sqrt{n}T_n\\&=\int_{\mathcal{T}}t\pi_{\full}(t\mid \z)\dt\theta+\sqrt{n}T_{n,\full}
\end{flalign*}where $\mathcal{T}:=\{t=\sqrt{n}(\theta-\theta_0-T_{n,\full}):\theta\in\Theta\}$ and $$T_{n,\full}=\Gamma_{\full}Z_n/\sqrt{n}=\mathcal{I}^{-1}\dot\ell(\theta_0)/n.$$
From Lemma \ref{lem:CLTmiss}, 
$$
\sqrt{n}T_{n,\full}\Rightarrow N(\I^{-1}\eta,\I^{-1}).
$$ Hence, the stated result follows if $\int_{\mathcal{T}}t\pi_{\full}(t\mid \z)=o_p(1)$. 

To show this, we first note that 
$$
\int_{\mathcal{T}}t\pi_{\full}(t\mid \z)\dt t=\int_{\mathcal{T}}t[\pi_{\full}(t\mid \z)-N\{t;0,\mathcal{I}^{-1}\}]\dt t,
$$since $0=\int_{\mathcal{T}}tN\{t;0,\mathcal{I}^{-1}\}\dt t$. We then have 
$$
\left\|\int_{\mathcal{T}}t[\pi_{\full}(t\mid \z)-N\{t;0,\mathcal{I}^{-1}\}]\dt t\right\|\le \int_{\mathcal{T}}\|t\||\pi_{\full}(t\mid \z)-N\{t;0,\mathcal{I}^{-1}\}|\dt t
$$
The regularity conditions in Assumption \ref{ass:regular}-\ref{ass:regular2} are sufficient to apply Theorem 8.2 of \cite{lehmann2006theory}, to deduce that, for $\mathcal{T}:=\{t=\sqrt{n}(\theta-\theta_0-T_{n,\full}):\theta\in\Theta\}$, 
\begin{equation*}
	\int_{\mathcal{T}}\|t\||\pi_{\full}(t\mid \z)-N\{t;0,\mathcal{I}^{-1}\}|=o_p(1). 
\end{equation*}
\end{proof}

\begin{proof}[Proof of Theorem \ref{thm:cut}.]
The proof follows similar arguments to the proof of Theorem \ref{thm:exact} if we set 
$$
T_{n,\cut}=\begin{pmatrix}
	T_{1,n,\cut}\\T_{2,n,\cut}
\end{pmatrix}=\begin{pmatrix}\Gamma_{1,\cut}Z_n/\sqrt{n}\\\Gamma_{2,\cut}Z_n/\sqrt{n}	
\end{pmatrix}\equiv 
\begin{pmatrix}
	\mathcal{I}_{p(11)}^{-1}\dot\ell_{p(1)}(\theta_{1,0})/{n}\\\mathcal{I}_{22}^{-1}\{\dot\ell_{c(2)}(\theta_0)/{n}-\I_{21}\I_{p(11)}^{-1}\dot\ell_{p(1)}(\theta_{1,0})/{n}\}.
\end{pmatrix}.
$$In particular, if we replace the full posterior in the proof of Theorem \ref{thm:exact} by the cut posterior $\pi_\cut(\theta_1\mid \z)$, and the full posterior mean $\overline{\theta}_{\full}$ with the cut posterior mean $\overline\theta_{1,\cut}=\int_{\Theta}\theta\pi_\cut(\theta_1\mid \z)\dt\theta$, we have that 
$$
\int_{\mathcal{V}_1}\|\vartheta_1\||\pi_{\cut}(\vartheta_1\mid\z)-N\{\vartheta_1;0,\I_{p(11)}^{-1}\}|\dt\vartheta_1=o_p(1). 
$$Using the above results and the following decomposition
\begin{flalign*}
	\sqrt{n}(\overline\theta_{1,\cut}-\theta_{1,0})&=\sqrt{n}(\overline\theta_{\mathrm{full}}-\theta_{1,0}-T_{1,n,\cut})+\sqrt{n}T_{1,n,\cut}\\&=\int_{\Theta_1}\sqrt{n}(\theta_1-\theta_{1,0}-T_{1,n,\cut})\pi_\cut(\theta_1\mid \z)\dt\theta_1+\sqrt{n}T_{1,n,\cut}\\&=\int_{\mathcal{V}_1}\vartheta_1\pi_{\cut}(\vartheta_1\mid \z)\dt\vartheta_1+\sqrt{n}T_{1,n,\cut},
\end{flalign*}we have that 
\begin{flalign*}
	\sqrt{n}(\overline\theta_{1,\cut}-\theta_{1,0})&=\sqrt{n}T_{1,n,\cut}+o_p(1),
\end{flalign*}	
Appealing to Lemma \ref{lem:CLTmiss} we have that 
$$
\sqrt{n}T_{1,n,\cut}\Rightarrow N(0,\I_{p(11)}^{-1}),
$$which yields the stated results for the $\theta_1$ dimension.

To prove the results for the $\theta_2$ dimension, first write 
\begin{flalign*}
	\sqrt{n}(\overline\theta_{2,\cut}-\theta_{2,0})&=\sqrt{n}(\overline\theta_{2,\cut}-\theta_{2,0}-T_{2,n,\cut})+\sqrt{n}T_{2,n,\cut}\\&=\int_{\Theta_2}\sqrt{n}(\theta_2-\theta_{2,0}-T_{2,n,\cut})\pi_\cut(\theta_2\mid \z)\dt\theta+\sqrt{n}T_{2,n,\cut}\\&=\int_{\mathcal{T}_2}\vartheta_2\pi_\cut(\vartheta_2\mid \z)\dt \vartheta_2 +\sqrt{n}T_{2,n,\cut},
\end{flalign*}where $\mathcal{V}_2:=\{\vartheta_2=\sqrt{n}(\theta_2-\theta_{2,0}-T_{2,n,\cut}):\theta_2\in\Theta_2\}$. From Lemma \ref{lem:CLTmiss}, the second term satisfies
$$
\sqrt{n}T_{2,n,\cut}\Rightarrow N(\I_{22}^{-1}\eta_2,\I_{22}^{-1}+\I_{22}^{-1}\I_{21}\I_{p(11)}^{-1}\I_{12}\I_{22}^{-1})
$$and the result then follows if $\int_{\mathcal{T}_2}\vartheta_2\pi_\cut(\vartheta_2\mid \z)\dt \vartheta_2=o_p(1)$. However, similar to the proof of Theorem \ref{thm:exact}
\begin{flalign*}
	\int_{\mathcal{T}_2}\vartheta_2\pi_\cut(\vartheta_2\mid \z)\dt \vartheta_2\le&  \int_{\mathcal{T}_2}\|\vartheta_2\||\pi_\cut(\vartheta_2\mid \z)\dt \vartheta_2- N(\vartheta_2;0,\I_{22}^{-1}+\I_{22}^{-1}\I_{21}\I_{p(11)}^{-1}\I_{12}\I_{22}^{-1})|\dt \vartheta_2\\&
	+\int_{\mathcal{T}_2}\vartheta_2 N(\vartheta_2;0,\I_{22}^{-1}+\I_{22}^{-1}\I_{21}\I_{p(11)}^{-1}\I_{12}\I_{22}^{-1})\dt\vartheta_2.
\end{flalign*}
The regularity conditions in Assumption \ref{ass:regular}-\ref{ass:regular2} are sufficient to apply Corollary 1 of \cite{frazier2022cutting} to deduce that 
\begin{equation*}
	\int_{\mathcal{T}_2}\|\vartheta_2\||\pi_\cut(\vartheta_2\mid \z)\dt \vartheta_2- N(\vartheta_2;0,\I_{22}^{-1}+\I_{22}^{-1}\I_{21}\I_{p(11)}^{-1}\I_{12}\I_{22}^{-1})|\dt \vartheta_2=o_p(1). 
\end{equation*}

\end{proof}

\begin{proof}[Proof of Corollary \ref{corr:bias_compare}]
For $j=1$, we must only note that the cut posterior mean exhibits no bias, so that the result is satisfied. 	

For $j=2$, recall that by Theorem \ref{thm:cut}, the bias of the cut posterior mean is given by $\mathcal{I}^{-1}_{22}\eta_2$. By Theorem \ref{thm:exact}, the full posterior for $\theta_2$ exhibits the asymptotic bias 
$$
\left(\mathcal{I}_{22}^{-1}+\mathcal{I}_{22}^{-1}\mathcal{I}_{21}\mathcal{I}_{11.2}^{-1}\mathcal{I}_{12}\mathcal{I}_{22}^{-1}\right)\eta_2-\mathcal{I}_{22}^{-1}\mathcal{I}_{21}\mathcal{I}_{11.2}^{-1}\eta_1.
$$The squared difference between the bias for $\theta_{2,0}$ under the cut and full posteriors is then positive so long as
\begin{flalign*}
	&\eta_2^\top \mathcal{I}_{22}^{-1}\mathcal{I}_{21}\mathcal{I}_{11.2}^{-1}\mathcal{I}_{12}\mathcal{I}_{22}^{-1}\mathcal{I}_{22}^{-1}\mathcal{I}_{21}\mathcal{I}_{11.2}^{-1}\mathcal{I}_{12}\mathcal{I}_{22}^{-1}\eta_2+\eta_1^\top\mathcal{I}_{11.2}^{-1}\mathcal{I}_{12}\mathcal{I}_{22}^{-1}\mathcal{I}_{22}^{-1}\mathcal{I}_{21}\mathcal{I}_{11.2}^{-1}\eta_1\\&>2\eta_1^\top\mathcal{I}_{11.2}^{-1}\mathcal{I}_{12}\mathcal{I}_{22}^{-1}\mathcal{I}_{22}^{-1}\mathcal{I}_{21}\mathcal{I}_{11.2}^{-1}\mathcal{I}_{12}\mathcal{I}_{22}^{-1}\eta_2
\end{flalign*}
Writing 
$$
X=\mathcal{I}_{22}^{-1}\mathcal{I}_{21}\mathcal{I}_{11.2}^{-1}\mathcal{I}_{12}\mathcal{I}_{22}^{-1}\eta_2,\quad Y=\mathcal{I}_{22}^{-1}\mathcal{I}_{21}\mathcal{I}_{11.2}^{-1}\eta_1,
$$we can rewrite the above restriction as 
$$
X^\top X+ Y^\top Y-2 X^\top Y=\|X-Y\|^2\ge0.
$$
\end{proof}

\begin{proof}[Proof of Lemma \ref{lem:bias}.]
Firstly, note that, in the proof of Theorem \ref{thm:exact}  we have shown that $$\|\sqrt{n}(\overline\theta_{\full}-\theta_{0})-\Gamma_{\full}Z_n\|=o_p(1),$$ and that, by Theorem \ref{thm:cut}, $$\|\sqrt{n}(\overline\theta_{\cut}-\theta_{0})-\Gamma_{\cut} Z_n\|=o_p(1).$$ Using these results and Lemmass \ref{lem:matrix_equal1}-\ref{lem:matrix_equal2}, it follows that 
$$
\begin{pmatrix}
	\sqrt{n}(\overline\theta_{\cut}-\theta_{0})\\\sqrt{n}(\overline\theta_{\full}-\theta_{0})\end{pmatrix}=\begin{pmatrix}
	\Gamma_{\cut}\\\Gamma_{\full}
\end{pmatrix}Z_n+o_p(1)\Rightarrow\begin{pmatrix}\Gamma_{\cut}\\\Gamma_{\full}
\end{pmatrix}(\xi+\tau).
$$
The continuous mapping theorem then yields the stated results. 
\end{proof}

\section{HPV prevalence and cancer incidence}\label{app:HPV}

We return to the example in Section 2.1 of the main text, concerning the relationship
between HPV prevalence and cervical cancer incidence.  We use it to
demonstrate how the semi-modular posterior can vary for different loss
functions.  In particular, we consider an S-SMP for a loss function targeting
the HPV prevalence in country $j$, $\theta_{1,j}$, and show how the S-SMP
mixing weight and semi-modular inferences vary across $j$.  
We use the mixing weight
$$\widehat{\omega}_+=\min\{1,\widehat{\omega}\},\;\;\;\widehat{\omega}=\frac{{\sigma_{\text{cut}}^{(j)}}^2-{\sigma^{(j)}}^2}{(\bar{\theta}_{1,\text{cut}}^{(j)}-\bar{\theta}_1^{(j)})^2}\mathbb{I}({\sigma_{\text{cut}}^{(j)}}^2-{\sigma^{(j)}}^2>0),$$
where ${\sigma_{\text{cut}}^{(j)}}^2$ and ${\sigma^{(j)}}^2$ are the 
cut and full marginal posterior variances for $\theta_{1,j}$, and
$\bar{\theta}_{1,\text{cut}}^{(j)}$ and $\bar{\theta}_1^{(j)}$ are the
cut and full marginal posterior means.  We denote the marginal SMP for $\theta_{1,j}$ by $\pi_{\text{cut}}^{(j)}(\theta_{1,j}|X)$,
and Figure \ref{hpv-figure2} shows
kernel estimates of these densities together with the SMP 
mixing weights $\widehat{\omega}_+$, and the
cut and full posterior $\theta_{1,j}$ marginals.  The countries
are ordered in the plot (left to right, top to bottom) according to the cut
posterior mean values.  Interestingly, the pooling weights
vary widely according to which country is under analysis;  when the difference in location
of the cut and full posterior distributions is large compared to the posterior 
variability, the SMP is close to the cut posterior, whereas the SMP is closer
to the full posterior otherwise.  

{	The results obtained in Figure \ref{hpv-figure2} give a novel insight into the behavior of model misspecification in this example. While it is known that the Poisson assumption is inadequate, there exists significant variability in the adequacy of this assumption for the purpose of inference in specific countries: for certain countries we obtain a pooling weight of unity, so that the SMP corresponds to the full posterior, while for other countries the pooling weight is close to zero, so that the SMP corresponds to the cut posterior. In view of Corollary \ref{corr:gross}, when the model is grossly misspecified the SMP produces a pooling weight that is asymptotically zero. Hence, the empirical results in Figure \ref{hpv-figure2} demonstrate that there are certain countries for which the Poisson model is adequate and others for which it is not. Critically, however, the SMP we have proposed does not require us to choose which posterior to use for which country \textit{a priori}, and instead the SMP lets the observed data guide our choice of which posterior to use according to which posterior is better supported under the observed data.\footnote{MCMC computations for the full posterior distribution were done using Stan \citep{carpenter_stan_2017}, whereas cut posterior samples for $\theta_1$ given $X$ are generated directly using conjugacy, followed by sampling importance resampling based on $1,000$ draws from an asymptotic normal approximation to the conditional posterior of $\theta_2$ given $\theta_1$, $Y$, to obtain each $\theta_2$ sample.   
	Kernel density estimates in the plots are based on $1,000$ posterior samples. }
}

\begin{figure}[htb]
\centerline{\includegraphics[width=120mm]{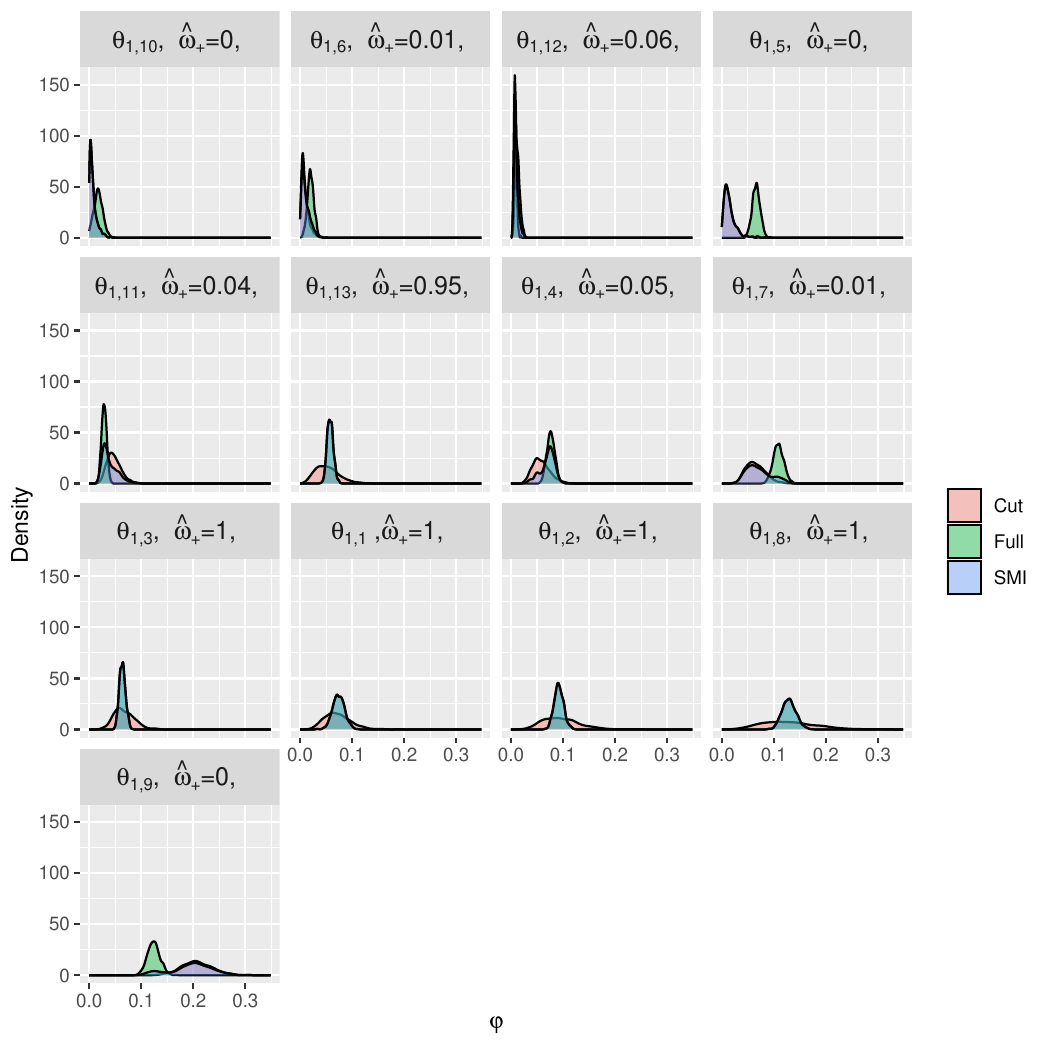}}
\caption{\label{hpv-figure2}Marginal cut, full and semi-modular posterior (SMP) densities for HPV prevalence in different countries.  The SMP obtained for country $j$ uses a loss function depending on $j$ as described in the text.  $\hat{\omega}_+$ is the estimated mixing weight in the SMP.  Countries are ordered (left to right, top to bottom) according to the cut posterior mean values.}
\end{figure}


\section{Additional Results for Biased Means Example}\label{app:biased_mean}
In this section, we present the repeated sampling behavior of the cut, full and S-SMP in the simulation design presented in Section \ref{sec:bias_mean}  of the main text and additional results on the behavior of the S-SMP for additional choices of $d$. Please refer to the main text for full details of the simulation design. 

Table \ref{tab:biased_rep_d1} compares the behavior of the posterior mean for the three methods across different values of $\delta$, and in the case where $d_1=1$. Similar results for the case of $d_1=5$ are reported in Table \ref{tab:biased_rep_d5}. In each table we report the Bias of the posterior mean (Bias), the average posterior variance across the replications (Var) and the Monte Carlo coverage of a 95\% posterior credible sets (Cov). 

Regarding the results for  $d_1=1$ and $d_1=5$, we see that in both cases the the S-SMP attempts to trade-off bias in the full posterior for a significant reduction in variability. For both cases, the average cut posterior variance is about 0.01 across all values of $\delta$, while the average S-SMP variance ranges from about $0.003$ to $0.005$, which is significantly smaller than that of the cut posterior. Further, the bias exhibited by the S-SMP is generally much smaller than that exhibited by the full posterior, but is, by construction, larger than that exhibited by the cut posterior. 

Recall that, as discussed in the remarks raised after the statement of Lemma \ref{lem:two} in the main text, under model misspecification only the cut posterior for $\theta_1$ can deliver credible sets that are also valid frequentist confidence sets. Analyzing the coverage of the posteriors across the two designs, we see that this is indeed the case, with only the cut posterior delivering accurate coverage. In comparison, in both cases, the full posterior has coverage that is zero or nearly zero once $\delta>0.10$. The S-SMP fairs better than the full posterior, but can deliver unreliable coverage when $d_1=1$. In contrast, as the dimension of $\theta_1$ increases to $d_1=5$, the coverage of the S-SMP  dramatically appreciates and in most cases is only slightly smaller than the nominal level for most values of $\delta$. While not reported for the sake of brevity, extending these experiments to larger values of $d_1$ show that the coverage of the S-SMP improves as $d_1$ increases, so that when $d_1=25$ the coverage for the S-SMP is close to the nominal level across all values of $\delta$ used in the experiment. 

Lastly, we demonstrate the impact on the accuracy of the S-SMP, as measured by expected squared error, as $d_1$ increases. To this end, we present the expected squared error of the cut, full and S-SMP for $d_1=1, 5,10,25$. We plot these results visually in Figure \ref{fig_risk2}. Analyzing these results, we see that as $d_1$ increases, and for $\delta$ small, the S-SMP resembles the full posterior, but as $\delta$ increases the S-SMP more closely resemble the cut posterior. This is a consequence of the behavior of the pooling weight used in the S-SMP, which has an easier time shifting between the cut and full posterior when $d>2$; we refer to Section \ref{sec:allrisk} in the main paper for a technical explanation of this phenomenon.

\begin{table}[!htp]\centering
\scriptsize
\begin{tabular}{lrrrrrrrrrr}\toprule
	& &\underline{Cut} & & &\underline{Full} & & &\underline{S-SMP} & \\
	$\delta$ &Bias &Var &Cov &Bias &Var &Cov &Bias &Var &Cov \\\midrule
	0.1000 &-0.0081 &0.0099 &0.9400 &-0.0131 &0.0001 &0.7540 &-0.0098 &0.0028 &0.8190 \\
	0.1500 &-0.0080 &0.0099 &0.9410 &-0.0194 &0.0001 &0.5430 &-0.0121 &0.0029 &0.7040 \\
	0.2000 &-0.0082 &0.0099 &0.9410 &-0.0257 &0.0001 &0.3100 &-0.0143 &0.0029 &0.5800 \\
	0.2500 &-0.0080 &0.0099 &0.9430 &-0.0318 &0.0001 &0.1410 &-0.0161 &0.0030 &0.5050 \\
	0.3000 &-0.0082 &0.0099 &0.9400 &-0.0380 &0.0001 &0.0360 &-0.0182 &0.0030 &0.4560 \\
	0.3500 &-0.0080 &0.0099 &0.9430 &-0.0441 &0.0001 &0.0060 &-0.0198 &0.0031 &0.4520 \\
	0.4000 &-0.0081 &0.0099 &0.9400 &-0.0503 &0.0001 &0.0010 &-0.0218 &0.0032 &0.4610 \\
	0.4500 &-0.0081 &0.0099 &0.9420 &-0.0566 &0.0001 &0.0000 &-0.0234 &0.0033 &0.4720 \\
	0.5000 &-0.0082 &0.0099 &0.9420 &-0.0628 &0.0001 &0.0000 &-0.0248 &0.0034 &0.4810 \\
	0.5500 &-0.0081 &0.0099 &0.9430 &-0.0690 &0.0001 &0.0000 &-0.0261 &0.0036 &0.4950 \\
	0.6000 &-0.0082 &0.0099 &0.9420 &-0.0752 &0.0001 &0.0000 &-0.0276 &0.0037 &0.5010 \\
	0.6500 &-0.0081 &0.0099 &0.9400 &-0.0814 &0.0001 &0.0000 &-0.0289 &0.0038 &0.4950 \\
	0.7000 &-0.0081 &0.0099 &0.9430 &-0.0877 &0.0001 &0.0000 &-0.0302 &0.0040 &0.5010 \\
	0.7500 &-0.0081 &0.0099 &0.9410 &-0.0939 &0.0001 &0.0000 &-0.0314 &0.0042 &0.5170 \\
	0.8000 &-0.0079 &0.0099 &0.9420 &-0.1001 &0.0001 &0.0000 &-0.0321 &0.0043 &0.5440 \\
	0.8500 &-0.0081 &0.0099 &0.9410 &-0.1065 &0.0001 &0.0000 &-0.0332 &0.0045 &0.5580 \\
	0.9000 &-0.0082 &0.0099 &0.9430 &-0.1126 &0.0001 &0.0000 &-0.0340 &0.0047 &0.5790 \\
	\bottomrule
\end{tabular}
\caption{Repeated sampling results for $\beta_1$ in the biased means example presented in Section \ref{sec:bias_mean} of the main text when $d_1=1$. Bias is the average bias of the posterior mean across the replications. Var is the average of the posterior variance, and Cov is the actual coverage of the 95\% posterior credible set.}
\label{tab:biased_rep_d1}
\end{table}

\begin{table}[!htp]\centering
\scriptsize
\begin{tabular}{lrrrrrrrrrr}\toprule
	& &\underline{Cut }& & &\underline{Full} & & &\underline{S-SMP} & \\
	Delta &Bias &Var &Cov &Bias &Var &Cov &Bias &Var &Cov \\\midrule
	0.1000 &-0.0074 &0.0099 &0.9550 &-0.0128 &0.0005 &0.9960 &-0.0107 &0.0039 &0.9800 \\
	0.1500 &-0.0075 &0.0099 &0.9530 &-0.0186 &0.0005 &0.9900 &-0.0135 &0.0039 &0.9790 \\
	0.2000 &-0.0075 &0.0099 &0.9540 &-0.0247 &0.0005 &0.9640 &-0.0162 &0.0040 &0.9770 \\
	0.2500 &-0.0073 &0.0099 &0.9540 &-0.0306 &0.0005 &0.8840 &-0.0185 &0.0041 &0.9560 \\
	0.3000 &-0.0075 &0.0099 &0.9530 &-0.0366 &0.0005 &0.7440 &-0.0213 &0.0042 &0.9220 \\
	0.3500 &-0.0075 &0.0099 &0.9570 &-0.0426 &0.0005 &0.5120 &-0.0235 &0.0043 &0.8750 \\
	0.4000 &-0.0074 &0.0099 &0.9540 &-0.0486 &0.0005 &0.2570 &-0.0256 &0.0045 &0.8360 \\
	0.4500 &-0.0073 &0.0099 &0.9550 &-0.0546 &0.0004 &0.1060 &-0.0275 &0.0047 &0.8130 \\
	0.5000 &-0.0074 &0.0099 &0.9550 &-0.0605 &0.0004 &0.0270 &-0.0293 &0.0049 &0.8150 \\
	0.5500 &-0.0074 &0.0099 &0.9550 &-0.0665 &0.0004 &0.0060 &-0.0306 &0.0051 &0.8290 \\
	0.6000 &-0.0074 &0.0099 &0.9550 &-0.0726 &0.0004 &0.0000 &-0.0318 &0.0053 &0.8400 \\
	0.6500 &-0.0073 &0.0099 &0.9560 &-0.0787 &0.0004 &0.0000 &-0.0326 &0.0055 &0.8570 \\
	0.7000 &-0.0074 &0.0099 &0.9560 &-0.0848 &0.0004 &0.0000 &-0.0335 &0.0058 &0.8640 \\
	0.7500 &-0.0073 &0.0099 &0.9530 &-0.0908 &0.0004 &0.0000 &-0.0338 &0.0060 &0.8790 \\
	0.8000 &-0.0075 &0.0099 &0.9550 &-0.0970 &0.0004 &0.0000 &-0.0342 &0.0062 &0.8830 \\
	0.8500 &-0.0074 &0.0099 &0.9560 &-0.1031 &0.0004 &0.0000 &-0.0343 &0.0065 &0.8980 \\
	0.9000 &-0.0074 &0.0099 &0.9550 &-0.1093 &0.0003 &0.0000 &-0.0344 &0.0067 &0.8990 \\
	\bottomrule
\end{tabular}
\caption{Repeated sampling results for $\beta_1$ in the biased means example presented in Section \ref{sec:bias_mean} of the main text when $d_1=5$. Bias is the average bias of the posterior mean across the replications. Var is the average of the posterior variance, and Cov is the actual coverage of the 95\% posterior credible set.}
\label{tab:biased_rep_d5}
\end{table}

\begin{figure}[!htb]
\centering
\begin{subfigure}[b]{0.49\textwidth}
	\centerline{\includegraphics[width=160mm,height=60mm]{risk_mean_figure.eps}}
	\captionsetup{width=.99\linewidth}
\end{subfigure}

\vspace{5pt}

\begin{subfigure}[b]{0.49\textwidth}
	\centerline{\includegraphics[width=160mm,height=70mm]{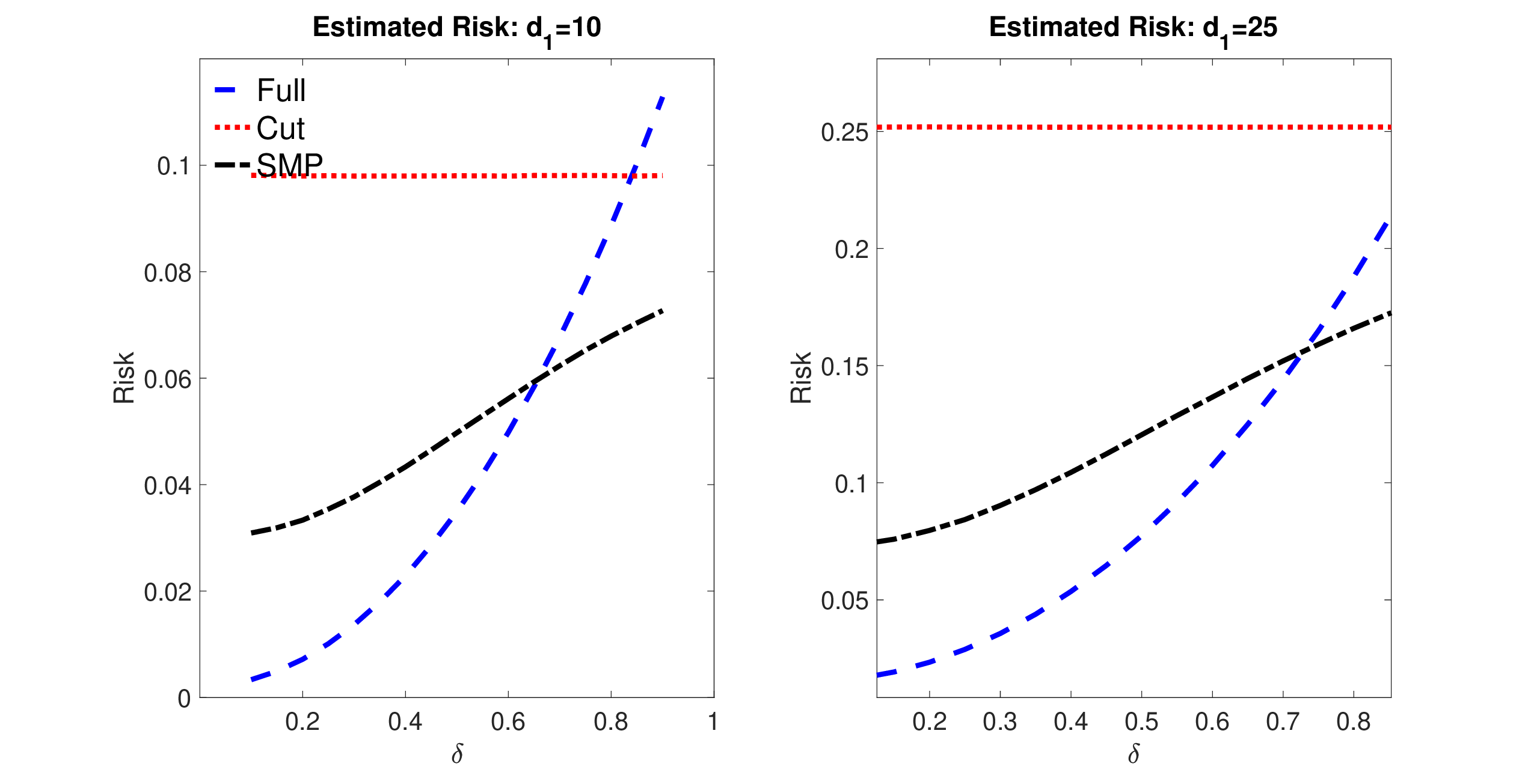}}
	\captionsetup{width=.99\linewidth}
\end{subfigure}
\caption{{Monte Carlo estimate of expected risk under squared error loss across different levels of contamination ($\delta$)  when $d_1\in\{1,5,10,25\}$. Full refers to the expected risk for $\theta_{1,0}$ associated with the full posterior based on both sets of data; Cut is the P-risk associated to the cut posterior; SMP is the P-risk for the proposed semi-modular posterior. }}
\label{fig_risk2}
\end{figure}

\newpage

{
\bibliographystyle{abbrvnat}
\bibliography{mod_bib_semi}

\begin{thebibliography}{}

\bibitem[\protect\citeauthoryear{Bennett and Wakefield}{Bennett and
  Wakefield}{2001}]{bennett+w01}
Bennett, J. and J.~Wakefield (2001).
\newblock Errors-in-variables in joint population
  pharmacokinetic/pharmacodynamic modeling.
\newblock {\em Biometrics\/}~{\em 57\/}(3), 803--812.

\bibitem[\protect\citeauthoryear{Bissiri, Holmes, and Walker}{Bissiri
  et~al.}{2016}]{bissiri2016general}
Bissiri, P.~G., C.~C. Holmes, and S.~G. Walker (2016).
\newblock A general framework for updating belief distributions.
\newblock {\em Journal of the Royal Statistical Society. Series B, Statistical
  methodology\/}~{\em 78\/}(5), 1103.

\bibitem[\protect\citeauthoryear{Blangiardo, Hansell, and
  Richardson}{Blangiardo et~al.}{2011}]{blangiardo+hr11}
Blangiardo, M., A.~Hansell, and S.~Richardson (2011).
\newblock A {B}ayesian model of time activity data to investigate health effect
  of air pollution in time series studies.
\newblock {\em Atmospheric Environment\/}~{\em 45\/}(2), 379--386.

\bibitem[\protect\citeauthoryear{Carmona and Nicholls}{Carmona and
  Nicholls}{2020}]{carmona2020semi}
Carmona, C. and G.~Nicholls (2020).
\newblock Semi-modular inference: enhanced learning in multi-modular models by
  tempering the influence of components.
\newblock In {\em International Conference on Artificial Intelligence and
  Statistics}, pp.\  4226--4235. PMLR.

\bibitem[\protect\citeauthoryear{Carmona and Nicholls}{Carmona and
  Nicholls}{2022}]{carmona+n22}
Carmona, C. and G.~Nicholls (2022).
\newblock Scalable semi-modular inference with variational meta-posteriors.
\newblock arXiv:2204.00296.

\bibitem[\protect\citeauthoryear{Carpenter, Gelman, Hoffman, Lee, Goodrich,
  Betancourt, Brubaker, Guo, Li, and Riddell}{Carpenter
  et~al.}{2017}]{carpenter_stan_2017}
Carpenter, B., A.~Gelman, M.~D. Hoffman, D.~Lee, B.~Goodrich, M.~Betancourt,
  M.~Brubaker, J.~Guo, P.~Li, and A.~Riddell (2017, January).
\newblock Stan: {A} {Probabilistic} {Programming} {Language}.
\newblock {\em Journal of Statistical Software\/}~{\em 76\/}(1), 1--32.
\newblock Number: 1.

\bibitem[\protect\citeauthoryear{Castillo}{Castillo}{2014}]{castillo2014bayesian}
Castillo, I. (2014).
\newblock On {B}ayesian supremum norm contraction rates.
\newblock {\em Annals of {S}tatistics\/}.

\bibitem[\protect\citeauthoryear{Chakraborty, Nott, Drovandi, Frazier, and
  Sisson}{Chakraborty et~al.}{2022}]{chakraborty2022modularized}
Chakraborty, A., D.~J. Nott, C.~Drovandi, D.~T. Frazier, and S.~A. Sisson
  (2022).
\newblock Modularized {B}ayesian analyses and cutting feedback in
  likelihood-free inference.
\newblock {\em Statistics and Computing\/}~(To appear).

\bibitem[\protect\citeauthoryear{Claeskens and Hjort}{Claeskens and
  Hjort}{2003}]{claeskens2003focused}
Claeskens, G. and N.~L. Hjort (2003).
\newblock The focused information criterion.
\newblock {\em Journal of the American Statistical Association\/}~{\em
  98\/}(464), 900--916.

\bibitem[\protect\citeauthoryear{Davidson}{Davidson}{1994}]{davidson1994stochastic}
Davidson, J. (1994).
\newblock {\em Stochastic limit theory: An introduction for econometricians}.
\newblock OUP Oxford.

\bibitem[\protect\citeauthoryear{Frazier and Nott}{Frazier and
  Nott}{2024}]{frazier2022cutting}
Frazier, D.~T. and D.~J. Nott (2024).
\newblock Cutting feedback and modularized analyses in generalized {B}ayesian
  inference.
\newblock {\em Bayesian Analysis\/}~{\em 1\/}(1), 1--29.

\bibitem[\protect\citeauthoryear{Gneiting and Raftery}{Gneiting and
  Raftery}{2007}]{gneiting2007strictly}
Gneiting, T. and A.~E. Raftery (2007).
\newblock Strictly proper scoring rules, prediction, and estimation.
\newblock {\em Journal of the American statistical Association\/}~{\em
  102\/}(477), 359--378.

\bibitem[\protect\citeauthoryear{Green and Strawderman}{Green and
  Strawderman}{1991}]{green1991james}
Green, E.~J. and W.~E. Strawderman (1991).
\newblock A {J}ames-{S}tein type estimator for combining unbiased and possibly
  biased estimators.
\newblock {\em Journal of the American Statistical Association\/}~{\em
  86\/}(416), 1001--1006.

\bibitem[\protect\citeauthoryear{Gr{\"u}nwald and Van~Ommen}{Gr{\"u}nwald and
  Van~Ommen}{2017}]{grunwald2017inconsistency}
Gr{\"u}nwald, P. and T.~Van~Ommen (2017).
\newblock Inconsistency of {B}ayesian inference for misspecified linear models,
  and a proposal for repairing it.
\newblock {\em Bayesian Analysis\/}~{\em 12\/}(4), 1069--1103.

\bibitem[\protect\citeauthoryear{Hansen}{Hansen}{2016}]{hansen2016efficient}
Hansen, B.~E. (2016).
\newblock Efficient shrinkage in parametric models.
\newblock {\em Journal of Econometrics\/}~{\em 190\/}(1), 115--132.

\bibitem[\protect\citeauthoryear{Hjort and Claeskens}{Hjort and
  Claeskens}{2003}]{hjort2003frequentist}
Hjort, N.~L. and G.~Claeskens (2003).
\newblock Frequentist model average estimators.
\newblock {\em Journal of the American Statistical Association\/}~{\em
  98\/}(464), 879--899.

\bibitem[\protect\citeauthoryear{Jacob, Murray, Holmes, and Robert}{Jacob
  et~al.}{2017}]{jacob2017better}
Jacob, P.~E., L.~M. Murray, C.~C. Holmes, and C.~P. Robert (2017).
\newblock Better together? {S}tatistical learning in models made of modules.
\newblock {\em arXiv preprint arXiv:1708.08719\/}.

\bibitem[\protect\citeauthoryear{Jacob, O'Leary, and Atchad{\'e}}{Jacob
  et~al.}{2020}]{jacob2020unbiased}
Jacob, P.~E., J.~O'Leary, and Y.~F. Atchad{\'e} (2020).
\newblock Unbiased {M}arkov chain {M}onte {C}arlo methods with couplings.
\newblock {\em Journal of the Royal Statistical Society: Series B (Statistical
  Methodology)\/}~{\em 82\/}(3), 543--600.

\bibitem[\protect\citeauthoryear{Judge and Mittelhammer}{Judge and
  Mittelhammer}{2004}]{judge2004semiparametric}
Judge, G.~G. and R.~C. Mittelhammer (2004).
\newblock A semiparametric basis for combining estimation problems under
  quadratic loss.
\newblock {\em Journal of the American Statistical Association\/}~{\em
  99\/}(466), 479--487.

\bibitem[\protect\citeauthoryear{Kim and White}{Kim and
  White}{2001}]{kim2001james}
Kim, T.-H. and H.~White (2001).
\newblock James-{S}tein-type estimators in large samples with application to
  the least absolute deviations estimator.
\newblock {\em Journal of the American Statistical Association\/}~{\em
  96\/}(454), 697--705.

\bibitem[\protect\citeauthoryear{Kleijn and van~der Vaart}{Kleijn and van~der
  Vaart}{2012}]{kleijn2012}
Kleijn, B. and A.~van~der Vaart (2012).
\newblock The {B}ernstein-von-{M}ises theorem under misspecification.
\newblock {\em Electron. J. Statist.\/}~{\em 6}, 354--381.

\bibitem[\protect\citeauthoryear{Lee and Lee}{Lee and
  Lee}{2018}]{lee2018optimal}
Lee, K. and J.~Lee (2018).
\newblock Optimal {B}ayesian minimax rates for unconstrained large covariance
  matrices.
\newblock {\em Bayesian {A}nalysis\/}.

\bibitem[\protect\citeauthoryear{Lehmann and Casella}{Lehmann and
  Casella}{2006}]{lehmann2006theory}
Lehmann, E.~L. and G.~Casella (2006).
\newblock {\em Theory of point estimation}.
\newblock Springer Science \& Business Media.

\bibitem[\protect\citeauthoryear{Liu, Bayarri, and Berger}{Liu
  et~al.}{2009}]{bayarri2009modularization}
Liu, F., M.~Bayarri, and J.~Berger (2009).
\newblock Modularization in {B}ayesian analysis, with emphasis on analysis of
  computer models.
\newblock {\em Bayesian Analysis\/}~{\em 4\/}(1), 119--150.

\bibitem[\protect\citeauthoryear{Liu and Goudie}{Liu and
  Goudie}{2022a}]{liu+g22}
Liu, Y. and R.~J.~B. Goudie (2022a).
\newblock A general framework for cutting feedback within modularized
  {B}ayesian inference.
\newblock arXiv:2211.03274.

\bibitem[\protect\citeauthoryear{Liu and Goudie}{Liu and
  Goudie}{2022b}]{liu+g20}
Liu, Y. and R.~J.~B. Goudie (2022b).
\newblock Stochastic approximation cut algorithm for inference in modularized
  {B}ayesian models.
\newblock {\em Statistics and Computing\/}~{\em 32\/}(7).

\bibitem[\protect\citeauthoryear{Liu and Goudie}{Liu and
  Goudie}{2023}]{liu+g21}
Liu, Y. and R.~J.~B. Goudie (2023).
\newblock Generalized geographically weighted regression model within a
  modularized {B}ayesian framework.
\newblock {\em Bayesian Analysis\/}~(To appear).

\bibitem[\protect\citeauthoryear{Lunn, Best, Spiegelhalter, Graham, and
  Neuenschwander}{Lunn et~al.}{2009}]{lunn+bsgn09}
Lunn, D., N.~Best, D.~Spiegelhalter, G.~Graham, and B.~Neuenschwander (2009).
\newblock Combining {MCMC} with `sequential' {PKPD} modelling.
\newblock {\em Journal of Pharmacokinetics and Pharmacodynamics\/}~{\em 36},
  19--38.

\bibitem[\protect\citeauthoryear{Maucort-Boulch, Franceschi, and
  Plummer}{Maucort-Boulch et~al.}{2008}]{maucort2008international}
Maucort-Boulch, D., S.~Franceschi, and M.~Plummer (2008).
\newblock International correlation between human papillomavirus prevalence and
  cervical cancer incidence.
\newblock {\em Cancer Epidemiology and Prevention Biomarkers\/}~{\em 17\/}(3),
  717--720.

\bibitem[\protect\citeauthoryear{Miller}{Miller}{2021}]{miller2021asymptotic}
Miller, J.~W. (2021).
\newblock Asymptotic normality, concentration, and coverage of generalized
  posteriors.
\newblock {\em Journal of Machine Learning Research\/}~{\em 22\/}(168), 1--53.

\bibitem[\protect\citeauthoryear{Newey}{Newey}{1985}]{newey1985maximum}
Newey, W.~K. (1985).
\newblock Maximum likelihood specification testing and conditional moment
  tests.
\newblock {\em Econometrica: Journal of the Econometric Society\/}, 1047--1070.

\bibitem[\protect\citeauthoryear{Nicholls, Lee, Wu, and Carmona}{Nicholls
  et~al.}{2022}]{nicholls+lwc22}
Nicholls, G.~K., J.~E. Lee, C.-H. Wu, and C.~U. Carmona (2022).
\newblock Valid belief updates for prequentially additive loss functions
  arising in semi-modular inference.
\newblock {\em arXiv preprint arXiv:2201.09706\/}.

\bibitem[\protect\citeauthoryear{Plummer}{Plummer}{2015}]{plummer2015cuts}
Plummer, M. (2015).
\newblock Cuts in {B}ayesian graphical models.
\newblock {\em Statistics and Computing\/}~{\em 25\/}(1), 37--43.

\bibitem[\protect\citeauthoryear{Pompe and Jacob}{Pompe and
  Jacob}{2021}]{pompe2021asymptotics}
Pompe, E. and P.~E. Jacob (2021).
\newblock Asymptotics of cut distributions and robust modular inference using
  posterior bootstrap.
\newblock {\em arXiv preprint arXiv:2110.11149\/}.

\bibitem[\protect\citeauthoryear{Rieder}{Rieder}{2012}]{rieder2012robust}
Rieder, H. (2012).
\newblock {\em Robust Asymptotic Statistics: Volume I}.
\newblock Springer Science \& Business Media.

\bibitem[\protect\citeauthoryear{Rousseau}{Rousseau}{1997}]{rousseau1997asymptotic}
Rousseau, J. (1997).
\newblock Asymptotic bayes risks for a general class of losses.
\newblock {\em Statistics \& probability letters\/}~{\em 35\/}(2), 115--121.

\bibitem[\protect\citeauthoryear{Shen and Wasserman}{Shen and
  Wasserman}{2001}]{shen2001rates}
Shen, X. and L.~Wasserman (2001).
\newblock Rates of convergence of posterior distributions.
\newblock {\em The Annals of Statistics\/}~{\em 29\/}(3), 687--714.

\bibitem[\protect\citeauthoryear{Stone}{Stone}{1961}]{stone1961opinion}
Stone, M. (1961).
\newblock The opinion pool.
\newblock {\em The Annals of Mathematical Statistics\/}, 1339--1342.

\bibitem[\protect\citeauthoryear{Styring, Charles, Fantone, Hald, McMahon,
  Meadow, Nicholls, Patel, Pitre, Smith, So{\l}tysiak, Stein, Weber, Weiss, and
  Bogaard}{Styring et~al.}{2017}]{styring17}
Styring, A., M.~Charles, F.~Fantone, M.~Hald, A.~McMahon, R.~Meadow,
  G.~Nicholls, A.~Patel, M.~Pitre, A.~Smith, A.~So{\l}tysiak, G.~Stein,
  J.~Weber, H.~Weiss, and A.~Bogaard (2017).
\newblock Isotope evidence for agricultural extensification reveals how the
  world's first cities were fed.
\newblock {\em Nature Plants\/}~{\em 3}, 17076.

\bibitem[\protect\citeauthoryear{Van~der Vaart}{Van~der
  Vaart}{2000}]{van2000asymptotic}
Van~der Vaart, A.~W. (2000).
\newblock {\em Asymptotic statistics}, Volume~3.
\newblock Cambridge University Press.

\bibitem[\protect\citeauthoryear{Yu, Nott, and Smith}{Yu
  et~al.}{2023}]{yu2021variational}
Yu, X., D.~J. Nott, and M.~S. Smith (2023).
\newblock Variational inference for cutting feedback in misspecified models.
\newblock {\em Statistical {S}cience\/}~{\em 38\/}(3), 490--509.

\end{thebibliography}
}

\end{document}